\documentclass[a4paper]{easychair}
\usepackage{graphicx}
\usepackage{amsmath}
\usepackage{amssymb}
\usepackage{multicol}
\usepackage[utf8]{inputenc}

\newtheorem{proposition}{Proposition}
\newtheorem{definition}{Definition}
\newtheorem{example}{Example}
\newtheorem{remark}{Remark}
\newtheorem{theorem}{Theorem}
\newcommand{\forgetlongv}[1]{#1}
\newcommand{\forgetshortv}[1]{}

\newcommand{\class}[1]{[#1]}

\newcommand{\N}{{\cal N}}

\newcommand{\0}{{\cal O}}

\newcommand{\ra}{\rightarrow}

\newcommand{\grs}{\mathcal R}

\title{Parallel Graph Rewriting with Overlapping Rules}

\author{Rachid Echahed\inst{1} \and {Aude Maignan\inst{2}}}
\institute{LIG - CNRS and Universit\'e de Grenoble Alpes,  France\\
  \email{rachid.echahed@imag.fr}
\and
LJK - Universit\'e de Grenoble Alpes and CNRS,   France\\
  \email{aude.maignan@imag.fr}
}
\authorrunning{R. Echahed and A. Maignan} 
\titlerunning{Parallel Graph Rewriting}
\begin{document}

\maketitle

\newcommand {\nodes} {{\cal N}}
\newcommand {\ports} {{\cal P}}
\newcommand {\pn} {{\cal P \hspace{-3pt}N}}
\newcommand {\pp} {{\cal P \hspace{-3pt}P}}
\newcommand {\att} {\lambda}
\newcommand {\lhs} {l}
\newcommand {\rhs} {r}
\newcommand {\dom} {Dom}
\newcommand {\trs} {\mathcal R}
\newcommand {\topar} {\Rightarrow}
\newcommand {\pregraph} {pregraph}
\newcommand {\merge} {\equiv}
\newcommand {\admissiblecsgrs} {conflict free environment sensitive graph
rewrite system}
\newcommand {\fp} {full parallel}
\newcommand {\pa} {parallel up to automorphism}
\newcommand{\Att}{{\cal A}}
\newcommand{\esrr}{ESRR}
\newcommand{\esrs}{ESRS}

\begin{abstract}

We tackle the problem of simultaneous transformations of networks
represented as graphs. Roughly speaking, one may distinguish two kinds
of simultaneous or parallel rewrite relations over complex structures
such as graphs: (i) those which transform disjoint subgraphs in
parallel and hence can be simulated by successive mere sequential and
local transformations and (ii) those which transform overlapping
subgraphs simultaneously. In the latter situations, parallel
transformations cannot be simulated in general by means of successive
local rewrite steps.  We investigate this last problem in the
framework of overlapping graph transformation systems.  As parallel
transformation of a graph does not produce a graph in general, we
propose first some sufficient conditions that ensure the closure of
graphs by parallel rewrite relations. Then we mainly introduce and
discuss two parallel rewrite relations over graphs.  One relation is
functional and thus deterministic, the other one is not functional for
which we propose sufficient conditions which ensure its confluence.

\end{abstract}

\section{Introduction}
\label{sect:1}

Graph structures are fundamental tools that help modeling complex
  systems. In this paper, we are interested in the evolution of such structures
  whenever the dynamics is described by means of systems of rewrite
  rules. 
\forgetlongv{
 Rewriting techniques are being investigated for different
structures such as strings \cite{BO93}, trees \cite{baader98} or
graphs \cite{handbook1}. 
}
Roughly speaking, a rewrite rule can be
 defined as a pair $ l \to r$ where the left-hand and the right-hand
sides are of the same structure. A rewrite system, consisting of a set
of rewrite rules, induces a rewrite relation ($\to$) over the
considered structures. The rewrite relation corresponds to a
sequential application of the rules, that is to say, a structure $G$
rewrites into a structure $G'$ if there exits a rule $l \to r$ such
that $l$ occurs in $G$. Then $G'$ is obtained from $G$ by replacing
$l$ by $r$.

Besides this classical rewrite relation, one may think of a parallel
rewrite relation which rewrites a structure $G$ into a structure $G'$
by firing, \emph{simultaneously}, some rules whose left-hand sides
occur in $G$. Simultaneous or parallel rewriting of a structure $G$
into $G'$ can be used as a means to speed up the computations
performed by rewrite systems and, in such a case, parallel rewriting can
be simulated by successive sequential rewrite steps. However, there are situations where
parallel rewrite steps cannot be simulated by sequential steps as in
formal grammars \cite{KleijnR80}, cellular automata (CA)
\cite{Wolfram02} or L-systems \cite{lindenmayer96}.
%, or other frameworks tailored to model
%complex natural systems \cite{P-systems}. 
This latter problem is of
interest in this paper in the case where structures are graphs.

Graph rewriting is a very active area where one may distinguish
two main stream approaches, namely (i) the algorithmic approaches
where transformations are defined by means of the actual actions one
has to perform in order to transform a graph, and (ii) the algebraic
approaches where graph transformations are defined in an abstract
level using tools borrowed from category theory such as pushouts,
pullbacks etc. \cite{handbook1}.  
In this paper, we
introduce a new class of graph rewrite systems following an algorithmic
approach where rewrite rules may overlap. That is to say, in the
process of graph transformation, it may happen that some occurrences
of left-hand sides of different rules can share parts of the graph to
be rewritten.  This overlapping of the left-hand sides, which can be
very appealing in some cases, turns out to
be a source of difficulty to define rigorously the notion of parallel
rewrite steps. In order to deal with such a difficulty we follow the
rewriting  modulo approach (see, e.g. \cite{PetersonS81}) where a
rewrite step can be composed with an equivalence relation. Another
complication comes from the fact that a graph can be reduced in parallel in
a structure which is not always a graph but rather a structure we call
\emph{pregraph}. Thus, we propose sufficient conditions
under which graphs are closed under parallel rewriting. The rewrite
systems we obtain generalize some known models of computation such as
CA, L-systems and more generally substitution systems
\cite{Wolfram02}.  
\forgetlongv{
 As a simple example illustrating this work, we may
 refer to \emph{mesh refinement} \cite{Cou} and \emph{adaptative
 mesh refinement} (AMR) which constitute a very usefull technique in
 physics \cite{Ber,Ple}, astrophysics \cite{Kle, Mai} or in biology
 \cite{Plu}.
}

\forgetlongv{The paper is organized as follows. The next section introduces the
notions of pregraphs and graphs in addition to some preliminary results
linking pregraphs to graphs. In Section~\ref{sect:3}, a class of rewrite
 systems, called \emph{environment sensitive rewrite systems} is
 introduced together with a parallel rewrite
 relation. We show that graphs are not closed under such
 rewrite relation and propose sufficient conditions under which the 
outcome of a rewrite step is always a  graph. Then, in
Section~\ref{sect:4}, we define two particular parallel rewrite relations, one
performs full parallel rewrite steps whereas the second relation uses the
possible symmetries that may occur in the rules and considers only
matches up to automorphisms of the left-hand
sides. 
Section~\ref{sect:5} illustrates our framework through some
examples. 
Concluding
remarks and related work are given in Section~\ref{sect:6}.
}

\forgetshortv{The paper is organized as follows. The next section introduces the
notions of pregraphs and graphs in addition to some preliminary definitions. In Section~\ref{sect:3}, a class of rewrite
 systems, called \emph{environment sensitive rewrite systems} is
 introduced together with a parallel rewrite
 relation. We show that graphs are not closed under such
 rewrite relation and propose sufficient conditions under which the 
outcome of a rewrite step is always a  graph. Then, in
Section~\ref{sect:4}, we define two particular parallel rewrite relations, one
performs full parallel rewrite steps whereas the second one uses the
possible symmetries that may occur in the rules and considers only
matches up to automorphisms of the left-hand
sides. 
Section~\ref{sect:5} illustrates our framework through some
examples. 
Concluding
remarks and related work are given in Section~\ref{sect:6}.
%Some missing proofs and a second example may be found in the appendix.
\emph{Due to lack of space the missing proofs are provided in
\cite{EM16}}. Two additional examples which have been implemented are given in the appendix.}

\section{Pregraphs and Graphs}
\label{sect:2}

In this section we first fix some notations and give preliminary
definitions and properties.
$2^A$ denotes the power set of A. $A \uplus B$ stands for the
disjoint union of two sets $A$ and $B$.
In the following, we introduce the notion of (attributed)
\emph{pregraphs}, which denotes a class of structures we use to define
parallel graph transformations.  Elements of a pregraph may be
attributed via a function $\att$ which assigns, to elements of a
pregraph, attributes in some sets which underly a considered
attributes' structure $\Att$. For instance $\Att$ may be a
$\Sigma$-algebra \cite{ST2012} or merely a set.

\begin{definition}[Pregraph] \ \\
A \pregraph\ $H$ is a tuple $H= (\nodes_H, \ports_H, \pn_H, \pp_H,
\Att_H, \att_H)$
such that~:
\begin{itemize}
\item $\nodes_H$ is a finite set of nodes and  $\ports_H$ is a finite set of ports,
\item $\pn_H$ is a relation $\pn_H \subseteq \ports_H \times \nodes_H$,
\item $\pp_H$ is a symmetric binary relation on ports, $ \pp_H
  \subseteq \ports_H \times \ports_H$, 
\item $\Att_H$ is a structure of attributes,
\item $\att_H$ is a function  $\att_H: \ports_H \uplus \nodes_H  \ra
  2^{{\Att_H}}$ such that $\forall x \in \nodes_H \uplus \ports_H,
  card(\att_H(x))$ is finite.
%  defined by means of two  functions 
%  $\att_{\nodes}: \nodes_H \ra 2^{{Att}_{\nodes}}$ and $ \att_{\ports}: \ports_H \ra 2^{{Att}_{\ports}}$
%where    ${Att}_{\nodes}$ (respect. ${Att}_{\ports}$) is a set of
%attributes of nodes (respect. of ports).
   
\end{itemize}
\end{definition}

An element $(p,n)$ in $\pn_H$ means that port $p$ is \emph{associated}
to node $n$. 
An element $(p_1, p_2)$ in $\pp_H$ means that port $p$ is
\emph{linked} to port $p_2$.  
In a \pregraph, a port can be associated (resp. linked) to several
nodes (resp. ports).

\begin{example}
Figure~\ref{pregraph1} shows an example of a pregraph where
\begin{figure}[t] \centering
\includegraphics[scale=0.35]{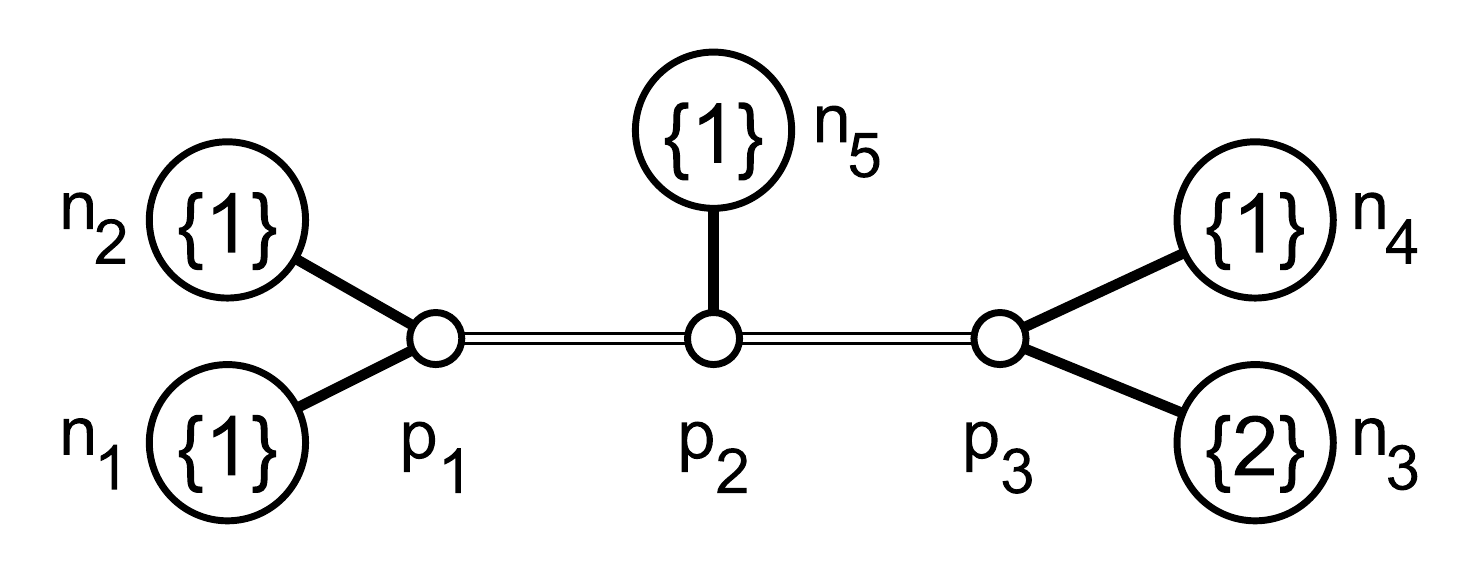} 

\caption{Example of a pregraph $H$ such that: $\Att_H = \mathbb{N}$, $\nodes_H=\{n_1,n_2,n_3,n_4,n_5\}$, $\ports_H=\{p_1,p_2,p_3\}$, 
$\pn_H=\{(p_1,n_1),(p_1,n_2),(p_2,n_5),(p_3,n_3),(p_3,n_4) \}$,
$\pp_H=\{(p_1,p_2),(p_2,p_3),(p_2,p_1),(p_3,p_2) \}$.  $\pp_H$ could
be reduced to its non symmetric port-port connection
$\{(p_1,p_2),(p_2,p_3)\}$. 
${\att_H}(n_i) = \{1\}$ for $i \in
\{1, 2, 4, 5\}$, ${\att_H}(n_3) = \{2\}$. ${\att_H}(p_j) = \emptyset$, for $j \in
\{1, 2, 3\}$. Port attributes ($\emptyset$) have not been displayed on the figure. 
}
\label{pregraph1}
\end{figure} 
the node attributes are natural numbers and Figure~\ref{pregraph1-var2} shows an example where attributes could be expressions such as $\frac{x+y}{2}$.
\begin{figure}[t] \centering
\includegraphics[scale=0.35]{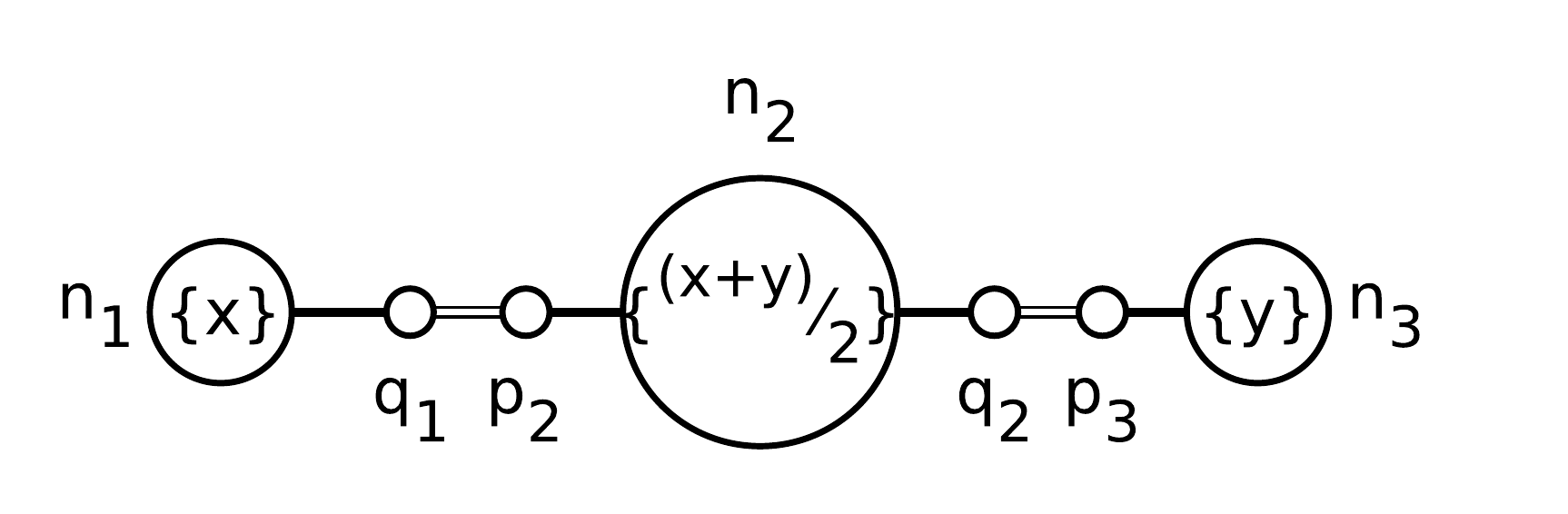} 
\caption{Example of a pregraph $H$ such that: $\Att_H =
  (\mathbb{Q}[x,y]; +, / )$, $\nodes_H=\{n_1,n_2,n_3\}$, $\ports_H=\{p_2,p_3,q_1,q_2\}$, 
$\pn_H=\{(q_1,n_1),(p_2,n_2),(q_2,n_2),(p_3,n_3) \}$,
$\pp_H$  reduced to its non symmetric port-port connection is $\pp_H=\{(p_2,q_1),(p_3,q_2) \}$.  
${\att_H}(n_1) = \{x\}$, ${\att_H}(n_2) = \{(x+y)/2\}$, ${\att_H}(n_3) = \{y \}$,  ${\att_H}(p_2) ={\att_H}(p_3) = \emptyset$, ${\att_H}(q_1) ={\att_H}(q_2) = \emptyset$. %Port attributes ($\emptyset$) have not been reported on the figure. 
}
\label{pregraph1-var2}
\end{figure} 
\forgetlongv{ In Figure~\ref{pregraph1-var},
 node attributes are variables ranging over $\mathbb{N}$. The
 introduction of variables as attributes allows one to model node neighborhood-sensitive dynamics at the rewriting rule
level as it will be illustrated in  Section~5.  

\begin{figure}[t] \centering
\includegraphics[scale=0.35]{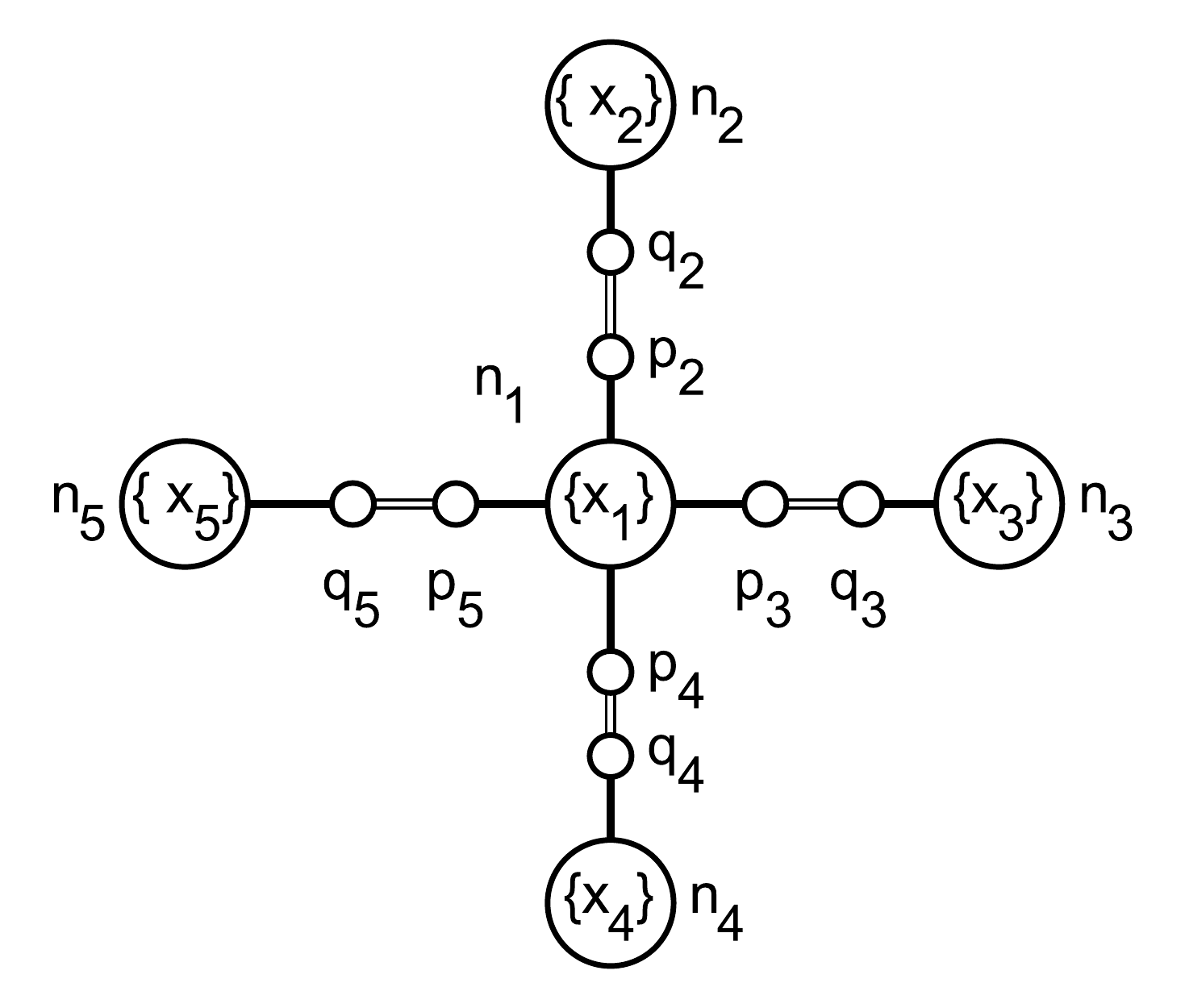}

\caption{Example of a pregraph $H$ such that: $\Att_H = \{\mathbb{N}\cup \{x_1,x_2,x_3,x_4,x_5\}, + , \times,=?= \}$, $\nodes_H=\{n_1,n_2,n_3,n_4,n_5\}$, $\ports_H=\{p_2,p_3,p_4,p_5,q_2,q_3,q_4,q_5\}$, 
$\pn_H=\{(p_2,n_1),(p_3,n_1),(p_4,n_1),(p_5,n_1),(q_2,n_2),(q_3,n_3),(q_4,n_4),(q_5,n_5) \}$,
$\pp_H$  reduced to its non symmetric port-port connection is $\pp_H=\{(p_2,q_2),(p_3,q_3),(p_4,q_4),(p_5,q_5) \}$.  
${\att_H}(n_i) = \{x_i\}$ for $i \in
\{1, 2, 3,4, 5\}$,  ${\att_H}(p_j) ={\att_H}(q_j) = \emptyset$, for $j \in
\{ 2, 3, 4,5\}$. %Port attributes ($\emptyset$) have not been reported on the figure. 
}
\label{pregraph1-var}
\end{figure} 
}
\end{example}

Below we introduce the definition of graphs used in this paper. In
order to encode classical graph edges between nodes, restrictions over
port associations are introduced.
Intuitively, an edge $e$ between two nodes $n_1$ and $n_2$ will be
encoded as two semi-edges $(n_1, p_1)$ and $(n_2, p_2)$ with
$p_1$ and $p_2$ being ports which are linked via an association $(p_1, p_2)$.

\begin{definition}[Graph] \label{def_graph}
A graph, $G$, is a \pregraph\ $G=( \nodes, \ports, \pn, \pp, \Att, \att)$
such that~:
\begin{itemize}
\item[(i)]
$\pn$ is a relation $\subseteq \ports \times \nodes$ which associates
at most one node 
to every port\footnote{The relation $\pn$ could be seen as a partial 
  function $\pn : \ports \rightarrow \nodes$ which
associates to a given port $p$, a node $n$, $\pn(p) = n$ ;  thus building a semi-edge  ``port-node''.}. That is to say, $\forall p \in \ports, \forall n_1, n_2
\in \nodes, ( (p,n_1) \in \pn \; and \;
  (p, n_2) \in \pn) \implies n_1 = n_2 $.

\item[(ii)] $\pp$ is a symmetric binary relation\footnote{The relation
  $\pp$ could also be seen as an injective (partial) function from ports to
  ports such that $\forall p \in \ports, \pp(p) \neq p$ and
$\forall p_1, p_2 \in \ports, \pp(p_1)=p_2$ iff $\pp(p_2)=p_1$.} on ports, $ \pp \subseteq \ports \times
  \ports$, such that $\forall p_1, p_2, p_3 \in \ports, ( (p_1,p_2)
  \in \pp$  and 
  $(p_1,p_3) \in \pp) \implies p_2 = p_3 $ and $\forall p \in \ports, (p,p)
  \not\in \pp$. 

% \item $\att =(\A_{\N}, \A_{\Pe})$ is a partial function $\att_{\nodes}: \nodes \ra {Att}_{\nodes}$ and $ \att_{\ports}: \ports \ra {Att}_{\ports}$
% where    ${Att}_{\nodes}$ (respect. ${Att}_{\ports}$) is a set of
% attributes of nodes $\nodes$ (respect. ports $\ports$).

\end{itemize}
\end{definition}

The main idea of our proposal is based on the use of equivalence
relations over nodes and ports (merging certain nodes and ports under
some conditions) in order to perform parallel graph rewriting in
presence of overlapping rules. Thus, to a given \pregraph\ $H$, we
associate two equivalence relations on ports, $\merge^P$, and on
nodes, $\merge^N$, as defined  below. 

\begin{definition}[$\merge^P$, $\merge^N$]
\label{def:equivalence}
Let  $H= (\nodes_H, \ports_H, \pn_H, \pp_H, \Att_H,\att_H)$ be a
\pregraph. We define two relations $\merge^P$ and
$\merge^N$ respectively on ports ($\ports_H$) and nodes ($\nodes_H$) of $H$ as follows:
\begin{itemize}
\item $\merge^P$ is defined as $(\pp_H \bullet \pp_H)^*$
\item $\merge^N$ is defined as $(\pn_H^- \bullet \merge^P \bullet \pn)^*$
\end{itemize}
where $\bullet$ denotes relation composition, $^-$ the converse of a
relation and $^*$ the reflexive-transitive closure of a relation.
We write $\class{n}$ (respectively, $\class{p}$) the equivalence class of
  node $n$ (respectively, port $p$).
\end{definition}
%%%%%%%%%%%%%%%%%%%%%%%%%%%%%%%%%%

Roughly speaking, relation $\merge^P$ is the closure of the first part of condition (ii)
in Definition~\ref{def_graph}. The base case says that if two ports
$p_1$ and $p_2$ are linked to a same port $p$, then $p_1$ and $p_2$
are considered to be equivalent. 
$\merge^N$ is almost the closure of condition (i)
in Definition~\ref{def_graph}. That is, two nodes $n_1$ and $n_2$,
which are associated to a same port (or two equivalent ports),
are considered as equivalent nodes. 

\forgetshortv{
\begin{proposition}
The relations $\merge^P$ 
and $\merge^N$ are equivalence relations.
\end{proposition}
}
\forgetlongv{
\begin{proposition}
Let  $H= (\nodes_H, \ports_H, \pn_H, \pp_H, \Att_H,\att_H)$ be a
\pregraph. The relations $\merge^P$ 
and $\merge^N$ are equivalence relations.
\end{proposition}
 \begin{proof}

 The reflexivity and transitivity of  $\merge^P$ and  $\merge^N$ follow
 directly from their respective definitions.
% %
 The symmetry of $\pp_H$ implies directly the symmetry of $\merge^P$ and $\merge^N$.

\end{proof}
}
\begin{remark}
The  relations $\merge^P$ and
$\merge^N$ 
can be computed incrementally as follows:

\noindent Base cases:
$\merge^P_0 =  \{(x,x) \mid x \in \ports_H \} $ and \;
$\merge^N_0 = \{(x,x) \mid x \in \nodes_H \} $

\noindent Inductive steps:

\noindent {\bf Rule I}: if $q,q' \in \ports_H$ such that, $q \merge^P_i
q'$, $(q,p_1) \in \pp_H$ and $(q', p_2) \in \pp_H$ 
%and $\att_H(p_1)= \att_H(p_2)$ 
then $p_1 \merge^P_{i+1} p_2 $.
\\
\noindent {\bf Rule II}: if $p_1 \in \ports_H$, $p_2 \in \ports_H$,
$(p_1,n_1) \in \pn_H$,  $(p_2,n_2) \in \pn_H$ and  $p_1 \merge^P_i p_2$ 
% and $\att_H(n_1)= \att_H(n_2)$ 
then $n_1 \merge^N_{i+1} n_2$. \\
\noindent {\bf Rule III}: If $n_1 \merge^N_i n'$ and $n'
\merge^N_i n_2$ then $n_1 \merge^N_{i+1} n_2$.

%Thanks to the inductive steps above, if $q \merge^P q'$ in H then there exists a path between $q$ and $q'$.
\end{remark}

\forgetshortv{
\begin{proposition}
\label{prop1}
The limits of the series $(\merge_i^P)_{i \geq 0}$ and
$(\merge_i^N)_{i \geq 0}$ are
respectively  $\merge^P$ and $\merge^N$.
\end{proposition}
}

\forgetlongv{
\begin{proposition}
\label{prop1}
The limit of the series $(\merge_i^P)_{i \geq 0}$ is $\merge^P$.
\end{proposition}
\begin{proof}
Since the set of ports is finite then the limit of the series is reached within a finite number of steps.

$\Rightarrow$ : Let $p_1,p_2 \in \ports_H$, such that $p_1 \merge_k^P p_2$ for some $k$,
let us  prove by induction on $k$, that $(p_1,p_2) \in (\pp \bullet \pp)^k$.
\begin{itemize}
\item case $k=0$ : $p_1 \merge_0^P p_2$ thus $p_1=p_2$ and $(p_1,p_2) \in (\pp \bullet \pp)^0 $.
\item Induction step, case $k=k'+1$ : Let us assume $p_1 \merge_{k'+1}^P p_2$. In this case, from rule I, there exist $q,q' \in \ports_H$ such that, $q \merge^P_{k'}
q'$, $(q,p_1) \in \pp_H$ and $(q', p_2) \in \pp_H$.  $q \merge^P_{k'} q' $ implies by induction hypothesis that $(q,q') \in ( \pp \bullet \pp)^{k'} $. 
Thus $(p_1,p_2) \in ( \pp \bullet \pp)^{k'+1}$.
\item Therefore for all $k$, $p_1 \merge_k^P p_2 $ implies $(p_1,p_2) \in (\pp \bullet \pp)^k$, and thus, $p_1 \merge^P p_2 $.
\end{itemize}

$\Leftarrow$
Let $p_1 \merge^P p_2$. By definition of $\merge^P$, there exists a
natural number $k$ such that $(p_1,p_2) \in (\pp \bullet \pp)^k$. It
is then straightforward that
$p_1 \merge_k^P p_2$.
\end{proof}

Likewise, we can easily show the following proposition regarding
relation $\merge^N$. 
\begin{proposition}
\label{prop2}
The limit of the series $(\merge_i^N)_{i \geq 0}$ is $\merge^N$.
\end{proposition}
\begin{proof}
Since the sets of nodes and ports are finite then the limit of the
series is reached within a finite number of steps.
$\Rightarrow$ : Let $n_1,n_2 \in \nodes_H$, such that $n_1 \merge_k^N n_2$ for some $k$ ,
let us  prove by induction on $k$, that $(n_1,n_2) \in \merge^N$.
\begin{itemize}
\item case $k=0$ : obvious.
\item Induction step, case $k=k'+1$ : Let us assume $n_1
  \merge_{k'+1}^N n_2$. We distinguish two sub-cases according to the
  used rules, i.e. Rule II or Rule III.

    \begin{itemize}
    \item[Rule III.]
According to Rule III, there exists a node $n'$ such that $n_1
\merge_{k'}^N n'$ and $n'
\merge_{k'}^N n_2$. From the induction hypothesis, we have $n_1
\merge_{k'}^N n' \implies n_1 \merge^N n'$ and  $n'\merge_{k'}^N n_2
\implies n' \merge^N n_2$. Then by transitivity of $\merge^N$ we have
$n_1 \merge^N n_2$.

    \item[Rule II.]
According to Rule II, there exist two ports $p_1$ and $p_2$ in
$\ports_H$ such that
$(p_1,n_1) \in \pn_H$,  $(p_2,n_2) \in \pn_H$ and  $p_1 \merge^P_{k'}
p_2$.
From Proposition~\ref{prop1},  $p_1 \merge^P_{k'}
p_2 \implies p_1 \merge^P p_2$, and thus, there exists an index $i$
such that $(p_1, p_2) \in (\pp_H \bullet \pp_H)^i$. Since $(p_1, n_1) \in
\pn_H$ and $(p_2, n_2) \in \pn_H$ we conclude that $(n_1, n_2) \in
(\pn_H^{-} \bullet \merge^P \bullet \pn_H)$ 
     \end{itemize}

\end{itemize}

$\Leftarrow$
Let $n_1 \merge^N n_2$. Then by definition of $\merge^N$ there exists
a natural number $k$ such that $(n_1,n_2) \in (\pn_H^{-} \bullet
\merge^P \bullet \pn_H)^k$. 
This means that there is a chain of connections consisting of tuples of
the form $(m_i, p_i).p_i \merge^P p'_i.(p'_i , m_{i+1})$ for $i \in \{0, ..., k\}$
such that $m_0 = n_1$ and $m_k = n_2$.
From rule III, it is easy to deduce the existence of $k'$, such that $n_1
\merge^N_{k'} n_2$.

\end{proof}
}
%%%%%%%%%%%%%%%%%%%%%%%%%%%%%%
%Intuitively, two ports or two nodes are equivalent if they are
%associated to a same port. 
The equivalence relations  $\merge^P$ 
and $\merge^N$ are used to introduce the notion of quotient \pregraph\
as defined below.

\begin{definition} [Quotient Pregraph]
\label{def_pregraph}
Let  $H= (\nodes_H, \ports_H, \pn_H, \pp_H, \Att_H, \att_H)$ be a \pregraph\ and
$\merge^P$ and  $\merge^N$ two equivalence relations over ports and
nodes respectively.
We write $\overline{H}$ the \pregraph\  $\overline{H}= (\nodes_{\overline{H}},
\ports_{\overline{H}}, \pn_{\overline{H}}, \pp_{\overline{H}}, \Att_{\overline{H}}, \att_{\overline{H}})$
where
$\nodes_{\overline{H}}=\{\class n \; | \; n \in \nodes_H\}$,  
$\ports_{\overline{H}}=\{\class p \; | \; p \in \ports_H\}$,  
$\pn_{\overline{H}}=\{(\class p, \class n) \; | \; (p,n) \in \pn_H \}$, 
$\pp_{\overline{H}}=\{(\class p, \class q) \; | \; (p,q) \in \pp_H \}$,
$\Att_{\overline{H}} = \Att_H $
 and
 $\att_{\overline{H}}(\class x)=\cup_{x' \in \class
   x}\att_{H}(x') $ where $ {\class x} \in \nodes_{\overline{H}} \uplus \ports_{\overline{H}}$
\end{definition}

\forgetlongv{
\begin{example}
Let  $H= (\nodes_H, \ports_H, \pn_H, \pp_H, \Att_H, \att_H)$ be a pregraph as
depicted on the left of Figure \ref{ex1}, with 
$\nodes_H=\{n_1,n_2,n_3,n_4 \}$, $\ports_H=\{p_1,p_2,p_3,p_4,p_5,p_6\}$, 
$\pn_H= \{(p_1,n_1), (p_5,n_1), $
$(p_2,n_2), (p_6,n_2),(p_4,n_3), (p_3,n_4) \}$,
$\pp_H= \{(p_1,p_2), (p_1,p_4), (p_2,p_3), (p_3,p_4),(p_5,p_6) \}$, \\
$\Att_H=\mathbb{N}$, $\att_{\nodes_{H}}(n_1)=\att_{\nodes_{H}}(n_4)=\{1\}$, $\att_{\nodes_{H}}(n_2)=\att_{\nodes_{H}}(n_3)=\{2\}$, $\forall i\in \{1,...,6\}, \att_{\ports_{H}}(p_i)= \emptyset$, 

We obtain $\overline{H}= (\nodes_{\overline{H}}, \ports_{\overline{H}}, \pn_{\overline{H}},
\pp_{\overline{H}}, \Att_{\overline H}, \att_{\overline{H}})$, as
depicted on the right of Figure \ref{ex1}, with
\begin{itemize}
\item $\nodes_{\overline{H}}=\{\class{n_1}, \class{n_2} \}$ with $ \class{n_1}= \{n_1,n_4\}$, $\class{n_2}= \{n_2,n_3\}$,
\item $\ports_{\overline{H}}=\{\class{p_1}, \class{p_2}, \class{p_5}, \class{p_6} \}$  with $ \class{p_1}= \{p_1,p_3\}$, $\class{p_2}= \{p_2,p_4\}$, $\class{p_5}= \{p_5\}$, $\class{p_6}= \{p_6\}$,
\item $\pn_{\overline{H}}=\{(\class{p_1},\class{n_1}), (\class{p_5}, \class{n_1}),(\class{p_2},\class{n_2}), (\class{p_6}, \class{n_2}) \}$,
\item $\pp_{\overline{H}}=\{(\class{p_1},\class{p_2}), (\class{p_5}, \class{p_6})\} $,
\item $\Att_{\overline H}=\Att_{ H}$
\item $\att_{\nodes_{\overline{H}}}(\class{n_1})=\{1\}$, $\att_{\nodes_{\overline{H}}}(\class{n_2})=\{2\}$, $\att_{\ports_{\overline{H}}}(\class{p_1})=\att_{\ports_{\overline{H}}}(\class{p_2})=\att_{\ports_{\overline{H}}}(\class{p_5})=\att_{\ports_{\overline{H}}}(\class{p_6})=\emptyset$.
\end{itemize}

\begin{figure}[t] \centering
\includegraphics[scale=0.32]{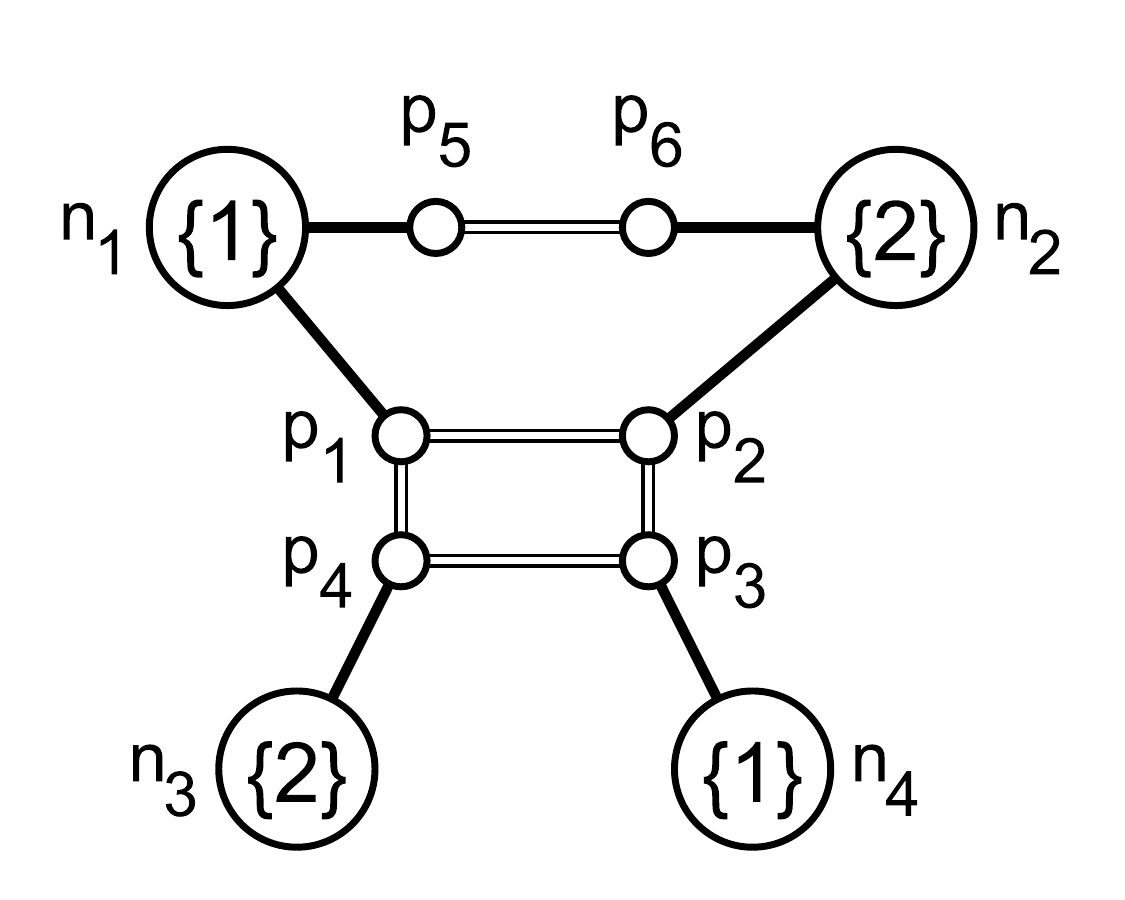}  \hspace{-0.2cm} \includegraphics[scale=0.32]{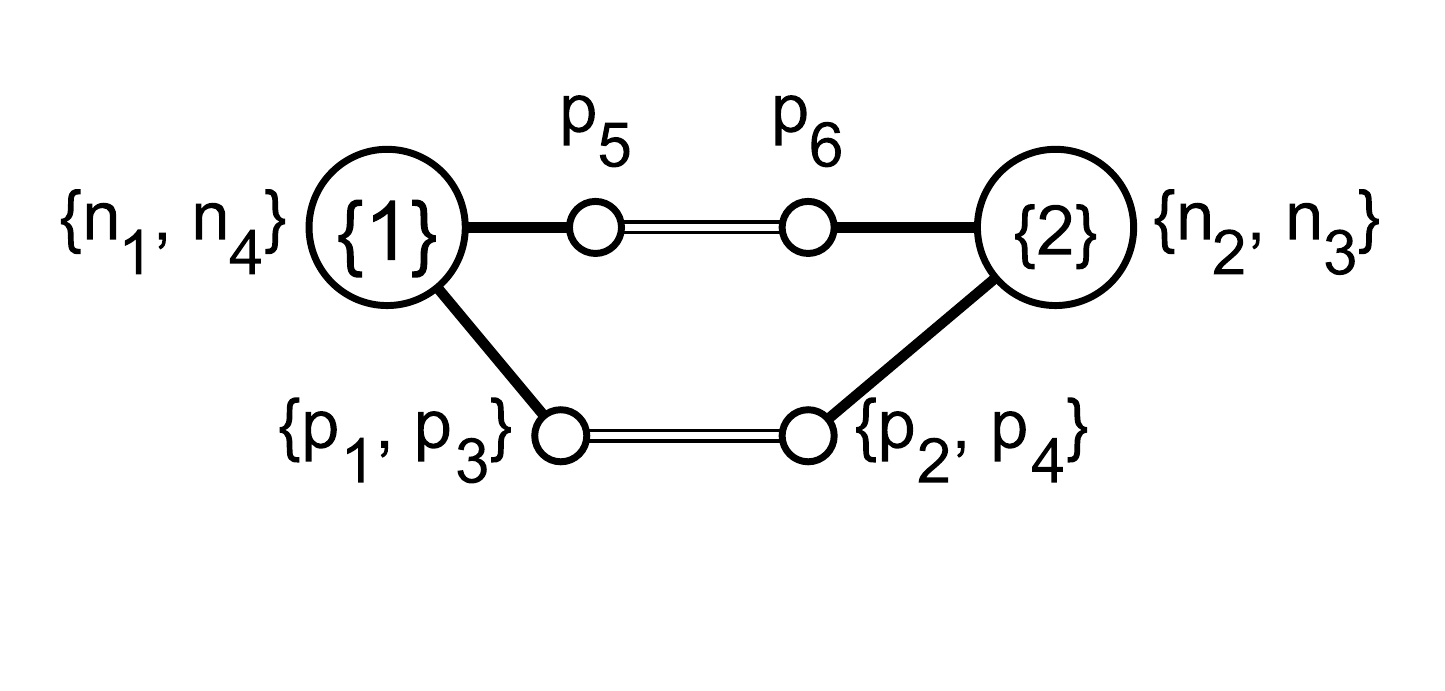}

%\vspace{-0.5cm}

\caption{(a) A pregraph $H$
 (b) and its corresponding quotient pregraph $\overline H$.}
\label{ex1}
\end{figure} 

\end{example}
}

\begin{example}
Figure~\ref{ex2} illustrates two computations of quotient pregraphs. 
\end{example}

\begin{remark} If $H$ is a graph, $\overline
  H$ and $H$ are isomorphic. Indeed,
in a graph, a port can be associated (resp. linked) to at most one
node (resp. one port).
\end{remark}

\forgetlongv{
The following definition introduces some vocabulary and notations.

\begin{definition}[Path, Loop]
\begin{itemize}
\item A path $\pi_H (p_1, p_k)$ between two  (possibly the same) nodes
  $p_1$ and $p_k$ in a pregraph $H$ is a sequence of ports of $H$
  written $\pi_H(p_1, p_k) = (p_1,p_2, \ldots, p_k)$ such that
  $\{(p_i,p_{i+1}) \; | \;
  i=1,2, \ldots,k-1\} \subseteq \pp_H$ and $k \in \mathbb{N}$ with  $k > 0$.
\item The length of a path $\pi_H(p_1, p_k) =(p_1,p_2, \ldots, p_k)$
  is $\sharp ({\pi_H}(p_1,p_k))=k-1$. 
\item  An even path (resp. odd path) is a path such that its length is even (resp. odd).
\item
A loop is a closed path, i.e., a path $\pi_H=(p_1,p_2, \ldots, p_k)$
such that $p_1=p_k$. An even loop  (resp. odd loop) is an even closed
path (resp. odd closed path). %has an even number of different ports. An odd loop has an odd number of different ports.
\end{itemize}
\end{definition}

From the definitions above, one can show the following statements.
%Remark :  if $p \merge^Pp'$, there exists a path from $p$ to $p'$.

\begin{proposition}\label{prop-even}
  Let $H= (\nodes_H, \ports_H, \pn_H, \pp_H, \Att_H, \att_H)$ be a \pregraph.
  Let $q, q'$ be two ports in $\ports_H$.  $q \merge^P q'$ iff there
  exists an even path between $q$ and $q'$ in $H$.
\end{proposition}
\forgetlongv{
 \begin{proof}
  If $q \merge^P q'$ then, by definition,
   $(q,q') \in (\pp_H \bullet \pp_H)^*$ hence there exists an even path
   between $q$ and $q'$.
% %
 Conversely, if there exists an even path between $q$ and $q'$ in $H$
 then $(q,q') \in (\pp_H \bullet \pp_H)^*$ and thus $q \merge^P q'$.
 \end{proof}
}

\begin{proposition}\label{prop:nooddloop}
  Let $H= (\nodes_H, \ports_H, \pn_H, \pp_H, \Att_H, \att_H)$ be a
  \pregraph. $\overline{H}$ is a graph iff
$H$ has no odd loop.
\end{proposition}
\forgetlongv{
 \begin{proof}
 Let $\overline{H}= (\nodes_{\overline{H}},
 \ports_{\overline{H}}, \pn_{\overline{H}}, \pp_{\overline{H}}, \Att_{\overline{H}},, \att_{\overline{H}})$.
 The relations $\pn_{\overline{H}}$ and $\pp_{\overline{H}}$ are functional by construction.
 In order to show that $\overline{H}$ is indeed a graph, 
 It remains to prove that $\pp_{\overline H}$ is not anti-reflexive iff
 there is an odd loop in $H$.

 \begin{itemize}
 %\item Assume there exists a port $q$ such that $(q, q) \in
 %  \pp_H$. Then $(\class q, \class q) \in \pp_{\overline{H}}$ and therefore
 %  $\overline{H}$ is not a graph.
 \item[$\Rightarrow$] 
 %Assume now that for all ports $q$, $(q,q) \not\in \pp_H$.
   Assume that $\pp_{\overline H}$ is not anti-reflexive. Then, there exists
   $q \in \ports_H$ such that $(\class q, \class q) \in \pp_{\overline H}$.
   Thus, either $(q, q) \in \pp_H$ which constitute an odd loop of
   length one or there exists a port $q'$, different from $q$, such
   that $(q,q') \in \pp_H$ and $q' \merge^P q$. In this last case, from
   Proposition\ref{prop-even}, $q' \merge^P q$ implies the existence of
   an even path from $q'$ to $q$. Then adding the link $(q,q')$ to this
   path builds a loop from $q$ to $q$ in $H$ of odd length.

 \item[$\Leftarrow$] Assume there is an odd loop containing a port $q$
   in $H$. Then either the loop is of the form $(q,q)$ and thus
   $(q, q) \in \pp_H$ and in this case
   $(\class q, \class q) \in \pp_{\overline H}$, or there exists a port $q'$
   different from $q$ such that the loop is of the form
   $(q, q', \ldots, q)$. In this last case, $(q,q') \in \pp_H$ and the
   path $\pi_H(q',q)$ is even. Thus, $\class q= \class{q'}$ which
   implies that $(\class q, \class q) \in \pp_{\overline H}$.
 \end{itemize}
 \end{proof}
}

}
\begin{figure}[t] \centering
\includegraphics[scale=0.28]{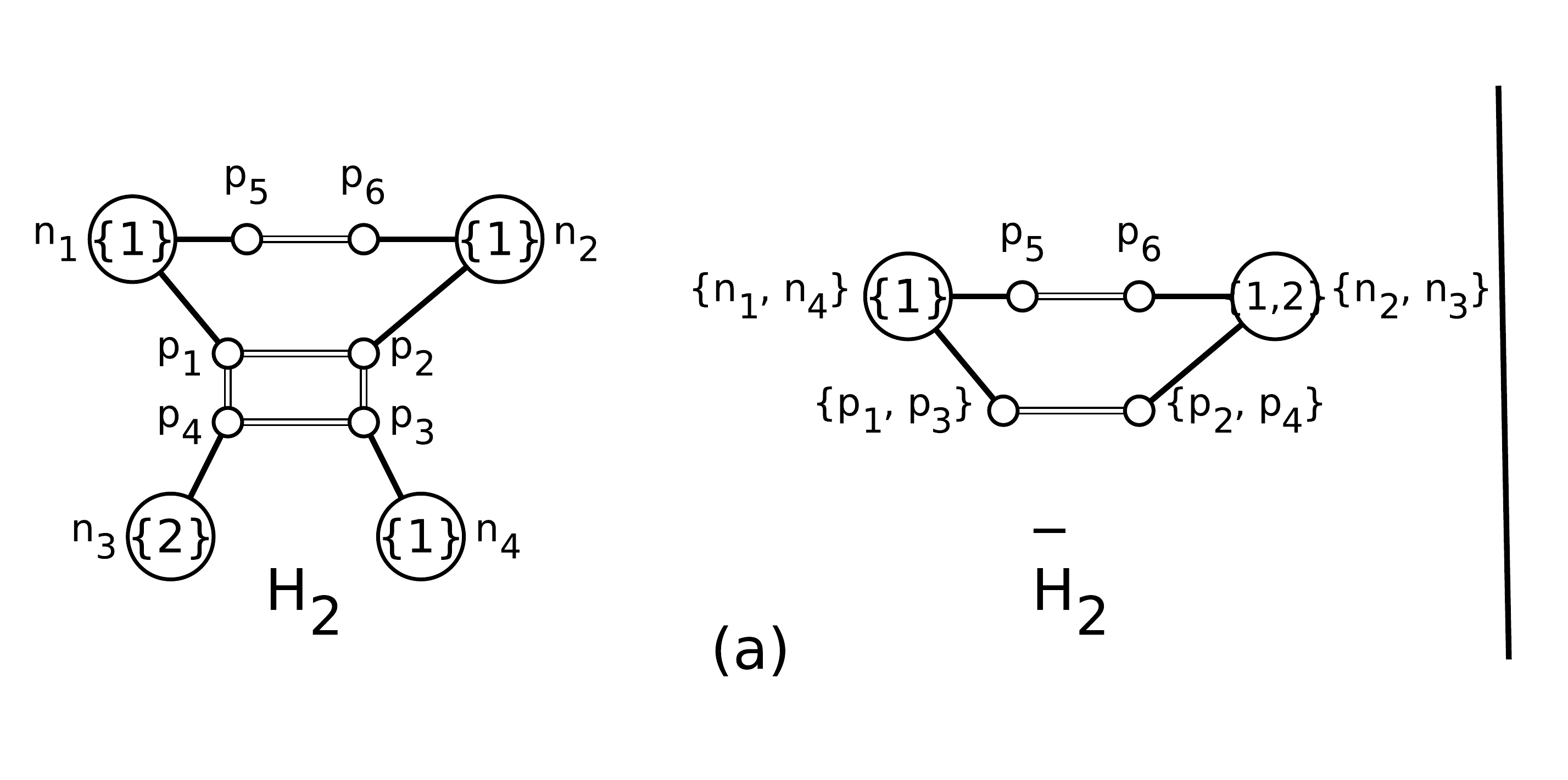} \hspace{-0.3cm}  \includegraphics[scale=0.28]{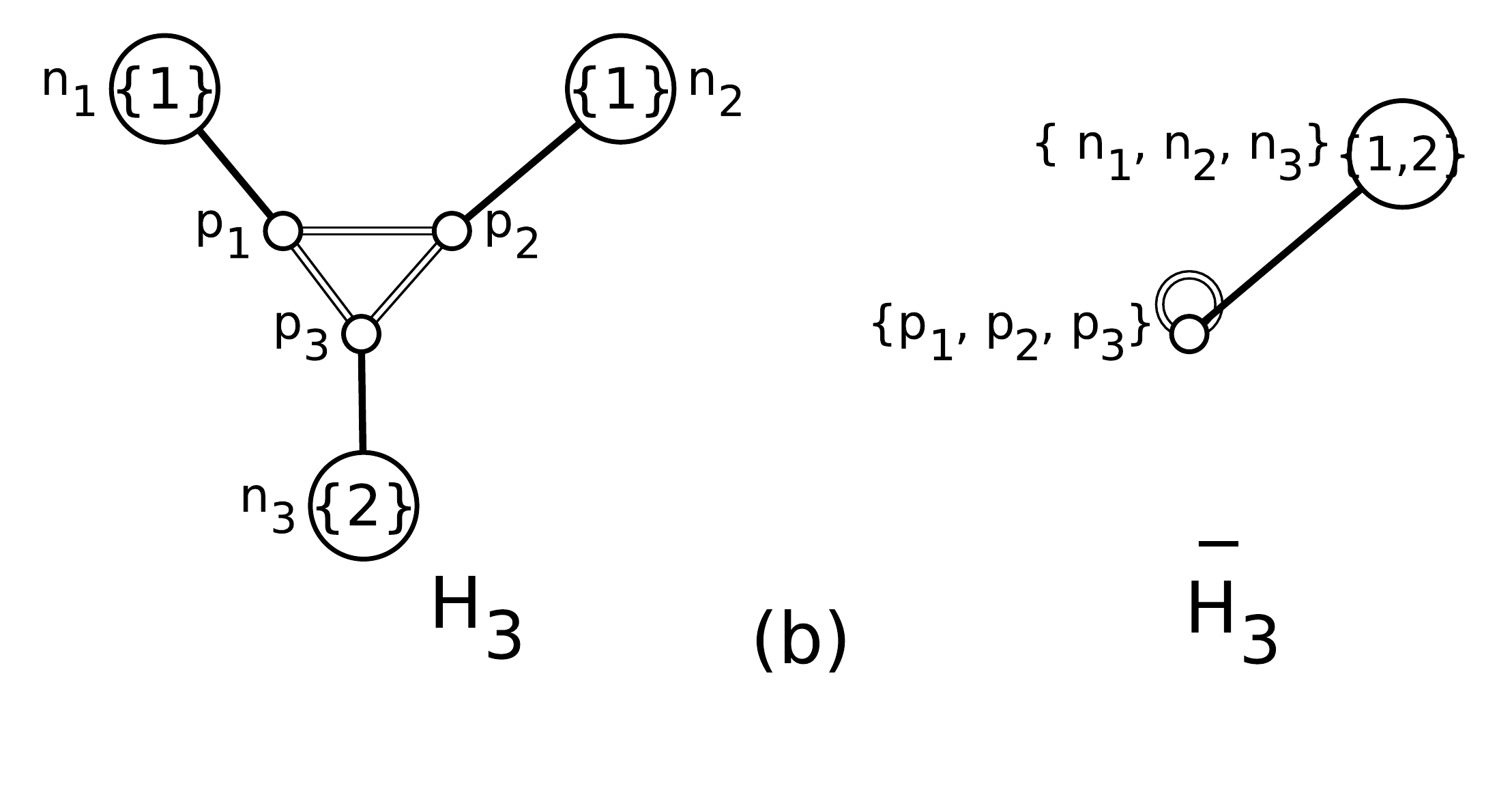}

%\vspace{-0.4cm}

\caption{(a) A pregraph $H_2$ and its corresponding quotient pregraph
  $\overline{H}_2$ which is a graph. (b) A pregraph $H_3$ and its
  corresponding quotient pregraph $\overline{H}_3$ which is not a
  graph.}
\label{ex2}
\end{figure} 

%\end{example}
Below, we define the notion of homomorphisms of
pregraphs and graphs. This notion assumes the existence of
homomorphisms over attributes \cite{DuvalEPR14}.

%\vspace{-0.5cm}

\begin{definition}[Pregraph and Graph  Homomorphism] 
  Let $l=( \nodes_l, \ports_l, \pn_l, \pp_l, \Att_l, \att_l)$ and
  $g=( \nodes_g, \ports_g, \pn_g, \pp_g, \Att_g, \att_g)$ be two pregraphs.
 Let $a : \Att_l \to \Att_g$ be a homomorphism
over attributes.
A
  \emph{pregraph homomorphism}, $ h^a : l \ra g $, between $l$ and $g$,
  built over attribute homomorphism $a$,
  is defined by two functions $h^a_N : \nodes_l \ra \nodes_g$ and
  $h^a_P : \ports_l \ra \ports_g$ such that 
(i)  $\forall (p_1, n_1) \in \pn_l$, $(h^a_P(p_1), h^a_N(n_1)) \in \pn_g, $
(ii) $\forall (p_1,p_2) \in \pp_l$,  $ (h^a_P(p_1), h^a_P(p_2)) \in \pp_g$, 
(iii)
$\forall n_1 \in \nodes_l, a(\att_l(n_1)) \subseteq
\att_{g}(h^a_N(n_1))$
and (iv)
$\forall p_1 \in \ports_l, a(\att_{l}(p_1)) \subseteq
\att_{g}(h^a_P(p_1))$.

A graph homomorphism is a pregraph homomorphism between two graphs.
%A \emph{match} is defined as an injective graph homomorphism.
\end{definition}

\noindent
Notation: Let $E$ be a set of attributes, we
denote by $a(E)$ the set $a(E) = \{a(e) \;|\; e \in E\}$.

\begin{proposition} \label{iso-pregraph} Let $H$ and $H'$ be two
  isomorphic pregraphs. Then $\overline{H}$ and $\overline{H'}$ are isomorphic.
\end{proposition}
\forgetlongv{
\begin{proof}
 Let $ h^a: H \to H'$ be a pregraph isomorphism. 
 We define $ {\overline{h^a}}: \overline{H} \to \overline{H}'$ as follows: for all ports
 $p, p'$ in $H$, nodes $n$ in H,  ${\overline{h^a_N}}(\class n) = {\class{
     h^a_N(n)}}$, ${\overline{h^a_P}}(\class p) = {\class{h^a_P(p)}}$, ${\overline{h^a}}(\class p, \class
 n) =
 {(\class{h^a_P(p)}, \class {h^a_N(n)})}$, ${\overline{h^a}}(\class p, \class {p'}) =
 {(\class{h^a_P(p)}, \class {h^a_P(p')})}$.

 $\overline{h^a}$ is clearly a pregraph
 isomorphism between $\overline{H}$ and $\overline{H'}$. 
 $\overline{h}$ is well defined as illustrated
 in the following three items.
 \begin{itemize}
 \item 
 We show that for all
 ports  $p_1 , p_2$ in $H$,
   $p_1 \merge_H^P p_2$ iff $h^a_P(p_1) \merge_{H'}^P
 h^a_P(p_2)$ :

 $h^a_P(p_1) \merge_{H'}^p h^a_P(p_2)$ iff there exists a path $\pi_{H'}(h^a_P(p_1),h^a_P(p_2))=(q_1, q_2, \ldots q_{k-1}, q_k)$ such that $q_1=h^a_P(p_1)$, $q_k=h^a_P(p_2)$ and 
 $\sharp_{\pi_{H'}}(h^a_P(p_1),h^a_P(p_2))$ is even (see, Proposition~\ref{prop-even}).
 It is equivalent to say that $(p_1, (h^a_P)^{-1}(q_2), \ldots, (h^a_P)^{-1}(q_{k-1}),p_2)$ is an even path of $H$ because $h$ is an isomorphism.
 We conclude that $p_1 \merge_H^p p_2$.

 \item For all nodes $n_1 , n_2$ in $H$, we show that $n_1 \merge_H^n n_2$ iff
 $h^a_N(n_1) \merge_{H'}^n h^a_N(n_2)$.

 By definition, $h^a_N(n_1) \merge_{H'}^n h^a_N(n_2)$ iff (a) $h^a_N(n_1)=h^a_N(n_2)$ or (b) there exists $q,q' \in \ports_{H'}$, $q \merge_{H'}^pq'$, $(q, h^a_N(n_1)) \in \pn_{H'}$ and $(q', h^a_N(n_2)) \in \pn_{H'}$
 or (c) there exists $n'' \in \nodes_{H'}$, $h^a_N(n_1) \merge_{H'}^n n''$ and $h^a_N(n_2) \merge_{H'}^n n''$.

 $h^a$ is an isomorphism thus (a) is equivalent to $n_1=n_2$ and (b) is equivalent to there exists $(h^a_P)^{-1}(q),(h^a_P)^{-1}(q') \in \ports_{H}$, $(h^a_P)^{-1}(q) \merge_{H'}^p (h^a_P)^{-1}(q')$, $((h^a_P)^{-1}(q), n_1) \in \pn_{H}$ and $((h^a_P)^{-1}(q'), n_2) \in \pn_{H}$.
The cases (a) and (b) are straight foward.
 Let us focus our attention on the case (c) :
 $h^a_N(n_1) \merge_{H'}^n h^a_N(n_2)$ such that it exists $n'' \in
 \nodes_{H'}$ and $q,q' \in \ports_{H'}$ which verify the condition $\{
 (q,n''),(q, h(n_1)),( q',n'') ,(q',h^a_N(n_2))\} \subset \pn_{H'}$. This is
 equivalent to : $(h^a_N)^{-1}(n'') \in \nodes_{H}$ and $(h^a_P)^{-1}(q),(h^a_P)^{-1}(q')
 \in \ports_{H}$ which verify the condition $$\{
 ((h^a_P)^{-1}(q),(h^a_P)^{-1}(n'')),((h^a_P)^{-1}(q), n_1),( (h^a_P)^{-1}(q'),(h^a_N)^{-1}(n''))
 ,((h^a_P)^{-1}(q'), n_2)\} \subset \pn_{H}$$ and in that case
 $ n_1 \merge^N_H n_2$. Moreover because $\merge^N$ is transitive we obtain that (c)
 is equivalent to : there exists $(h^a_N)^{-1}(n'') \in \nodes_{H}$, $n_1 \merge_{H}^n (h^a_N)^{-1} (n'')$ and $n_2 \merge_{H}^n (h^a_N)^{-1}(n'')$.
 Thus, $n_1 \merge_H^n n_2$.
\item
The pregraph homorphism of  $H \to H'$ and  $\overline H \to \overline H'$ are 
built over the same attribute homomorphism $a$, thus by construction 
the points (iii) and (iv) of the previous definition
imply 
$\cup_{n_i \in [n]} a(\att_{H}(n_i)) \subset \cup_{n_i \in [n]} \att_{H'}( h^a_N(n_i))$ and 
$\cup_{p_i \in [p]} a(\att_{ H}(p_i)) \subset \cup_{p_i \in [p]} \att_{H'}( h^a_P(p_i))$ 
 
thus
$a(\att_{\overline H}([n])) \subset \att_{\overline H'}(\overline h^a_N([n]))$ and 
$a(\att_{\overline H}([p])) \subset \att_{\overline H'}(\overline h^a_P([p]))$ 
 and $\overline{h^a}$ is a pregraph homomorphism from $\overline H$ to $\overline H'$.

\end{itemize}

 \end{proof}
}
%%\vspace{-0.2cm}
We end this section by defining an equivalence relation over
pregraphs.

%\vspace{-0.2cm}
\begin{definition}[Pregraph equivalence] \label{pregraph_equiv}
Let $G_1$ and $G_2$ be two pregraphs. We say that $G_1$ and $G_2$ are
equivalent and write $G_1 \merge G_2$ iff the quotient pregraphs
$\overline{G_1}$ and $\overline{G_2}$ are isomorphic.
\end{definition}

\forgetlongv{
The relation $\merge$ over pregraphs is obviously an equivalence relation.
}

\section{Graph Rewrite Systems}
\label{sect:3}
In this section, we define the considered  rewrite systems and provide
sufficient conditions ensuring the closure of graph structures under
the defined rewriting process.

\begin{definition}[Rewrite Rule, Rewrite System, Variant]
A rewrite rule is a pair $\lhs \ra \rhs$ where $\lhs$ and $\rhs$ are
graphs over the same sets of attributes. A rewrite system $\grs$ is a set
of rules. A \emph{variant} of a rule $\lhs \ra \rhs$ is a rule $\lhs'
\ra \rhs'$ where nodes, ports as well as the variables of the
attributes are renamed with \emph{fresh} names.
\end{definition}

Let $\lhs' \ra \rhs'$ be a variant of a rule $\lhs \ra \rhs$. Then
there is a \emph{renaming mapping} $h^a$, built over an attribute renaming
$ a: \Att_{\lhs} \to \Att_{\lhs'}$, and consisting of two maps
$h^a_N$ and $h^a_P$ over nodes and ports respectively :
$h^a_N : \nodes_{\lhs} \cup \nodes_{\rhs} \to \nodes_{\lhs'} \cup
\nodes_{\rhs'}$
 and
$h^a_P : \ports_{\lhs} \cup \ports_{\rhs} \to \ports_{\lhs'} \cup
\ports_{\rhs'} $
such that, the elements in $\nodes_{\lhs'}$ and $\ports_{\rhs'}$ are
new and the restrictions of $h^a$ to $\lhs \to \lhs'$ (respectively
$\rhs \to \rhs'$) are graph isomorphisms.

In general, parts of a left-hand side of a rule remain unchanged in
the rewriting process. This feature is taken into account in the
definition below which refines the above notion of rules by
decomposing the left-hand sides into an \emph{environmental} part,
intended to stay unchanged, and a \emph{cut} part which is
intended to be removed. As for the right-hand sides, they  are
partitioned into a \emph{new} part consisting of added items and an
environmental part (a subpart of the left-hand side) 
which is used to
specify how the new part is connected to the environment.

\begin{definition}[Environment Sensitive Rewrite Rule, Environment Sensitive Rewrite System]
An environment sensitive rewrite rule is a rewrite rule (\esrr \ for short)  $\lhs \ra \rhs$
where $\lhs$ and $\rhs$ are
graphs over the same attributes $\Att$
such that:

\vspace{-0.3cm}
\paragraph{-}$\lhs = (\nodes_\lhs, \ports_\lhs, \pn_\lhs, \pp_\lhs, \Att, 
  \att_\lhs)$
 where

\noindent
$ \nodes_{\lhs}= \nodes_{\lhs}^{cut} \uplus$ $\nodes_{\lhs}^{env},
  \ports_{\lhs}=\ports_{\lhs}^{cut} \uplus \ports_{\lhs}^{env},
  \pn_{\lhs}= \pn_{\lhs}^{cut} \uplus \pn_{\lhs}^{env},
  \pp_{\lhs} = \pp_{\lhs}^{cut} \uplus \pp_{\lhs}^{env} and  \ 
  \att_{\lhs}=\att_{\lhs}^{cut} \uplus$\footnote{Here, the function
    $\att_{\lhs}$ is considered as a set of pairs $(x,
    \att_{\lhs}(x))$, i.e. the graph of $\att_{\lhs}$.} $\att_{\lhs}^{env}$
with some additional constraints~:
     \begin{itemize}
       \item[(1)] {\bf on $\pn_{\lhs}$ :}
         %  if a node $n$ is in $\nodes_{\lhs}^{cut}$ then
         %  for all ports $p$ such that
         %  $(p,n) \in \pn_{\lhs}$ then   $(p,n) \in \pn^{cut}_{\lhs} $.
$\forall (p,n) \in \pn_\lhs, (n \in \nodes_\lhs^{cut}$ or $p \in \ports_\lhs^{cut}) \Rightarrow (p,n)
\in \pn_\lhs^{cut}$.

           %if a port $p$ is in $\ports_{\lhs}^{cut}$ then if 
           %there exists a node $n$ such that $(p, n) \in \pn_\lhs$
           %then $(p, n) \in \pn_{\lhs}^{cut}$.
%$\forall (p,n) \in \pn_\lhs, p \in \ports_\lhs^{cut} \Rightarrow (p,n)
%\in \pn_\lhs^{cut}$ 

       \item[(2)] {\bf on $\pp_{\lhs}$ :}
   %        if a port $p$ is in $\ports_{\lhs}^{cut}$ 
   %        then, for all $(p,p') \in \pp_{\lhs}$,
   %        $(p,p') \in \pp_{\lhs}^{cut}$.
$\forall (p,p') \in \pp_\lhs, p \in \ports_\lhs^{cut} \Rightarrow
(p,p') \in \pp_\lhs^{cut}$.

      \item[(3)] {\bf on $\att_{\lhs}$ :}
%  %          if a node $n$ is in $\nodes_{\lhs}^{cut}$ then $n$ is in
%  %        $\dom(\att_{\lhs}^{cut})$.
% $\forall n \in \nodes_\lhs^{cut}, n \in \dom(\att_{\lhs}^{cut})$
% % $\dom(\att_{\lhs}^{cut}) \subseteq \nodes_{\lhs}^{cut}$.
%  %          if a port $p$ is in $\ports_{\lhs}^{cut}$ then $p$ is in
%  %          $\dom(\att_{\lhs}^{cut})$.
% $\forall p \in \ports_\lhs^{cut}, p \in \dom(\att_{\lhs}^{cut})$
% % $\dom(\att_{\lhs}^{cut}) \subseteq \ports_{\lhs}^{cut}$.
 $\forall n \in \nodes_\lhs^{cut}, (n, \att_{\lhs}(n)) \in
 \att_{\lhs}^{cut}$ \; and \;
 $\forall p \in \ports_\lhs^{cut}, (p, \att_{\lhs}(p)) \in
 \att_{\lhs}^{cut}$.

      \end{itemize}
\vspace{-0.3cm}
\paragraph{-} $\rhs = (\nodes_\rhs, \ports_\rhs, \pn_\rhs, \pp_\rhs, \Att, \att_\rhs)$ where

\noindent
$\nodes_{\rhs}= \nodes_{\rhs}^{new} \uplus
  \nodes_{\rhs}^{env},\ports_{\rhs}=\ports_{\rhs}^{new} \uplus
  \ports_{\rhs}^{env},
  \pn_{\rhs}= \pn_{\rhs}^{new} \uplus
  \pn_{\rhs}^{env},\pp_{\rhs} = \pp_{\rhs}^{new} \uplus
  \pp_{\rhs}^{env},\att_{\rhs}= \att_{\rhs}^{new} \uplus
  \att_{\rhs}^{env}$
such that $\nodes_{\rhs}^{env} \subseteq \nodes_{\lhs}^{env}$,
$\ports_{\rhs}^{env} \subseteq \ports_{\lhs}^{env}$,  $\nodes_\rhs^{new} \cap \nodes_\lhs^{env} = \emptyset$ and
$\ports_\rhs^{new} \cap \ports_\lhs^{env} = \emptyset$
  with some additional constraints~:
\begin{itemize}

\item[(4)] {\bf on $\pn_{\rhs}$ :}
$\forall  (p,n) \in \pn_{\rhs}, (p, n) \in \pn_{\rhs}^{env}$ iff $( p \in \ports_{\rhs}^{env}$ and $n \in
\nodes_{\rhs}^{env}$ and $(p, n) \in \pn_{\lhs}^{env})$.

\item[(5)] {\bf on $\pp_{\rhs}$ :}
$\forall (p,p') \in \pp_{\rhs},
 (p,p') \in \pp_{\rhs}^{env}$ iff $(p \in \ports_{\rhs}^{env}$ and $p' \in \ports_{\rhs}^{env}$   and $(p, p') \in \pp_{\lhs}^{env})$.

 \item[(6)] {\bf on $\att_{\rhs}$ :}
% %such that $\att_{\rhs}(n)$ is defined,
% %$n \in \dom(\att_{\rhs}^{env}) $ iff  $n \in \dom(\att_{\lhs}^{env})$ 
$\forall n \in \nodes_{\rhs}^{env}, (\exists y, (n, y) \in \att_{\rhs}^{env})$ iff $
 (\att_{\rhs}^{env}(n) = \att_{\lhs}^{env}(n))
 $ ~;~ \\
 $\forall y \in \ports_{\rhs}^{env}, (\exists y, (p, y) \in \att_{\rhs}^{env})$ iff $
 (\att_{\rhs}^{env}(p) = \att_{\lhs}^{env}(p)).
 $

       \end{itemize}
An environment sensitive rewrite system (\esrs \ for short) is a set of environment sensitive
rewrite rules.

\end{definition}

Roughly speaking, 
 constraints (1),
(2) and (3) ensure that if an item (node or port) is to be removed
(belonging to a ``cut'' component) then links involving that item
should be removed too as well as its attributes (constraint
(3)). Constraints (4) and (5) ensure that links, considered as new
(belonging to ``new'' components), of a given right-hand side of a
rule, should not appear in the left-hand side. Constraint (6) ensures
that an item (node or port) is newly attributed in the right-hand side
iff it is a new item or it was assigned by $\att_{\lhs}^{cut}$ in the
left-hand side.

 \begin{proposition}
 Let  $\lhs \ra \rhs$ be a an \esrr\ %environment sensitive rewrite rule 
such that
 $\lhs = (\nodes_{\lhs}= \nodes_{\lhs}^{cut} \uplus
   \nodes_{\lhs}^{env},\ports_{\lhs}=\ports_{\lhs}^{cut} \uplus
   \ports_{\lhs}^{env},\pn_{\lhs}= \pn_{\lhs}^{cut} \uplus
   \pn_{\lhs}^{env},\pp_{\lhs} = \pp_{\lhs}^{cut} \uplus
   \pp_{\lhs}^{env},\Att, \att_{\lhs}=\att_{\lhs}^{cut} \uplus
   \att_{\lhs}^{env})$
 and 
 $\rhs = (\nodes_{\rhs}= \nodes_{\rhs}^{new} \uplus
   \nodes_{\rhs}^{env},\ports_{\rhs}=\ports_{\rhs}^{new} \uplus
   \ports_{\rhs}^{env},
   \pn_{\rhs}= \pn_{\rhs}^{new} \uplus
   \pn_{\rhs}^{env},\pp_{\rhs} = \pp_{\rhs}^{new} \uplus
   \pp_{\rhs}^{env},\Att, \att_{\rhs}=\att_{\rhs}^{new} \uplus
   \att_{\rhs}^{env})$.

 Then the following properties hold:
 \begin{itemize}
 \item 
 For all $  (p,n) \in \pn_{\rhs}$, $(p, n) \in  \pn_{\rhs}^{new}$ iff
   $ p \in \ports_{\rhs}^{new}$ or $n \in \nodes_{\rhs}^{new}$ or $(p
   \in \ports_{\rhs}^{env}$ and $n \in \nodes_{\rhs}^{env}$    and $
   (p, n) \not\in \pn_{\lhs}^{env})$
 
 \item
  For all $(p, p') \in \pp_{\rhs}$, $(p,p') \in \pp_{\rhs}^{new}$ iff $p\in
         \ports_{\rhs}^{new}$ or $p' \in \ports_{\rhs}^{new}$ or 
         $(p\in \ports_{\rhs}^{env}$ and $p' \in \ports_{\rhs}^{env}$
         and    $(p, p') \not\in \pp_{\lhs}^{env}(p))$ 

\item  For all $x \in \nodes_{\rhs} \cup \ports_{\rhs}$, $(x,
  \att_{\rhs} (x)) \in \att_{\rhs}^{new}$ iff $x \in
  \nodes_{\rhs}^{new} \cup \ports_{\rhs}^{new}$ or $(x, \att_{\lhs}(x)) \in \att _{\lhs}^{cut}$
 \end{itemize}

 \end{proposition}

\begin{example}
Let us consider a rule $R_T : l \rightarrow r$ which specifies a way to 
transform a triangle into four triangle graphs.
Figure~\ref{ex5} depicts the rule.  Black parts should be understood
as members of the \emph{cut} component of the left-hand side, yellow
items are in the \emph{environment} parts. The red items are \emph{new} in the
right-hand side.  More precisely, $l^{env}$ consists of
$\nodes_l^{env}= \{\alpha, \beta, \gamma \}$,
$\ports_l^{env}= \{\alpha_1,\alpha_2, \beta_1,\beta_2, \gamma_1,
\gamma_2 \}$,
$\pn_l^{env}=\{(\alpha_1,\alpha),(\alpha_2, \alpha), (\beta_1,\beta),
 (\beta_2,\beta),$ $(\gamma_1, \gamma), (\gamma_2, \gamma)\}$, and
$\pp_l^{env}=\emptyset$.  The cut component of the left-hand side
consists of three port-port connections and their corresponding
symmetric connections which will not be written
 : $\pp_{\lhs}^{cut}=\{(\alpha_2,\beta_1),(\beta_2,\gamma_1),
(\gamma_2,\alpha_1)\}$.
The environment component in the right-hand side allows to reconnect
the newly introduced items. 
$r^{env}$ consists of the ports
$\ports_r^{env}= \{\alpha_1,\alpha_2, \beta_1,\beta_2, \gamma_1,
\gamma_2 \}$.
 $r^{new}$ consists of $\nodes_r^{new}= \{U,V,W\}$, $\ports_r^{new}= \{
u_1,u_2, u_3, u_4, v_1, v_2, v_3, v_4, w_1,w_2,w_3,w_4 \}$, \\ $\pn_r^{new}= \{
(u_1,U),(u_2,U), (u_3,U),(u_4,U),(v_1,V), (v_2,V), (v_3,V),(v_4,V),
(w_1,W),(w_2,W),\\ (w_3,W), (w_4,W) \}$ and  $\pp_r^{new}= \{ (\alpha_1, w_2), (\alpha_2,u_1), (\beta_1,u_2), (\beta_2,v_1), (\gamma_1,v_2), (\gamma_2,w_1),
(u_3, w_3),\\ (u_4, v_4),   (w_4, v_3) \}$. The sets of attributes are empty in this example.

\begin{figure}[t] \centering
\includegraphics[scale=0.25]{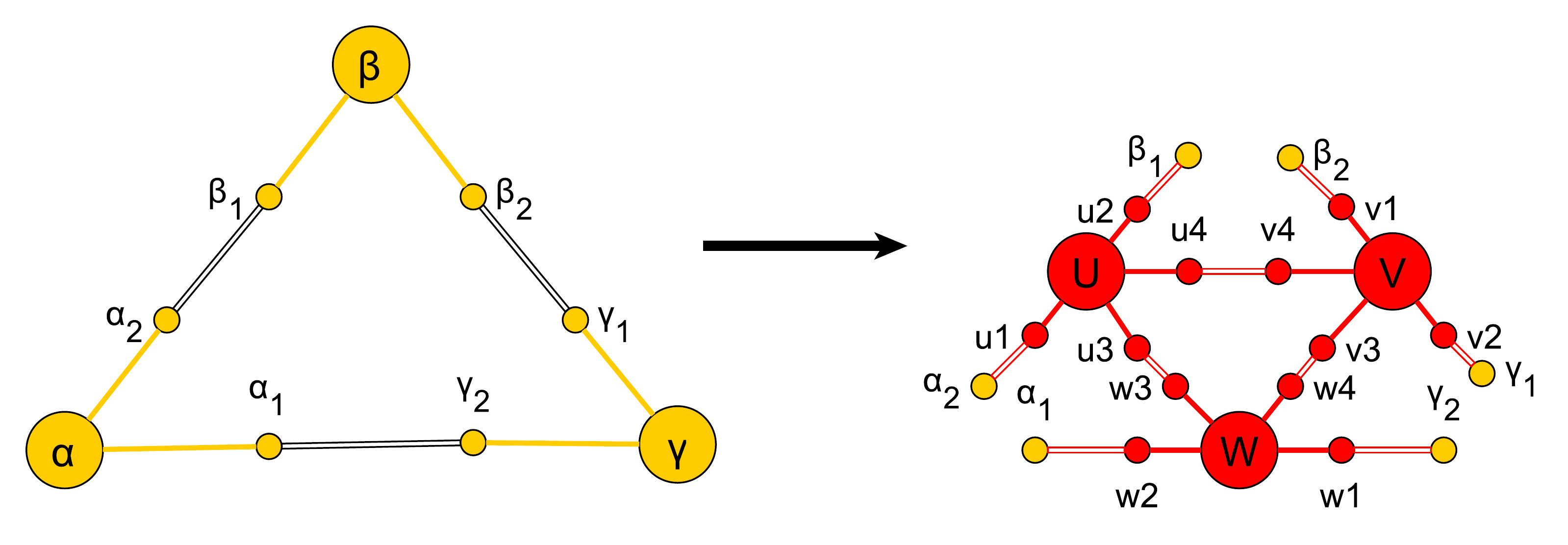}
%\vspace{-0.5cm}
\caption{ Rule $R_T$}
\label{ex5}
\end{figure}

\end{example}

\begin{remark}
From the definition of an environment sensitive rule,  the environment
components $\rhs^{env} = (\nodes_{\rhs}^{env},  \ports_{\rhs}^{env}, 
\pn_{\rhs}^{env},  \pp_{\rhs}^{env},  \Att, \att_{\rhs}^{env})$  and
$\lhs^{env} = (\nodes_{\lhs}^{env},  \ports_{\lhs}^{env},
\pn_{\lhs}^{env},  \pp_{\lhs}^{env},  \Att, \att_{\lhs}^{env})$ are graphs.
However, since $\pp_{\lhs}^{cut}$ may include  ports in
$\ports_{\lhs}^{env}$ and $\pn_{\lhs}^{cut}$ may include  nodes in
$\nodes_{\lhs}^{env}$ or ports in $\ports_{\lhs}^{env}$,
 the cut component $\lhs^{cut} = (\nodes_{\lhs}^{cut},  \ports_{\lhs}^{cut},
\pn_{\lhs}^{cut},  \pp_{\lhs}^{cut},  \Att, \att_{\lhs}^{cut})$ is in
general neither a graph  nor a pregraph.
For the same reasons $\rhs^{new} = (\nodes_{\rhs}^{new},  \ports_{\rhs}^{new},
\pn_{\rhs}^{new},  \pp_{\rhs}^{new}, \\ \Att, \att_{\rhs}^{new})$ is in
general neither a graph nor a pregraph.
\end{remark}

Finding an occurrence of a left-hand side of a rule within a graph to be
transformed  consists in finding a \emph{match}. This notion is
defined below. 

\begin{definition}[Match]
Let $l$ and $g$ be two graphs. A \emph{match} $m^{a} : l \to g$ is defined
as an injective graph homomorphism. $a :
\Att_l \to \Att_g$ being an injective homomorphism over attributes.
\end{definition}

\begin{example}

Figure~\ref{match1} gives a graph $l$ and a graph $g$.
\begin{figure}[t] \centering
\includegraphics[scale=0.25]{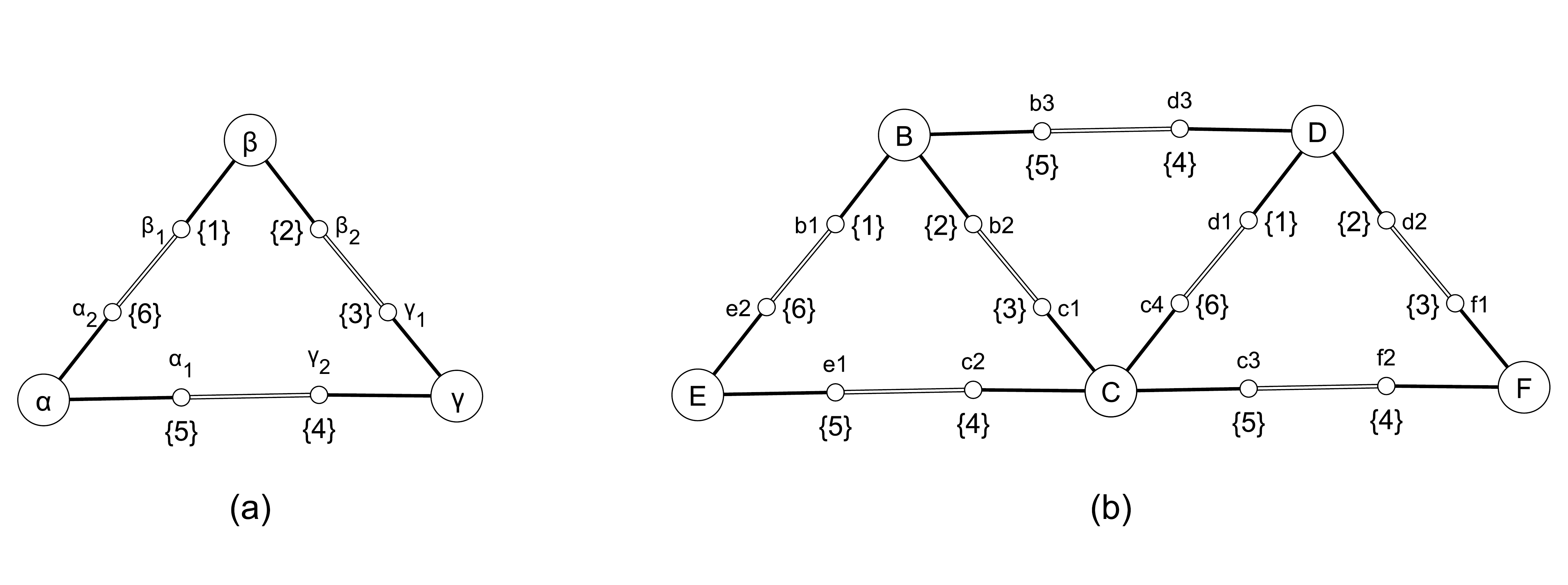} 

%\vspace{-0.6cm}

\caption{(a) A graph $l$. (b) A graph $g$.
 }
\label{match1}
\end{figure} 
Because of ports attributes, only two matches, $m_1^{id}$ and
$m_2^{id}$ can be defined from $l$ to $g$:
\begin{itemize}
\item $m_1^{id}$ : $m_1^{id}(\alpha)=E; m_1^{id}(\beta)= B$; $m_1^{id}(\gamma)=C$; $m_1^{id}(\alpha_1)=e_1; m_1^{id}(\alpha_2)=e_2; m_1^{id}(\beta_1)= b_1; m_1^{id}(\beta_2)= b_2$; $m_1^{id}(\gamma_1)=c_1$; $m_1^{id}(\gamma_2)=c_2$.
%The isomorphism of the port-node and port-port connections are easily deduced.
\item $m_2^{id}$ : $m_2^{id}(\alpha)=C; m_2^{id}(\beta)= D$; $m_2^{id}(\gamma)=F$; $m_2^{id}(\alpha_1)=c_3; m_2^{id}(\alpha_2)=c_4; m_2^{id}(\beta_1)= d_1; m_2^{id}(\beta_2)= d_2$; $m_2^{id}(\gamma_1)=f_1$; $m_2^{id}(\gamma_2)=f_2$.
\end{itemize}
Notice that the occurrences in $g$ of $m_1^{id}(l)$ and $m_2^{id}(l)$ overlap on node $C$.
\end{example}

\forgetlongv{
 
 \begin{definition}[Rewrite Step]

 Let $\lhs \ra \rhs$ be a rule,  $g$ a graph and
 $m^a : \lhs \ra g$ a match. 
%The attributes of $\lhs, \rhs$ and $g$
% are assumed to range over the same set (of attributes).
 Let $\lhs = (\nodes_{\lhs}= \nodes_{\lhs}^{cut} \uplus
   \nodes_{\lhs}^{env},\ports_{\lhs}=\ports_{\lhs}^{cut} \uplus
   \ports_{\lhs}^{env},\pn_{\lhs}= \pn_{\lhs}^{cut} \uplus
   \pn_{\lhs}^{env},\pp_{\lhs} = \pp_{\lhs}^{cut} \uplus
   \pp_{\lhs}^{env},  \Att,\att_{\lhs}=\att_{\lhs}^{cut} \uplus
   \att_{\lhs}^{env})$
 and 
 $\rhs = (\nodes_{\rhs}= \nodes_{\rhs}^{new} \uplus
   \nodes_{\rhs}^{env},\ports_{\rhs}=\ports_{\rhs}^{new} \uplus
   \ports_{\rhs}^{env},
   \pn_{\rhs}= \pn_{\rhs}^{new} \uplus
   \pn_{\rhs}^{env},\pp_{\rhs} = \pp_{\rhs}^{new} \uplus
   \pp_{\rhs}^{env}, \Att,\att_{\rhs}=\att_{\rhs}^{new} \uplus
   \att_{\rhs}^{env})$.
 A graph $g$ rewrites to $g'$ using a match $m^a$,  written $g \ra g'$ or
 $g \ra_{\lhs \ra \rhs, m^a} g'$ with $g'$ being a pregraph defined as follows:
 $g'= (\nodes_{g'}, \ports_{g'}, \pn_{g'}, \pp_{g'}, \Att_{g'},\att_{g'})$ such
 that
 \begin{itemize}
 \item $\nodes_{g'} = (\nodes_g - \nodes^{cut}_{m^a(\lhs)}) \uplus \nodes^{new}_{\rhs}$
 \item $\ports_{g'} = (\ports_g - \ports^{cut}_{m^a(\lhs)}) \uplus \ports^{new}_{\rhs}$ 
 \item $\pn_{g'} = (\pn_g - \pn_{m^a(\lhs)}^{cut}) \uplus  \pn^{new}_{m^a(\rhs)}$
 \item $\pp_{g'}  = (\pp_g -  \pp^{cut}_{m^a(\lhs)}) \uplus \pp^{new}_{m^a(\rhs)} $
 \item $\Att_{g'} = \Att_g $ and $\att_{g'} = (\att_g -  \att^{cut}_{m^a(\lhs)}) \uplus \att^{new}_{m^a(\rhs)}$

 \end{itemize}

 Notation: Let $p, p'$ be ports and $n$ be a node, in notation $m^a(r)$
 above, $m^a(p,p') = (m^a(p),m^a(p'))$, $ m^a(p, n) = (m^a(p), m^a(n))$, $m^a(p) = p$ if $p \in
 \ports_r^{new}$ and $m^a(n) = n$ if $n \in \nodes_r^{new}$.  
 \end{definition}
  \newcommand{\good}{well behaved}

  It is easy to see that graphs are not closed under the rewrite
  relation defined above. That is to say, when a graph $g$ rewrites
  into $g'$, $g'$ is a pregraph. To ensure that $g'$ is a graph we
  provide the following conditions.

 \begin{theorem}
  Let $\lhs \ra \rhs$ be an environment sensitive rewrite rule,  $g$ a graph and
 $m^a : \lhs \ra g$ a match. 
 %The attributes of $\lhs, \rhs$ and $g$
 %are assumed to range over the same set (of attributes).
 Let $g \ra_{\lhs \ra \rhs,m^a} g'$.
 $g'$ is a graph iff the two following constraints are verified :
 \begin{enumerate}
 \item
 If $p \in \ports_{\lhs}^{env} $, $(p,q) \in \pp_{\rhs}^{new}$ for some
 port $q$ and there is no $q'$ such that $(p,q') \in \pp_{\lhs}^{cut}$,
 then
 there is no $q''\in \ports_g$ such that $(m^a(p),q'') \in \pp_g$. 
 \item
 If $p \in \ports_{\lhs}^{env} $ ,  $(p,n) \in \pn_{\rhs}^{new}$ and there is no $n'$ such that  $(p,n') \in \pn_{\lhs}^{cut}$, then
 there is no $n''\in \nodes_g$ such that $(m^a(p),n'') \in \pn_g$. 

 \end{enumerate}

 \end{theorem}
 \forgetlongv{
 \begin{proof}
 $(\Leftarrow)$ Let $p$ be a port of $g'$.
 If the constraints  1. and 2. are verified then 
 \begin{itemize}
 \item If $p \in g'-m^a(r)$, $p$ has the same connections as in $g$. Since $g$ is a graph, $p$ is connected to at most one port and one node.
 \item If $p \in m^a(r^{env})$, thanks to constraints 1. and 2. $p$  has
   at most one connection to a node and one connection to a port in $g'$.
 \item If $p \in \ports_{r}^{new}$. Since $r$ is a graph,  $p$  has
   at most one connection to a node and one connection to a port in $g'$.
 \end{itemize}
 Thus, $g'$ is a graph.

 $(\Rightarrow)$
 It is easy to show, by contrapositive, that in case one of the constraints (1 and 2) is not verified, a
 counter example can be exhibited.

 \end{proof}
% %}%end forget
}
 Matches which fulfill the above two conditions are called \emph{\good}\ matches.

 \begin{example}
   Figure~\ref{badmatchrule}~(a) gives an example of toy rule.
   Figure~\ref{badmatchrule}~(b) is a graph $H$ such that
   the match $m_1^{id}$ as defined below is a \good\ match, whereas the match
   $m_2^{id}$ is not a \good\ match.
   $m_1^{id} : m_1^{id}(\alpha)=A, m_1^{id}(\beta)=B, m_1^{id}(\alpha_1)=a_1,
   m_1^{id}(\beta_1)=b_1 $
   and
   $m_2^{id}: m_2^{id}(\alpha)=C, m_2^{id}(\beta)=B, m_2^{id}(\alpha_1)=c_1,
   m_2^{id}(\beta_1)=b_1$.
% %
 The application of the toy rule on nodes $B$ and $C$ and the ports
 $b_1$ and $c_1$ (according to match $m_2^{id}$) leads to  a pregraph  which is not a graph.
 \begin{figure}[t] \centering
 \includegraphics[scale=0.2]{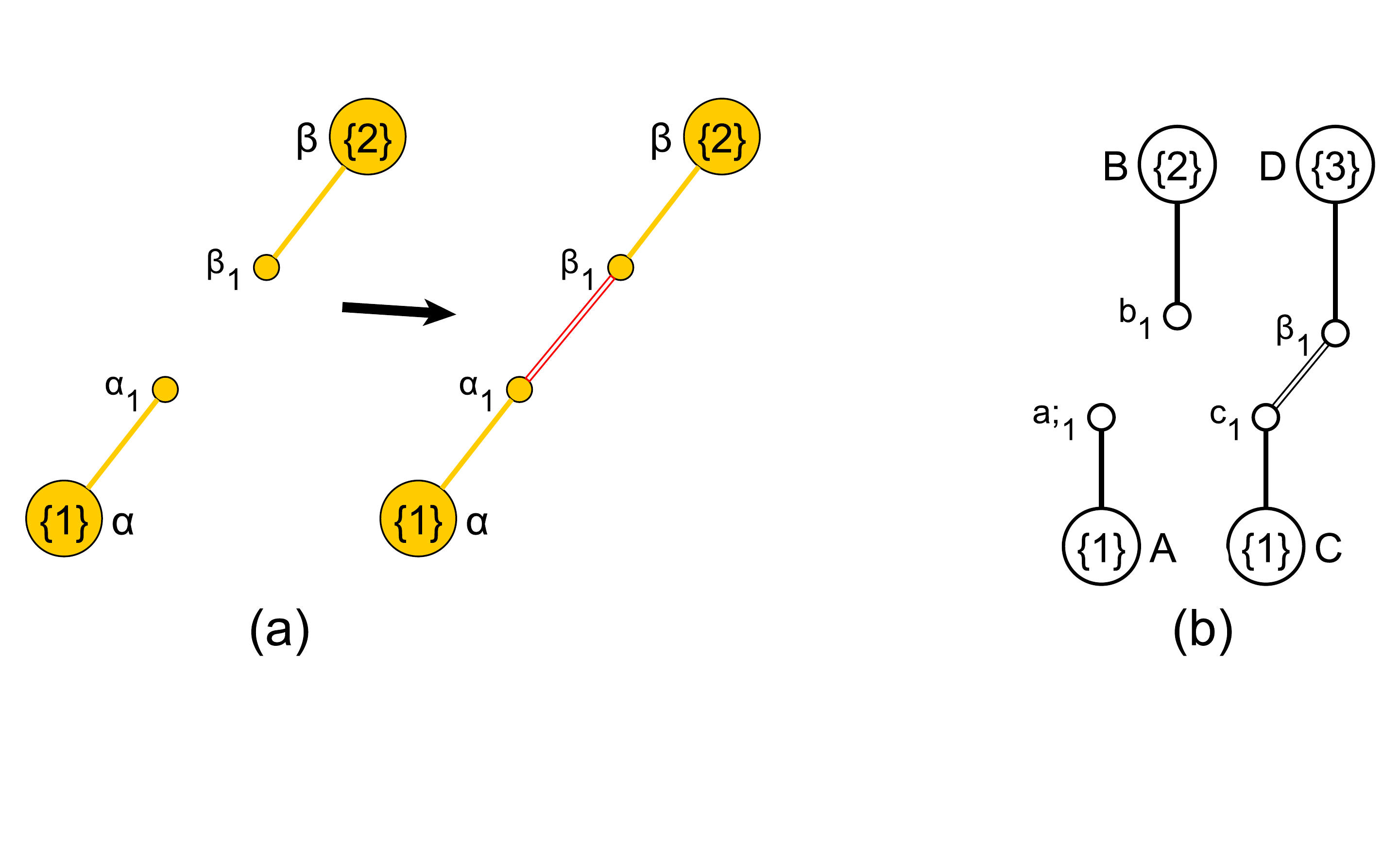}

% \vspace*{-1cm}

 \caption{Example of a toy rule (a) and a graph $H$ (b)}
 \label{badmatchrule}
%\vspace{-0.5cm}
 \end{figure} 

 \end{example}

}
In order to define the notion of parallel rewrite step, we have to
restrict a bit the class of the considered rewrite systems. Indeed, 
let $l_1 \ra r_1$ and $l_2 \ra r_2$ be two \esrr.
Applying these two rules in parallel on a graph $g$ is possible
only if there is ``no conflict'' while firing the two rules
simultaneously. A conflict may occur if some element of the
environment of $r_1^{env}$ is part of $l_2^{cut}$ and vice versa.
To ensure conflict free rewriting, we introduce the notion of
%\emph{\admissiblecsgrs}.
\emph{conflict free} \esrs.
Let us first define the notion of compatible rules.
\begin{definition} [compatible rules]
%Two environment sensitive rewrite rules  
Two \esrr's
$l_1 \ra r_1$ and $l_2 \ra r_2$ are said to be compatible
iff for all graphs $g$ and matches $m_1^{a_1} : l_1 \ra g$ and $m_2^{a_2} : l_2
\ra g$, (i) no element of $m_1^{a_1}(r_1^{env})$ is in $m_2^{a_2}(l_2^{cut})$ and
(ii)  no element of $m_2^{a_2}(r_2^{env})$ is in $m_1^{a_1}(l_1^{cut})$.
% and 
% (iii) no conflict on attributes 
%(i.e., $G - (m_1(l_1^{cut}) \cup
%m_2^{a_2}(l_2^{cut})) \cup (m_2^{a_2}(r_2) \cup m_1(r_1))/_{\merge}$ is a graph).

\end{definition}  

Conditions (i) and (ii) ensure that the constructions defined by
 $m_1^{a_1}(r_1)$ (respectively by $m_2^{a_2}(r_2)$) can actually be performed ;
 i.e, no element used in   $m_1^{a_1}(r_1)$ (respectively by $m_2^{a_2}(r_2)$) is
missing because of its inclusion in $m_2^{a_2}(l_2^{cut})$ (respectively in
$m_1^{a_1}(l_1^{cut})$). 
% Condition (iii) says that the attributes of ports
% and nodes are compatible and do not prevent merging nodes or ports
% when needed.
%\begin{example}
%\begin{figure}
For instance, the reader can easily verify that two variants of the rule
\begin{center}
\includegraphics[scale=0.3]{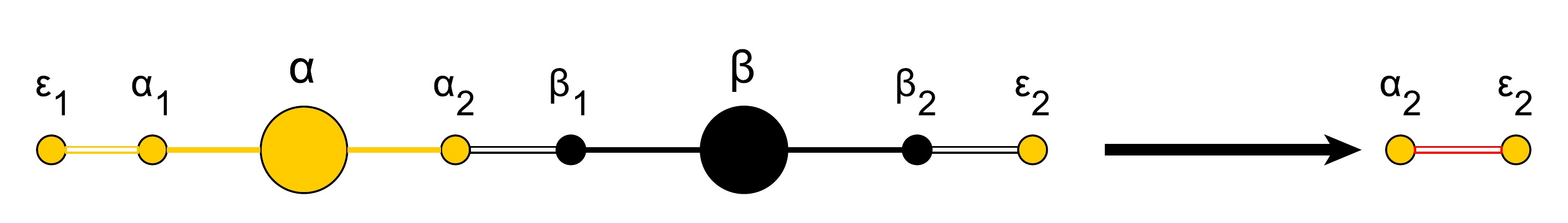}
\end{center}
are not compatible.
%\includegraphics[scale=0.35]{bool_initial.pdf}
%%\vspace{-0.5cm}
%\caption{ (a) a rule $l \to r$ such that two variants are not compatible  on (b) the graph $g$ }
%\label{bool}
%\end{figure} 
%\end{example}
Verifying that two given rules are compatible is decidable and can be
checked on a finite number, less than $max(size(l_1),
size(l_2))$, of graphs where the $size$ of a graph stands for its number of
 nodes and ports. 
% Notice that this result does not require the decidability of the unification
% problem of the considered attributes \cite{BaaderS94} because
% homorphism conditions over attributes require the inclusion of sets.

\begin{proposition}
\label{prop:compatibility}
 The problem of the verification 
of compatibility of two rules is decidable.
\end{proposition}
\forgetlongv{
 \begin{proof}
 Let $\rho_1= l_1 \ra r_1$ and $\rho_2 = l_2 \ra r_2$ be two rules.
 Assume that $\rho_1$ and $\rho_2$ are not compatible. Then there
 exists a graph $G$ such that:
 \begin{itemize}
 \item there exists a match $m_1^{a_1} : l_1 \to G$
 \item there exists a match $m_2^{a_2} : l_2 \to G$
 \item w.l.o.g, we assume that there exists an element, say $e$, in
 $m_1^{a_1}(r_1^{env})$ which belongs also to $m_2^{a_2}(l_2^{cut})$. 
 \end{itemize}
 Graph $G$ can be built as follows:
 Let $d$ be a graph such that there 
 exist two injective homomorphisms $h_1: d \ra m_1^{a_1}(l_1)$ and
 $h_2: d \ra m_2^{a_2}(l_2)$ such that $G$ is obtained as a pushout of $h_1$
 and $h_2$. That is to say, there exist two injective homomorphisms $h'_1 :
 m_1^{a_1}(l_1) \to G$ and $h'_2 : m_2^{a_2}(l_2) \to G$ such that  
  $h'_1(h_1(d)) = h'_2(h_2(d))$.  We consider subgraphs $d$ which
  contain at least $h_1^{-1}((m_1^{a_1})^{-1}(e))$ which is equal to
  $h_2^{-1}((m_2^{a_2})^{-1}(e))$. Notice that elements of graph $d$
  could be attributed by empty sets.

 Therefore, to check whether two rules  $\rho_1= l_1 \ra r_1$ and
 $\rho_2 = l_2 \ra r_2$  are compatible, one has to check whether there
 exist a subgraph $d$ and two injective homomorphisms
 $h: d \ra l_1$ and $h': d \ra l_2$ such that $d$ contains an item, $e$,
 such that $h(e) \in l_1^{cut}$ and $h'(e) \in r_2^{env}$ ($h(e) \in
 r_l^{env}$ and $h'(e) \in l_2^{cut}$) . Since homomorphisms $h$ and
 $h'$ are injective, the size (number of
 nodes and ports) of $d$ is less than $max(size(l_1), size(l_2))$.
 Obviously, $d$, $h$ and $h'$ exist iff the two rules are not
 compatible. 
 Indeed the graph $G'$ obtained as a pushout of homomorphisms $h$ and
 $h'$ contains at least one item which can be matched either by
 $l_i^{env}$ (and remains in $r_i^{env}$) and $l_j^{cut}$ with $(i,j)
 \in \{(1,2), (2,1)\}$. 

 Since the set of possible $d$'s is finite (up to isomorphism), 
 verifying whether two rules are compatible is decidable.
 \end{proof}
}

\begin{definition}
A \admissiblecsgrs\ is an \esrs\
consisting of pairwise compatible rules.
\end{definition}

\begin{definition}[parallel rewrite step] \label{defparallel}
Let $\trs$ be a \admissiblecsgrs\ 
%an environment sensitive rewrite system 
$\trs = \{L_i \ra
  R_i \mid i= 1 \ldots n\}$. Let $G$ be a graph. Let  $I$ be a set of
variants of rules in $\trs$, $I = \{\lhs_i \ra \rhs_i \mid i = 1
\ldots k\}$ and $M$ a set of matches  $M=\{m_i^{a_i} : \lhs_i \ra G \mid i = 1
\dots k\}$. We say that
graph $G$ rewrites into a pregraph $G'$ using the rules in $I$ and matches in
$M$, written $ G \topar_{I,M} G'$, $ G \topar_{M} G'$ or simply $ G
\topar G'$ if  $G'$ is obtained following the two steps below:

\noindent
{\bf First step:} A \pregraph\ $H = (\nodes_H, \ports_H, \pn_H, \pp_H, \Att_H,
\att_H)$ is computed using the different matches and rules as follows:

\forgetlongv{
\begin{itemize}
\item $\nodes_H = (\nodes_G - \cup_{i=1}^k \nodes_{m_i^{a_i}(\lhs_i)}^{cut}) \uplus
  \cup_{i=1}^k \nodes_{\rhs_i}^{new}$
\item $\ports_H = (\ports_G - \cup_{i=1}^k \ports_{m_i^{a_i}(\lhs_i)}^{cut}) \uplus
  \cup_{i=1}^k \ports_{\rhs_i}^{new}$
\item %$\pn_H = \cup_{i=0}^k PN_i$
$\pn_H = (\pn_G - \cup_{i=1}^k \pn_{m_i^{a_i}(\lhs_i)}^{cut}) \uplus \cup_{i=1}^k \pn_{m_i^{a_i}(\rhs_i)}^{new}$
 
\item %$\pp_H  = \cup_{i=0}^k PP_i$
$\pp_H  = (\pp_G - \cup_{i=1}^k \pp_{m_i^{a_i}(\lhs_i)}^{cut})\uplus \cup_{i=1}^k
\pp_{m_i^{a_i}(\rhs_i)}^{new} $

\item $\Att_H = \Att_G$ and $\att_H = (\att_G - \cup_{i=1}^k \att_{m_i^{a_i}(\lhs_i)}^{cut}) \cup
  \cup_{i=1}^n \att_{m_i^{a_i}(\rhs_i)}^{new}$
\end{itemize}
}
\forgetshortv{
$\bullet$ $\nodes_H = (\nodes_G - \cup_{i=1}^k \nodes_{m_i^{a_i}(\lhs_i)}^{cut}) \uplus
  \cup_{i=1}^k \nodes_{\rhs_i}^{new}$ 

\;\; $\bullet$ $\ports_H = (\ports_G - \cup_{i=1}^k \ports_{m_i^{a_i}(\lhs_i)}^{cut}) \uplus
  \cup_{i=1}^k \ports_{\rhs_i}^{new}$

$\bullet$ 
$\pn_H = (\pn_G - \cup_{i=1}^k \pn_{m_i^{a_i}(\lhs_i)}^{cut}) \uplus \cup_{i=1}^k \pn_{m_i^{a_i}(\rhs_i)}^{new}$

$\bullet$ %$\pp_H  = \cup_{i=0}^k PP_i$
$\pp_H  = (\pp_G - \cup_{i=1}^k \pp_{m_i^{a_i}(\lhs_i)}^{cut})\uplus \cup_{i=1}^k
\pp_{m_i^{a_i}(\rhs_i)}^{new} $

$\bullet$ \;\; $\Att_H = \Att_G$ and $\att_H = (\att_G - \cup_{i=1}^k \att_{m_i^{a_i}(\lhs_i)}^{cut}) \cup
  \cup_{i=1}^n \att_{m_i^{a_i}(\rhs_i)}^{new}$
}

%\vspace{-0.3cm}
\noindent
{\bf second step}:  $G' = \overline H$

\end{definition}

\forgetshortv{
 Notation: Let $p, p'$ be ports and $n$ a node, in notation $m^a(r)$
 above, $m^a(p,p') = (m^a(p),m^a(p'))$, $ m^a(p, n) = (m^a(p), m^a(n))$, $m^a(p) = p$ if $p \in
 \ports_r^{new}$ and $m^a(n) = n$ if $n \in \nodes_r^{new}$.  }

\forgetlongv{Notice that the rewrite step $G \topar G'$ is a rewrite
modulo step \cite{PetersonS81} of the form $G \to H \equiv \overline
H$.}

%%%%%%%%%%%%%%%%%%%%%%%%%%%%%%%
\begin{example}
Let us consider the graph $g$ depicted below and the following two
matches, $m_1$ and $m_2$,
of the rule $R_T$ depicted in  Figure~\ref{ex5}.
%\hspace*{0.5cm} \includegraphics[scale=0.2]{./Fig/appli_initial1.pdf}
\vspace*{-0.4cm}
\begin{multicols}{2}
%\column{2.5in}
\includegraphics[scale=0.2]{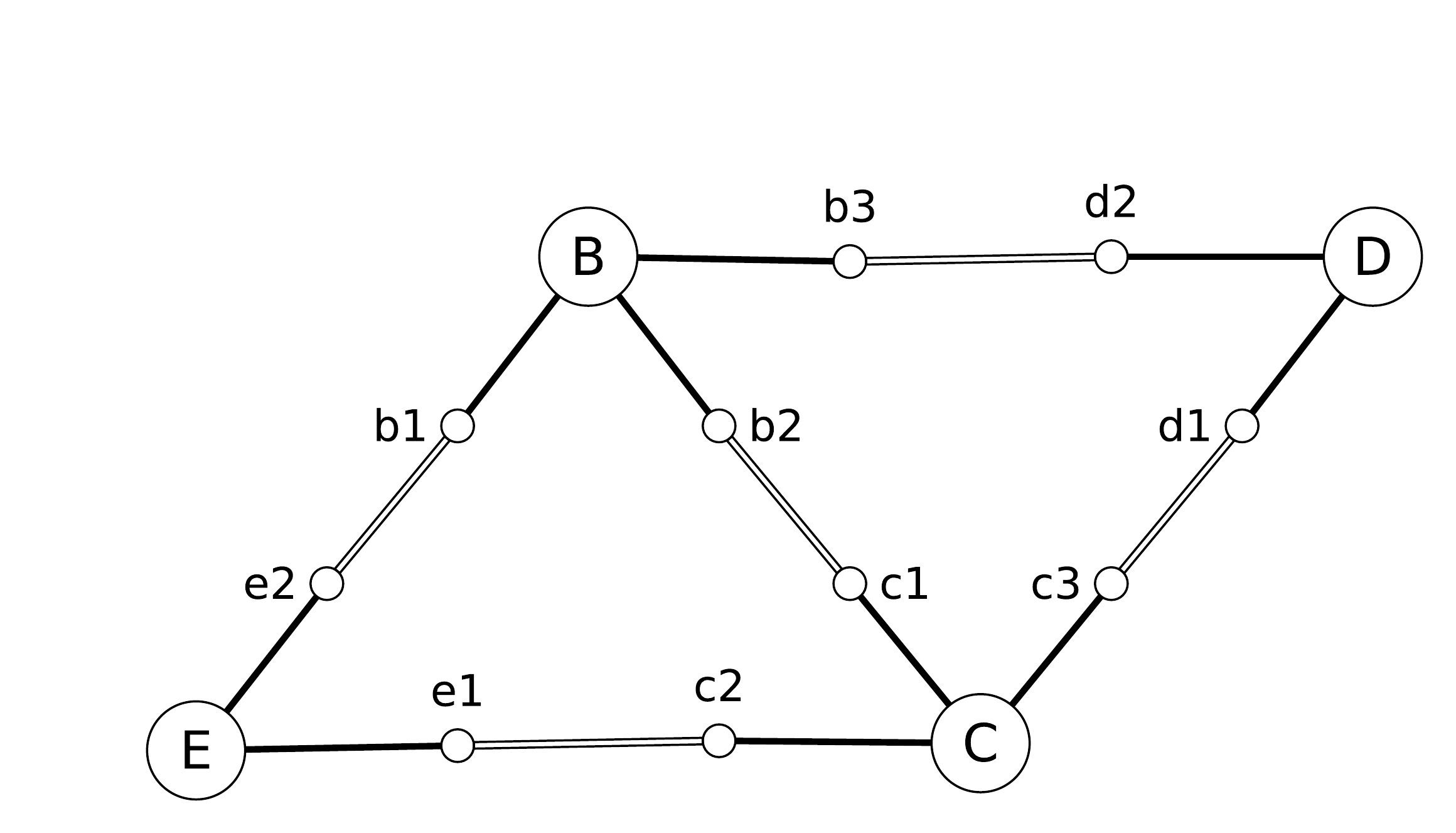}
%\column{1.5in}
\begin{itemize}
\item $m_1$ : $m_1(\alpha)=E; m_1(\beta)= B$; $m_1(\gamma)=C$; $m_1(\alpha_1)=e_1; m_1(\alpha_2)=e_2; m_1(\beta_1)= b_1; m_1(\beta_2)= b_2$; $m_1(\gamma_1)=c_1$; $m_1(\gamma_2)=c_2$.
The isomorphism of the port-node and port-port connections are easily deduced.
\item $m_2$ : $m_2(\alpha)=B; m_2(\beta)= D$; $m_2(\gamma)=C$; $m_2(\alpha_1)=b_2; m_2(\alpha_2)=b_3; m_2(\beta_1)= d_2; m_2(\beta_2)= d_1$; $m_2(\gamma_1)=c_3$; $m_2(\gamma_2)=c_1$.
\end{itemize}
The two matches overlap.
\end{multicols}

%\vspace{-0.4cm}

%\end{multicols}

Figure~\ref{steps} shows the different steps of the application of two
matches of the rule defined in Figure~\ref{ex5}. The pregraph, $H$, in the
middle is obtained after the first step of Definition~\ref{defparallel}.
Its quotient pregraph, $G'$, is the graph on the right. $G'$ has been
obtained by merging the nodes $S$ and $Y$ and the ports $s_1$ and
$y_1$ as well as ports $s_2$ and $y_2$. These mergings are depicted by
the quotient sets $[S], [s_1]$ and $[s_2]$. For sake of readability,
the brackets have been omitted for quotient sets reduced to one
element.
 
\begin{figure}[t] \centering
\includegraphics[scale=0.22]{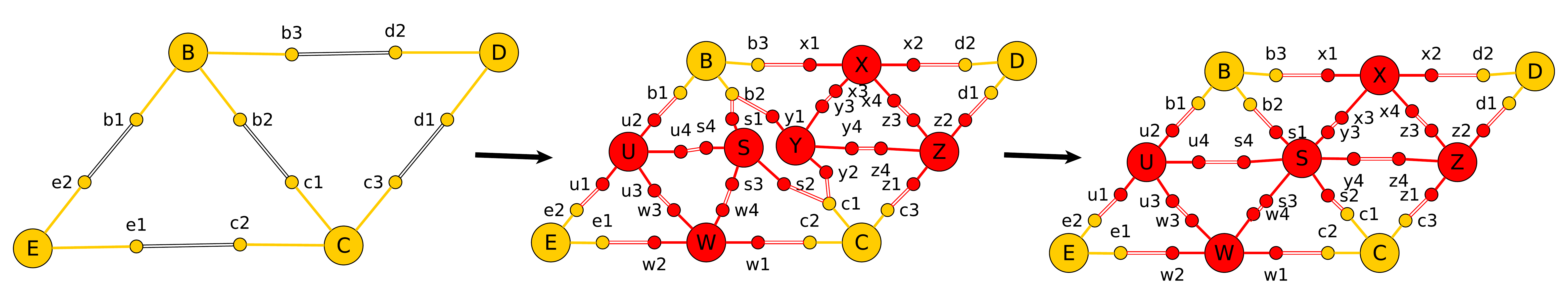}

$G$ \hspace*{3.7cm} $H$ \hspace*{3.7cm} $\overline H=G'$ 
\caption{A parallel rewrite step with overlapping between two triangles. Notice that two variants of $R_T$ with fresh new variables have been provided in order to produce the pregraph $H$. In the quotient graph $\overline H=G'$, $\class S= \{S,Y \}$, $\class{s_1}= \{s_1,y_1\}$, $\class{s_2}= \{s_2,y_2\}$.  }
\label{steps}
\end{figure} 
\end{example}
%%%%%%%%%%%%%%%%%%%%%%%%%%%%%%ù

As  a quotient pregraph is not necessarily a
graph (see Figure~\ref{ex2}), the above definition of parallel rewrite
step does not warranty, in general, the production of graphs only. 
% This means that graphs are not closed under parallel rewriting
%as defined above. 
Hence, we propose hereafter a sufficient condition, which could be
verified syntactically, that ensures that the outcome of a parallel
rewrite step is still a graph.

\begin{theorem} \label{theo:CStobeagraph}
Let $\trs$ be a \admissiblecsgrs\ $\trs = \{L_i \ra
  R_i \mid i= 1 \ldots n\}$. Let $G$ be a graph. Let  $I$ be a set of
variants of rules in $\trs$, $I = \{\lhs_i \ra \rhs_i \mid i = 1
\ldots k\}$ and $M$ a set of matches  $M=\{m_i^{a_i} : \lhs_i \ra G \mid i = 1
\dots k\}$. Let $G'$ be the pregraph such that $ G \topar_{I,M} G'$.
If $\forall p, p' \in \ports_{L_i}^{env}$,  $(p,p') \notin
\pp_{R_i}^{new}$,
then $G'$ is a graph.

\end{theorem}

 \forgetlongv{
 \begin{proof}
  We have to prove  that $H$ does not contain odd loops. Because of the previous constraint, it is enough to prove that 
  all ports of $H$ are not parts of a loop.
 \begin{itemize}
 \item If $p \in \cup_{i=1}^n \ports_{R_i}^{new}$, it is a new port contained in the graph $R_i$ thus $p$ has at most one connection port-port.

 \item If $p \in \cup_{i=1}^n \ports_{m_i^{a_i}(R_i)}^{env}$, $p $
   belongs to the graph $G$ and the only new port-port connections
   where $p$ is involved are those of $\cup_{i=1}^n \ports_{R_i}^{new}$.

 \item Else, if $p \in G/ \cup_{i=1}^n \ports_{R_i}^{new} \uplus
   \cup_{i=1}^n \ports_{m_i^{a_i}(R_i)}^{env}$,
 $p$ belongs to the non modified part of the graph. Its connections
 are unchanged and thus $p$ has at most one port-port connection.
 \item
 Finally, $p$ belongs to a path which is not a loop and $\overline H=G'$ is a graph.
 \end{itemize}
 \end{proof}
 }

%%%%%%%%%%%%%%%%%%%%%%%%%%%%%%%%%%%%%%%%%%%%%%%%%%%%%%%%%%%%%%%%%%%%%%%%%%%%%%%%%%%%%%%%
\vspace{0.3cm}
\section{Two Parallel Rewrite Relations}
\label{sect:4}

\newcommand{\longhookrightarrow}{}% teste si deja defini
\DeclareRobustCommand{\longhookrightarrow}{\lhook\joinrel\relbar\joinrel\rightarrow}
\newcommand{\longtwoheadrightarrow}{}% teste si deja defini
\DeclareRobustCommand{\longtwoheadrightarrow}{\relbar\joinrel\twoheadrightarrow}

\newcommand{\fpr} {\displaystyle \rightrightarrows}
\newcommand{\ar} {\displaystyle \rightrightarrows_{auto}}
\newcommand{\da} {distinguishing attributes} 
\newcommand{\syc} {symmetry condition} 

The set of matches, $M$, in Definition~\ref{defparallel} is not
constrained and thus the induced parallel rewrite relation is too
nondeterministic since at each step one may choose several
sets of matches leading to different rewrite outcomes. In this
section, we are rather interested in two confluent parallel rewrite
relations which are realistic and can be good candidates for
implementations.
The first one performs all possible reductions (up to node and port
renaming) whereas the second relation is more involved 
and  performs reductions up to left-hand sides'
automorphisms.

\subsection{Full Parallel Rewrite Relation}

We start by a technical definition of an equivalence relation, $\approx$, over matches.
\begin{definition}[$\approx$]
  Let $ L \to R$ be a rule and $G$ a graph. Let $l_1 \to r_1$ and
  $l_2 \to r_2$ be two variants of the rule $L \to R$. We denote by
  $h^{a_1}_1$ (respect. $h^{a_2}_2$) the (node, port and attribute) renaming mapping such
  that the restriction of $h^{a_1}_1$ (respectively, $h^{a_2}_2$) to
  $L \to \lhs_1$ (respectively $L \to \lhs_2$) is a graph
  isomorphism. 
%
%For $m_i,m_j \in M_{\infty}$, 
Let $m^{b_1}_1 : l_1 \to G$ and $m^{b_2}_2: l_2 \to G$ be two matches. We say that
$m^{b_1}_1$ and $m^{b_2}_2$ are equivalent and write  $m^{b_1}_1 \approx
m^{b_2}_2$ iff for all elements $x$ (in $\ports_{L}$, $\nodes_{L}$,
$\pp_{L}$ or $\pn_{L}$) of $L$,
$m^{b_1}_1(h^{a_1}_1(x))=m^{b_2}_2(h^{a_2}_2(x))$ and for all $x$ in
$\Att_{L}$, ${b_1}({a_1}(x)) = {b_2}({a_2}(x))$ .
\end{definition}

The relation $\approx$ is clearly an equivalence
relation. Intuitively, two matches $m^{b_1}_1 : l_1 \to G$ and
$m^{b_2}_2: l_2 \to G$ are equivalent, $m^{b_1}_1 \approx m^{b_2}_2$, whenever (i) $l_1$
and $l_2$ are left-hand sides of two variants of a same rule, say
$L \to R$, and (ii) $m^{b_1}_1$ and $m^{b_2}_2$ coincide on each element $x$ of
$L$.
%, expressed as $m_1(h_1(x))=m_2(h_2(x))$. 
%Notice that the fact
%$m^{b_1}_1(l_1) = m^{b_2}_2(l_2)$ is implied by, but not equivalent to, 
%$m^{b_1}_1(h^{a_1}_1(x))=m^{b_2}_2(h^{a_2}_2(x))$.

\begin{definition}[full parallel matches]
%Let us consider the possibly infinite set of matches
Let  $\grs$ be a graph rewrite system and $G$ a graph.
Let ${\cal M}_{\grs}(G)=\{m^{a_i}_i
: l_i \to G \; | \; m^{a_i}_i \mbox{ is a match and } l_i \to r_i \mbox{ is a variant of a rule
in }  \grs  \}$. A set, $M$, of \emph{full parallel matches}, with
respect to a graph rewrite system $\grs$ and a graph $G$, is a
maximal set such that 
%\begin{itemize} 
%\item 
(i) $M \subset  {\cal M}_{{\cal R}}(G)$ and 
%\item 
(ii) $\forall m^{a_1}_1,m^{a_2}_2 \in M, m^{a_1}_1 \not\approx m^{a_2}_2$.
%of not equivalent matches of ${\cal M}_{{\cal R}}(G)$
%$$M=\{m \in M_{\infty}, \not \exists m'\in M, m \approx m' \}$$
%\end{itemize}
\end{definition}

A set of full parallel matches $M$ is not unique because any rule in $\grs$ may have infinitely many
variants. However the number of non equivalent matches could be easily proven to be finite. 
\forgetlongv{
\begin{proposition}\label{finitude}
Let $M$ be a set of full parallel matches  with
respect to a graph rewrite system $\grs$ and a graph $G$. Then $M$ is finite.
% All maximal set of not equivalent matches of ${\cal M}_{{\cal R}}(G)$ have the same finite cardinality.
\end{proposition}
}
 \forgetlongv{\begin{proof}
 We assume that $G$ has a finite number of nodes, ports and attributes and $\grs$
 has a finite number of rules. 
Let $l_i  \to r_i$ be a rule in
 $\grs$. Let us assume now that nodes and ports of the left-hand side
 $l_i $ are attributed with the empty set. In this case,  matching $l_i$
 with subgraphs in $G$ remains to find a (non attributed) graph
 homomorphism between $l_i$ and $G$. Therefore, in this case,
 the number of possible matches of the
 left-hand side $l_i$ in graph $G$ is at most
 $\binom{k_i}{n} \times k_i!$ where  $n= card(\nodes_G) +
 card(\ports_G)$ and $k_i= card(\nodes_{l_i}) + card(\ports_{l_i})$.
  Thus $card(M)$ is bounded by $ \Pi_{1 \leq i \leq  card(\grs)}
 \binom{k_i}{n}  \times k_i!$ which is finite since $n$ and the
 $k_i$'s are finite.

Let us consider now the case where $l_i$ is attributed (that is to
say, there exists at leat a node or port, say x, such that $\att_{l_i}
(x) \not= \emptyset$). Let $m^a:l_i \to G$ be a match. $m$ is a
non-attributed graph homomorphism and $a : \Att_{l_i} \to \Att_G$ is
an attribute homomorphism which corresponds to a match over attributes in
the case where attributes in $l_i$ contain variables. We assume that
the matching problem over attributes is finitary. Thus for every $m$
there is a finite number, say $C_m$, of possible matchings over
attributes $a$. Let $l'_i$ be the graph obtained from $l_i$ by
removing all attributes (or equivalently said, by setting the
attribute function $\att_{l'_i}$ to the empty set.
 Let $C_i = max({C_m | m \mbox{ is
a non-attributed graph homomorphism } m:l'_i \to G})$. $C_i$ exists
since we assume that the matching problem is finitary.
 Then
  $card(M)$ is bounded by $ \Pi_{1 \leq i \leq  card(\grs)}
 \binom{k_i}{n}  \times k_i! \times C_i$ which is finite since $n$, the
 $k_i$'s and the $C_i$'s are finite.

 %if $m_i \approx m_j$, $m_i$ and $m_j$ are equal up to renaming of new nodes and
 %new ports of $r_i$ and $r_j$.
 \end{proof}
}
\begin{definition}[full parallel rewriting]
\label{fullparallel}
Let $\grs$ be an \esrs\ and $G$ a graph. Let $M$ be a set of
full parallel matches with respect to $\grs$ and $G$. We define the
\emph{full parallel} rewrite relation and write
$G \fpr_M G'$ or simply $G \fpr G'$, as the parallel rewrite step $ G \topar_{M} G'$.
\end{definition}

\begin{proposition}
\label{func}
Let $\grs$ be an \esrs. The rewrite relation $\fpr$ is
deterministic.  That is to say, for all graphs $g$,
$(g \fpr g_1 \; and \; g \fpr g_2)$ implies that $g_1$ and $g_2$ are
isomorphic.
\end{proposition}
\forgetlongv{
 \begin{proof}
   The proof is quite direct.  Let $M_1$ and
   $M_2$ be two different sets of full parallel matches such that
   $g \fpr_{M_1} g_1$ and $ g \fpr_{M_2} g_2$. By definition of sets
   of full parallel matches, %~\ref{fullparallel}
   for all matches $m^b\in M_1 $ there exists a match $m'^{b'}\in {M}_2$ such
   that $m^b \approx m'^{b'}$. Since $M_1$ and $M_2$ are finite (see
   Proposition~\ref{finitude}), there exists a natural number $k$ such
   that $M_1= \{ m_1^{b_1},m_2^{b_2},\ldots , m_k^{b_k} \}$ and
   $M_2= \{ m_1'^{b_1'},m_2'^{b'_2},\ldots , m_k'^{b'_k} \}$ such that for all
   $i \in \{1, \dots, k\}$, $m_i^{b_i} \approx m_i'^{b'_i}$.  Therefore, for every $i$
   such that $1 \leq i \leq k$, there exist a rule $L_i \to R_i$ in
   $\grs$ and two variants of it $l_i \to r_i$ and $l'_i \to r'_i$
   together with two renaming mappings $h_i^{a_i} : L_i \to l_i$ and
   $h_i'^{a'_i} : L_i \to l'_i$ such that for all elements $x \in L_i$,
   $m_i^{b_i}(h_i^{a_i}(x)) = m_i'^{b'_i}(h_i'^{a'_i}(x))$.

 By Definitions~\ref{defparallel} and \ref{fullparallel}, graphs $g_1$
 and $g_2$ are  quotient pregraphs of two pregraphs, respectively
 $H_1$ and $H_2$, obtained after the
 first step of parallel rewrite steps. The sets of nodes
 and ports of pregraphs $H_1$ and $H_2$ are defined as follows

 \begin{itemize}
 \item 
 $\nodes_{H_1} = (\nodes_g - \cup_{i=1}^k \nodes_{m_i^{b_i}(\lhs_i)}^{cut}) \uplus
   \cup_{i=1}^k \nodes_{\rhs_i}^{new}$
 \item  
 $\nodes_{H_2} = (\nodes_g - \cup_{i=1}^k \nodes_{{m'}_i^{b'_i}({\lhs'}_i)}^{cut}) \uplus
   \cup_{i=1}^k \nodes_{{\rhs'}_{i}}^{new}$
 \item 
 $\ports_{H_1} = (\ports_g - \cup_{i=1}^k \ports_{m_i^{b_i}(\lhs_i)}^{cut}) \uplus
   \cup_{i=1}^k \ports_{\rhs_i}^{new}$
  % where $\nodes$ is by turn $\nodes$, $\ports$, $\pp$ and $\pn$.
 \item
 $\ports_{H_2} = (\ports_g  \cup_{i=1}^k \ports_{{m'}_i^{b'_i}({\lhs'}_i)}^{cut}) \uplus
   \cup_{i=1}^k \ports_{{\rhs'}_i}^{new}$

\item
$\att_{H_1} = (\att_{g} - \cup_{i=1}^k \att_{m_i^{b_i}(\lhs_i)}^{cut})\uplus
\cup_{i=1}^k \att_{\rhs_i}^{new}$

\item
$\att_{H_2} = (\att_{g} - \cup_{i=1}^k \att_{{m'}_i^{{b'}_i}(\lhs_i)}^{cut})\uplus
\cup_{i=1}^k \att_{\rhs_i}^{new}$
 \end{itemize}  

 Now, We define a map $f^c : H_1 \to H_2$ by means of three maps on nodes, ports and attributes
 $f^c_N : \nodes_{H_1} \to \nodes_{H_2}$,
 $f^c_P : \ports_{H_1} \to \ports_{H_2}$ and 
$c: \Att_{H_1} \to \Att_{H_2}$ as follows

 \[
     f^c_N(x)= 
 \begin{cases}
     x & \text{if } x \in (\nodes_g - \cup_{i=1}^k \nodes_{m_i^{b_i}(\lhs_i)}^{cut})\\
     {h'}_i^{a'_i}( (h^{a_i}_i)^{-1}(x))              & \text{if }  x \in  \nodes_{\rhs_i}^{new}, for \; 1\leq i \leq k 
 \end{cases}
 \]
 
 \[
     f^c_P(x)= 
 \begin{cases}
     x & \text{if } x \in (\ports_g - \cup_{i=1}^k \ports_{m_i^{b_i}(\lhs_i)}^{cut})\\
     {h'}_i^{a'_i}( (h^{a_i}_i)^{-1}(x))               & \text{if }  x \in  \ports_{\rhs_i}^{new},  for \; 1\leq i \leq k
 \end{cases}
 \]

 and  %$c$ is defined as follows 
 $
     c(x)= 
 \begin{cases}
     b_i'\circ a_i' \circ a_i^{-1} \circ b_i^{-1}(x) & \text{if } \exists i \in \{1, \ldots, k\}, \exists e \in
     (\nodes_{H_1} \cup \ports_{H_1}) , (e,x) \in \att^{new}_{\rhs_i}\\
     x              & \text{otherwise}
 \end{cases}
 $

 $f^{c}$ is clearly a pregraph isomorphism between $H_1$ and $H_2$. As
 $g_1$ and $g_2$ are obtained as quotient pregraphs of  $H_1$ and $H_2$
 respectively, we conclude by using
 Proposition~\ref{iso-pregraph}, that  $g_1$ and $g_2$ are isomorphic.

 \end{proof}
   }
\begin{example} \label{ex_match}
$\qquad$
%\vspace{-0.6cm}
%\hspace{-3cm} 
\begin{multicols}{2}
%\begin{figure}[t] 
%\vspace*{-0.3cm}
\hspace{1.5cm}\includegraphics[scale=0.15]{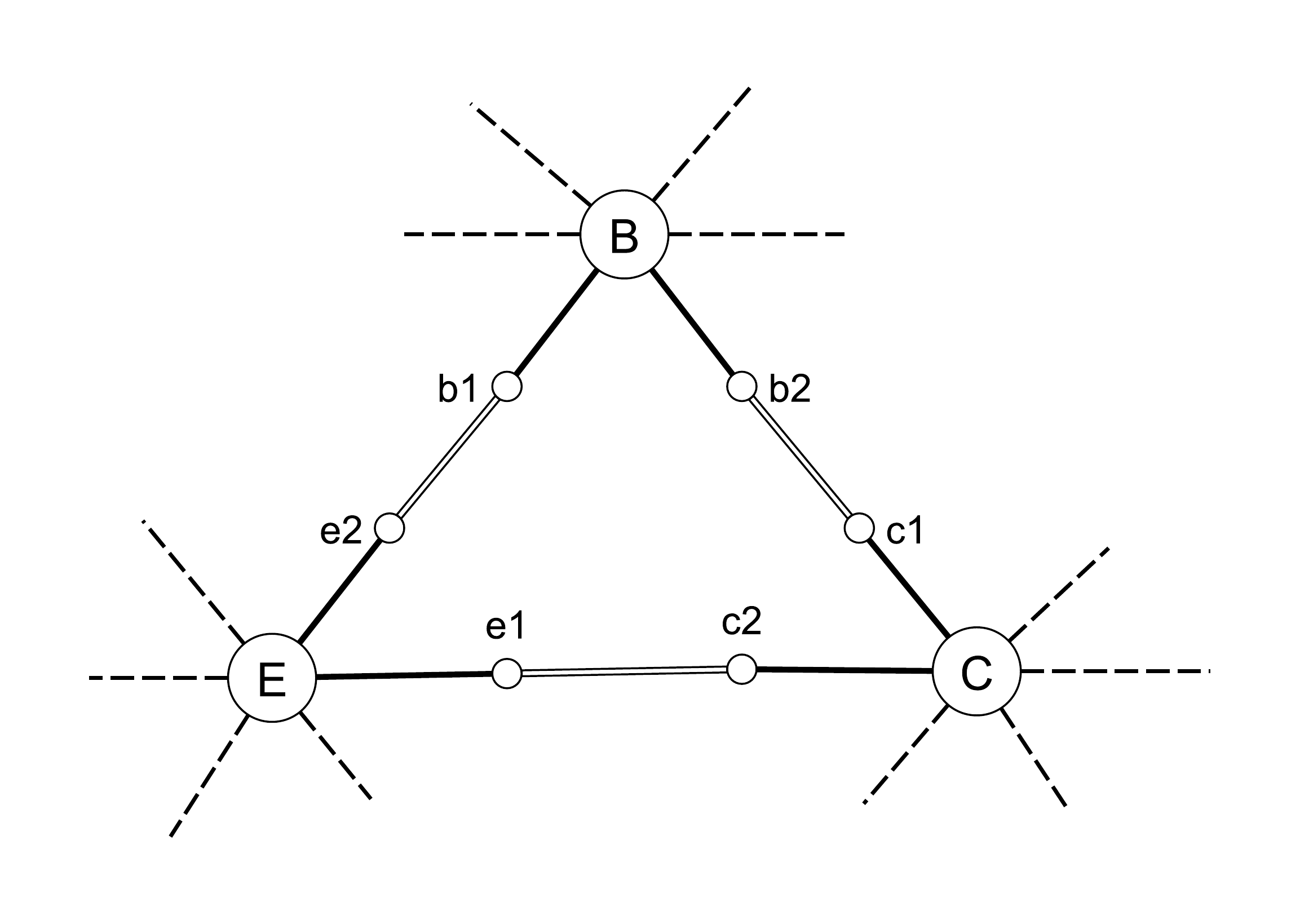}

%%\vspace{-0.5cm}
%\caption{ A subgraph to be matched}
%\label{subgraph}
%\end{figure} 

Let us consider the rule $R_T$ defined in Figure~\ref{ex5} and the
subgraph $s$ depicted on the side. %in Figure %~\ref{subgraph} at the
                             %right. 
The reader can verify
that there are six different matches, $m_1 \ldots m_6$, between the
left-hand side of $R_T$ and graph $s$.
\end{multicols}
\forgetlongv{
\noindent
 These matches are sketched below. Variants of
$R_T$ have been omitted for sake of readability. 

%\end{multicols}

\begin{itemize}
\item $m_1$ : $m_1(\alpha)=E; m_1(\beta)= B$; $m_1(\gamma)=C$; $m_1(\alpha_1)=e_1; m_1(\alpha_2)=e_2; m_1(\beta_1)= b_1; m_1(\beta_2)= b_2$; $m_1(\gamma_1)=c_1$; $m_1(\gamma_2)=c_2$.
%The isomorphism of the port-node and port-port connections are easily deduced.
\item $m_2$ : $m_2(\alpha)=E; m_2(\beta)= C$; $m_2(\gamma)=B$; $m_2(\alpha_1)=e_2; m_2(\alpha_2)=e_1; m_2(\beta_1)= c_2; m_2(\beta_2)= c_1$; $m_2(\gamma_1)=b_2$; $m_2(\gamma_2)=b_1$.
\item $m_3$ : $m_3(\alpha)=B; m_3(\beta)= E$; $m_3(\gamma)=C$; $m_3(\alpha_1)=b_2; m_3(\alpha_2)=b_1; m_3(\beta_1)= e_2; m_3(\beta_2)= e_1$; $m_3(\gamma_1)=c_2$; $m_3(\gamma_2)=c_1$.
\item $m_4$ : $m_4(\alpha)=B; m_4(\beta)= C$; $m_4(\gamma)=E$; $m_4(\alpha_1)=b_1; m_4(\alpha_2)=b_2; m_4(\beta_1)= c_1; m_4(\beta_2)= c_2$; $m_3(\gamma_1)=e_1$; $m_4(\gamma_2)=e_2$.
\item $m_5$ : $m_5(\alpha)=C; m_5(\beta)= B$; $m_5(\gamma)=E$; $m_5(\alpha_1)=c_2; m_5(\alpha_2)=c_1; m_5(\beta_1)= b_2; m_5(\beta_2)= b_1$; $m_5(\gamma_1)=e_2$; $m_5(\gamma_2)=e_1$.
\item $m_6$ : $m_5(\alpha)=C; m_5(\beta)= E$; $m_5(\gamma)=B$; $m_5(\alpha_1)=c_1; m_5(\alpha_2)=c_2; m_5(\beta_1)= e_1; m_5(\beta_2)= e_2$; $m_5(\gamma_1)=b_1$; $m_5(\gamma_2)=b_2$.
\end{itemize}

%For the sake of simplicity we have displayed the matches between the
%left-hand side of the rule
%and not  variants of the rule. Moreover, 
Here, the homomorphisms over attributes are always the identity, that
is why they have been omitted.
Thanks to the six matches and the rule $R_T$, the reader may
check that the subgraph $s$ can be rewritten, by using six different
variants of rule $R_T$, into a pregraph containing $3 \times 6$ new
nodes and $12 \times 6$ new ports.  The quotient pregraph has only $3$
new nodes but has $42$ new ports. Each pair of new nodes has $6$
connections.
}
\end{example}

This example shows that the \fp\ rewriting has to be used carefully since it may produce non intended results due to overmatching the same subgraphs. To overcome this issue, one may use attributes in order to lower the possible matches.
We call such attributes \emph{\da}. % (see Example XX) or define new rewrite relations which get rid of matches belonging to a same automorphic class.
\begin{figure}[t] \centering
\includegraphics[scale=0.2]{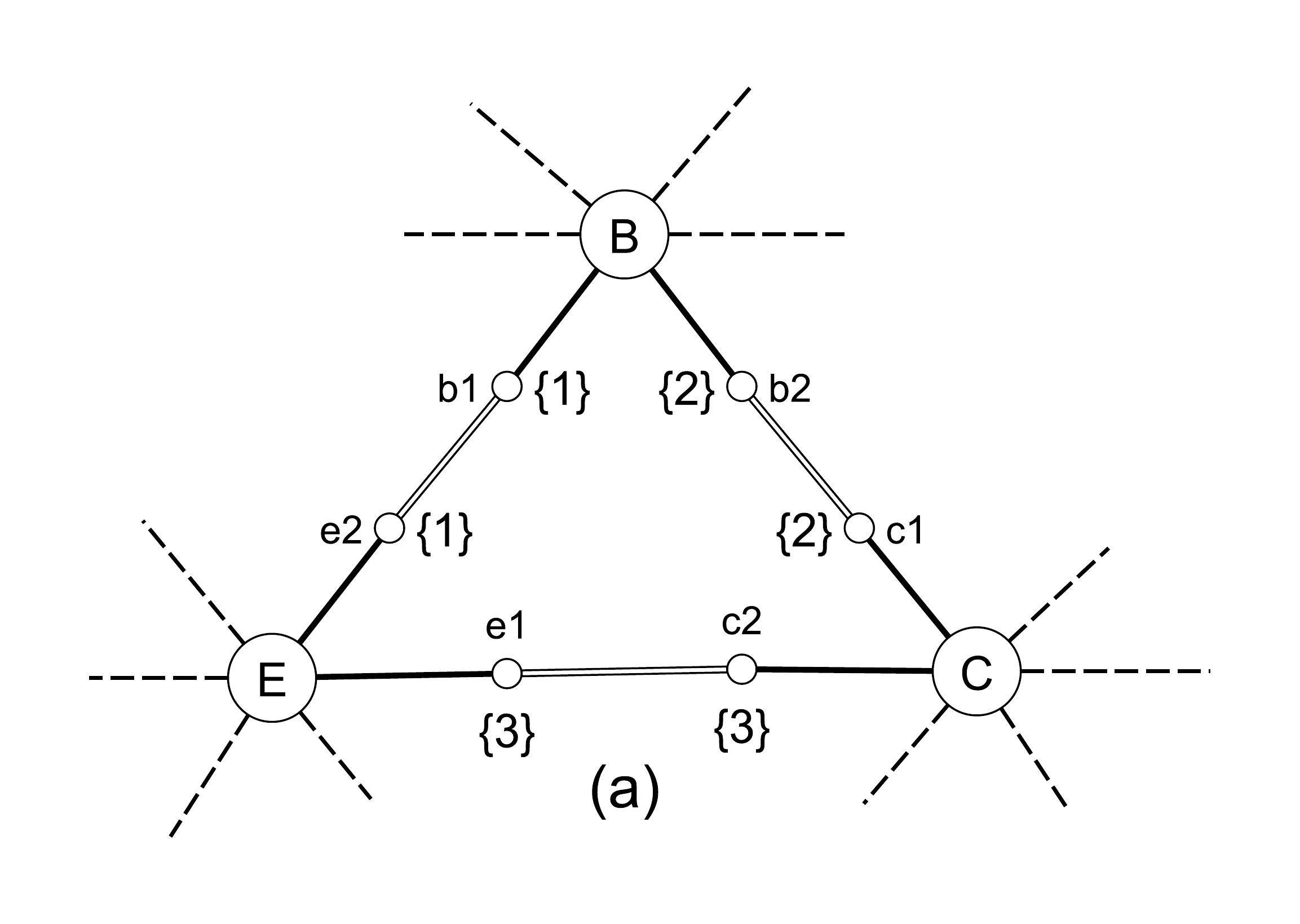}
\hspace*{-0.5cm} \includegraphics[scale=0.2]{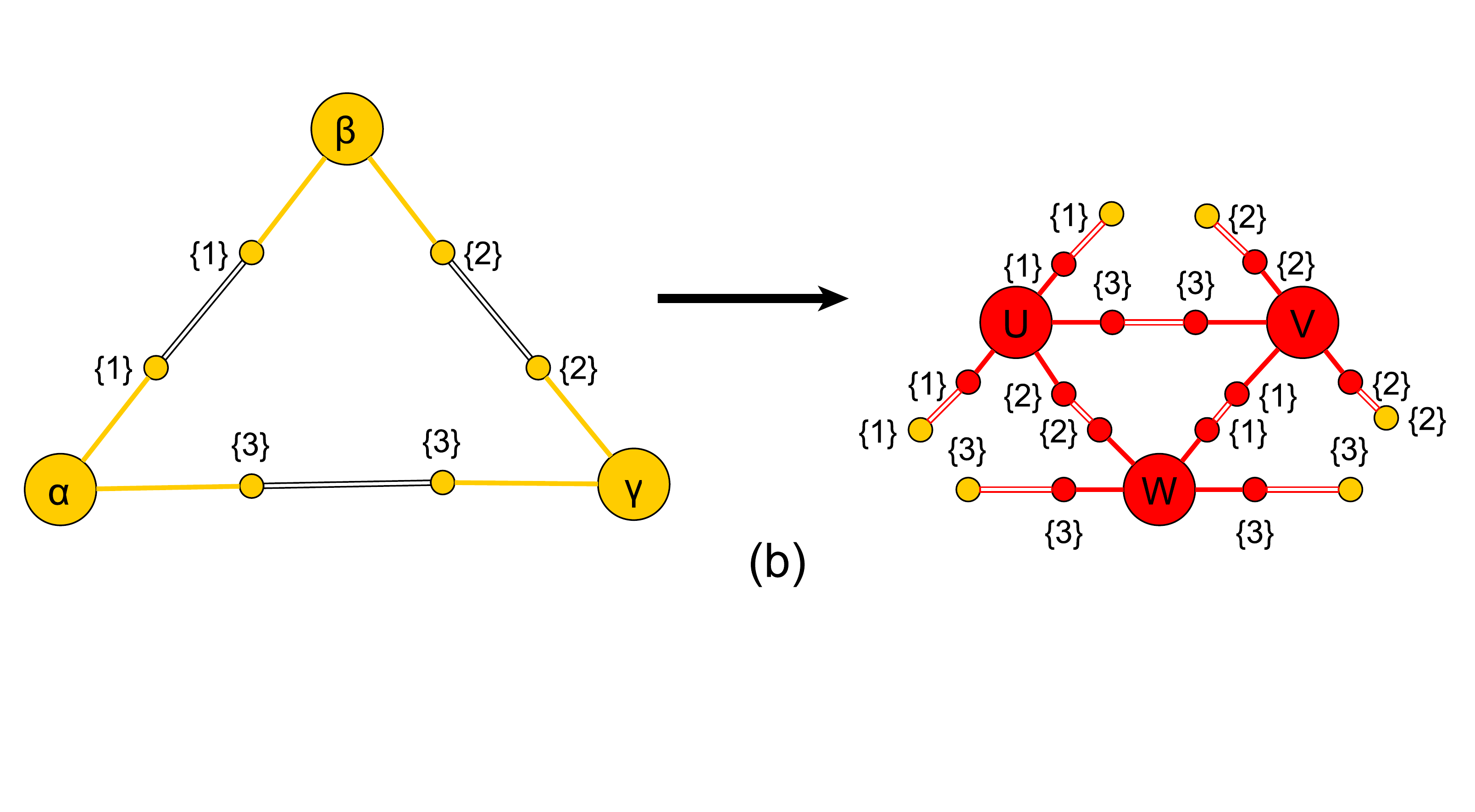}
%\vspace{-0.5cm}
\caption{(a) Subgraph $s$ with distinguishing attributes on ports. The attributes are $\{1, 2, 3 \}$. (b) Rule $R_T$ with distinguishing attributes.}
\label{subgraph_att}
%\vspace{-0.3cm}
\end{figure} 
%
%\begin{example}
  In order to consider only one match of the subgraph $s$ considered in
   Example~\ref{ex_match} by the rule $R_T$, one option is to apply 
  \fp\ rewrite relation with
  \da\ %(see distinguishing number \cite{Albertson - Collins}
  on the subgraph depicted in Figure~\ref{subgraph_att} (a) and
   rule $R_T$ with distinguishing attributes given in
  Figure~\ref{subgraph_att} (b),
leading to a pregraph whose
quotient  is a graph with  $3$ new nodes and $12$ new ports.
This graph is the expected one.

Another way to mitigate the problems of overmatching subgraphs, in
addition to the use of distinguishing attributes, consists in taking
advantage of the symmetries that appear in the graphs of rewrite
rules.  This leads us to define a new rewrite relation which gets rid
of multiple matches of the same left-hand-side of a fixed
rule. % belonging to a same automorphic class.
We call this relation \emph{\pa s} and is defined
below.% and is defined hereafter.
%%%%%%%%%%%%%%%%%%%%%%%%%%%%%%%%%%%%%%%%%%%%%%%%%%%%%%%%%%%%%%%%%%%%%%%%%%%%%%%%%%%%%
\subsection{Parallel  Rewrite Relation up to Automorphisms} 
\newcommand{\M}{{\cal M}_{\lhs}}

Let us consider a graph $g$ which rewrites into $g_1$ and $g_2$ using
an \esrr\ $l \to r$. This means that there exist two matches
$ \mu^{\beta_i}_i: l \to g$ with $i \in \{1, 2\}$ such that
$g \Rightarrow_{l \rightarrow r, \mu^{\beta_i}_i} g_i$. One may wonder
whether $g_1$ and $g_2$ are the same (up to isomorphism) whenever
matches $ \mu^{\beta_1}_1$ and $\mu^{\beta_2}_2$ are linked by means of an
automorphism of $l$. That is to say, when there exists an automorphism
$h^{a}: \lhs \to \lhs$ with $\mu^{\beta_1}_1= \mu^{\beta_2}_2 \circ
h^{a}$.
Intuitively, matches $ \mu^{\beta_1}_1$ and $\mu^{\beta_2}_2$ could be considered
as the same up to a permutation of nodes. We show below that $g_1$ and
$g_2$ are actually isomorphic but under some syntactic condition we
call \emph{symmetry condition}.

Notation: Let $g$ be a graph with attributes in $\Att$. We write
$H(g)$ to denote the set of automorphisms of $g$, i.e. $H(g)$
is the set of isomorphisms $ h^{a}: g \to g$, with $a$ being an
isomorphism on the attributes of $g$, $a: \Att \to \Att$.

\begin{proposition} \label{prop-iso} Let $l \to r$ be an \esrr.  let
  $l_1 \to r_1$ and $l_2 \to r_2$ be two variants of the rule
  $l \to r$.  Let $v_1^{c_1}$, $v_1'^{c_1}$, $v_2^{c_2}$, $v_2'^{c_2}$
  be the isomorphisms reflecting the variant status of these two rules
  with $v_1^{c_1}: l \to l_1$, $v_1'^{c_1}: r \to r_1$,
  $v_2^{c_2}: l \to l_2$ and $v_2'^{c_2}: r \to r_2$ such that
  $l_i = v_i^{c_i}(l)$, $r_i = v_i'^{c_i}(r)$ and
  $v_i^{c_i}(r_i^{env}) = v_i'^{c_i}(r_i^{env})$ for $i \in \{1, 2\}$.
  Let $g$ be a graph and $g'_1$ and $g'_2$ be two pregraphs.  Let
  $g \Rightarrow_{l_1 \rightarrow r_1, m^{b_1}_1} g_1'$ and
  $g \Rightarrow_{l_2 \rightarrow r_2, m^{b_2}_2} g_2'$ be two rewrite
  steps such that there exist two automorphisms $h^{a}: \lhs \to \lhs$
   and $h'^{a}: \rhs \to \rhs$
  such that (i) with $m^{b_1}_1=m^{b_2}_2 \circ v_2^{c_2}\circ h^a \circ (v_1^{c_1})^{-1}$ and (ii) for all elements $x$ of $r^{env}$, $h'^{a}(x)=h^{a}(x)$.
  Then, $g_1'$ and $g_2'$ are isomorphic.

\end{proposition}

 \begin{proof}[Sketch] 
   The sketch of the proof is depicted in Figure~\ref{iso}. The
   attributes structures used in the rule $l \to r$ (respectively,
   $l_1 \to r_1$ and $l_2 \to r_2$) are denoted $A$ (respectively,
   $A_1$ and $A_2$) whereas the attibutes structure of the transfomed
   graph $g$ is denoted $B$. 
% %
% %
   From the hypotheses, we can easily infer the exitence of two
   isomorphisms $h_v^{c_v}: l_1 \to l_2$ and $h_v'^{c_v}:r_1 \to r_2$ such that
   $h_v^{c_v} = v_2^{c_2}\circ h^a \circ (v_1^{c_1})^{-1} $ and
   $h_v'^{c_v} = v_2'^{c_v}\circ h'^{a} \circ (v_1'^{c_1})^{-1}$.
And we have $c_v=c_2 \circ a \circ c_1^{-1}$.

 Let $g \Rightarrow_{l_1 \rightarrow r_1, m_1^{b_1}} g_1'$ and
   $g \Rightarrow_{l_2 \rightarrow r_2, m_2^{b_2}} g_2'$ 
 such that $m_1^{b_1}(l_1) = m_2^{b_2}(l_2)$.
 By definition of a rewrite step, there exist a pregraph $g_1$ (respect.
 a pregraph $g_2$) and an injective
 homomorphism $m_1'^{b_1}: \rhs_1 \to g_1$ (respect. $m_2'^{b_2}: \rhs_2 \to g_2$)
 such that $g'_1 = \overline g_1$ (respect. $g'_2 = \overline g_2$). Moreover,
 since, by definition,  $r^{env}$ is included in $l^{env}$ for any \esrr\ $l \to r$, we have $m_1'^{b_1}(r_1^{env})=m_1^{b_1}(r_1^{env})$
 (respect. $m_2'^{b_2}(r_2^{env})=m_2^{b_2}(r_2^{env})$), where $
 m_i'^{b_i}$, for  $i \in \{1, 2\}$, are defined as follows:
\forgetshortv{
\noindent
 for $n \in \nodes_{r_i}, {m'}^{b_i}_i(n)= \left \{ \begin{array}{l} m^{b_i}_i(n) \;
                                             if \; n \in \nodes_{r_i}^{env}\\
 n \; otherwise  \end{array} \right.$
 for $p \in \ports_{r_i}, {m'}^{b_i}_i(p)= \left \{ \begin{array}{l} m^{b_i}_i(p) \;
                                             if \; p \in \ports_{r_i}^{env}\\
 p \; otherwise  \end{array} \right.$
}

\forgetlongv{

 for $n \in \nodes_{r_i}, {m'}^{b_i}_i(n)= \left \{ \begin{array}{l} m^{b_i}_i(n) \;
                                             if \; n \in \nodes_{r_i}^{env}\\
 n \; otherwise (n \in \nodes_{r_i}^{new}) \end{array} \right.$

 for $p \in \ports_{r_i}, {m'}^{b_i}_i(p)= \left \{ \begin{array}{l} m^{b_i}_i(p) \;
                                             if \; p \in \ports_{r_i}^{env}\\
 p \; otherwise (p \in \ports_{r_i}^{new}) \end{array} \right.$
}

 Now, let us define the isomorphism $h''^d: g_1 \to g_2$ with $d (x) =
 $ if $ x \in b_1(A_1)$ then $b_2 \circ c_v \circ b_1^{-1}(x)$ else $x$.
 Let us consider $x$
 such that $x$ is an element of $r^{env}$ (port or node).
 We have $m_1^{b_1}(v_1^{c_1}(x))=m_2^{b_2}(v_2^{c_2}(h^{a}(x)))$ is an element of $g$. 
 Moreover $m_1'^{b_1}(v_1'^{c_1}(x)) \in g_1 $ and $m_2'^{b_2}(v_2'^{c_2}({h'}^{a}(x)) \in g_2$.
 Let us denote $y=m_1^{b_1}(v_1^{c_1}(x))$.
 By construction $m_1'^{b_1}(v_1'^{c_1}(x))=m_1^{b_1}(v_1^{c_1}(x))=y$ because $x \in r^{env}$.
 From the hypothesis we have  $h^{a}(x)={h'}^{a}(x)$. Thus
  $m_2'^{b_2}(v_2'^{c_2}({h'}^{a}(x))=m_2'^{b_2}(v_2'^{c_2}(h^{a}(x))$ and then we have
 $m_2'^{b_2}(v_2'^{c_2}({h'}^{a}(x))=m_2^{b_2}(v_2^{c_2}(h^{a}(x))=y$.
 Then,
 for all elements  $z$ of
 the non-modified part of $g$ which is $g-m_1^{b_1}(v_1^{c_1}(l))$ ($z$ can be a port or a node if $y$ is not a node) such that $(z,y) \in g$,
  we have that $(z,y) \in g_1$ and $(z,y)\in g_2$ and ${h''}^d=Id^d$ on $g_1-m_1'^{b_1}(v_1'^{c_1}(r))$.
 Finally the definition of $h''^d$ is :

 For $y \in \nodes_{g_1} \cup\ports_{g_1}, h''^d(y)= \left \{ \begin{array}{l} m_2'^{b_2}(h_v'^{c_v}((m_1'^{b_1})^{-1}(y))) \; if \; y \in m_1'^{b_1}(r_1)\\
 y \; otherwise \end{array} \right.$

 For all types of existing connections $(y,z)$ of $g_1$ where $y$ and
 $z$ in $ \nodes_{g_1} \cup\ports_{g_1}$, $h''^d(y,z)=(h''^d(y),h''^d(z))$ is
 in $g_2$. By construction, the homomorphism conditions on attributes
 are fulfilled by $h''^d$. Thus,  $h''^d: g_1 \to g_2$ is a pregraph homomorphism.
 In addition, $h''^d$ is  bijective by construction.
 From $h''^d$ and 
 Proposition~\ref{iso-pregraph}, we infer the
 isomorphism $h^{(3)d}:  g'_1 \to  g'_2$.
% where $g'_1 = \overline g_1$ and $g'_2 = \overline g_2$.

 \begin{figure}[t] \centering
 \includegraphics[scale=0.3]{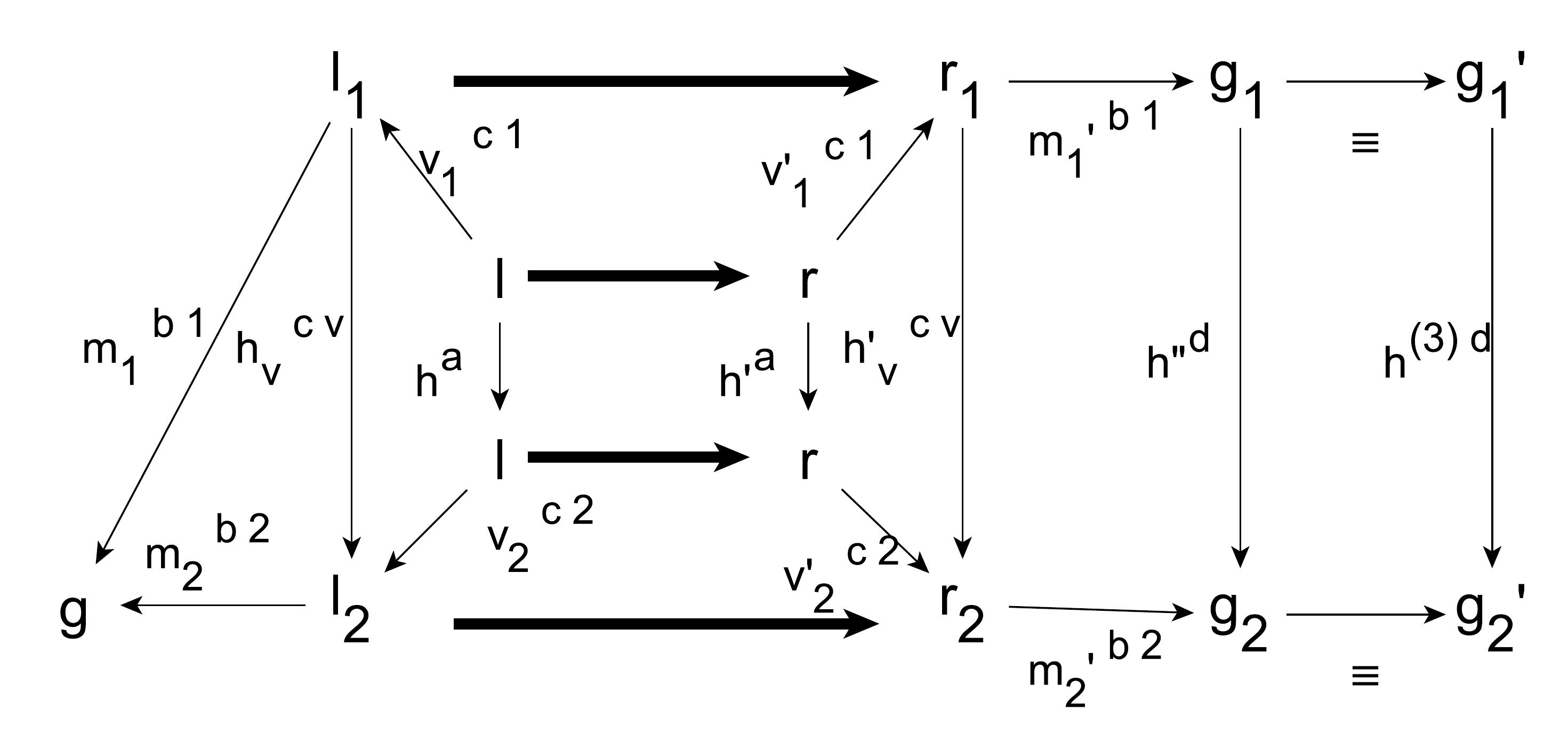} \hspace{1.2cm} \includegraphics[scale=0.3]{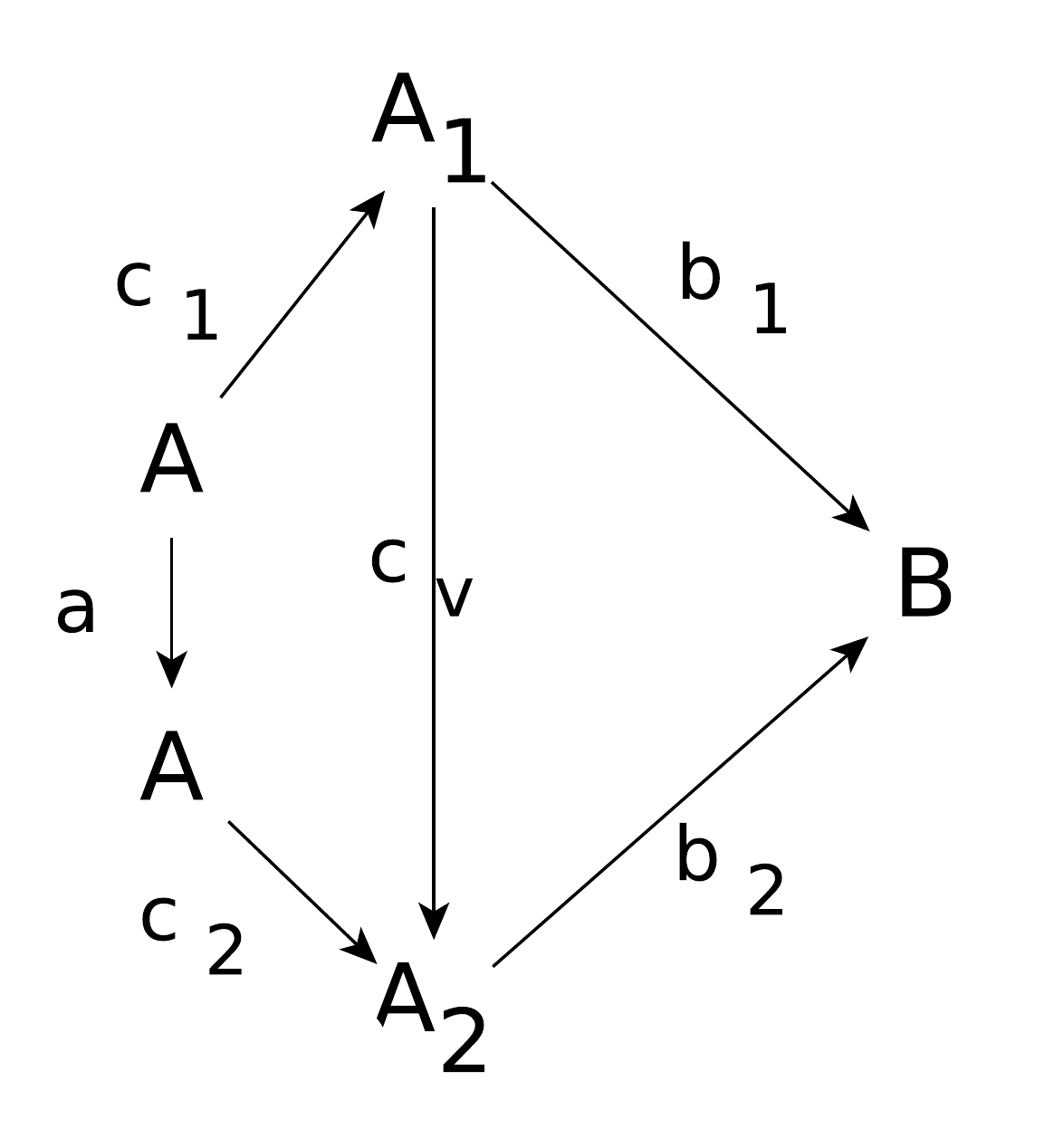}
 \caption{ Sketch of the proof of Proposition \ref{prop-iso}}
 \label{iso}
 \end{figure} 

 \end{proof}

\begin{definition}[Symmetry Condition]  
An \esrr\ $\lhs \rightarrow \rhs$ verifies the
\emph{\syc} \ iff 
$\forall  h^{a} \in H(\lhs), \; \exists \; h'^{a} \in H(\rhs), \;
\mbox{such that}  \;\forall x \in \rhs^{env}, h^{a}(x)=h'^{a}(x)$
\end{definition}

 The reader can check that the rule $R_T$ verifies the \syc.

\begin{definition}[Matches up to automorphism, $\sim^{\lhs}$]
  Let $\lhs \to \rhs$ be an \esrr\ satisfying the symmetry condition.
Let $\lhs_1 \to \rhs_1$ and $\lhs_2 \to \rhs_2$ be two different variants of
the rule $\lhs \to \rhs$. Let $v^{c_1}_1 : \lhs \to \lhs_1$ and $v^{c_2}_2 :
\lhs \to \lhs_2$ be the isomorphisms that reflect the variant status
of $\lhs_1$ and $\lhs_2$ of $\lhs$. Let $m^{b_1}_1: \lhs_1 \to g$ and $m^{b_2}_2:
\lhs_2 \to g$ be two matches such that $m^{b_1}_1(\lhs_1) = m^{b_2}_2(\lhs_2)$. We
say that matches $m^{b_1}_1$ and $m^{b_2}_2$ are equal up to ($l$-)automorphism  and
write $ m^{b_1}_1 \sim^{\lhs} m^{b_2}_2$
iff
there exists an automorphism $h^{a} : \lhs \to \lhs$ such that $m^{b_1}_1 = m^{b_2}_2 \circ v^{c_2}_2
\circ h^{a} \circ {v^{c_1}_1}^{-1}$. 
\end{definition}

\begin{definition}[Rewriting up to automorphisms]
  Let $\grs$ be a  \admissiblecsgrs\ whose rules satisfy the symmetry
  condition and $g$ a graph.  Let
  $M(\grs,g)^{auto} = \{m^{a_i}_i: \lhs_i \to g \; | \; \lhs_i \mbox{ is the
    left-hand side of a variant of a } \mbox{rule } \lhs \to \rhs \mbox{ in } \grs 
  \mbox{ and } m^{a_i}_i \\ \mbox{is a match up to automorphism} \}$.
  We define the rewrite relation $\ar$ which rewrites graph $g$ by
  considering only matches up to automorphisms. I.e., the set of matches $M$ of Definition~\ref{defparallel} is 
  $M(\grs,g)^{auto}$.

%$M=\{m_i : l_i \to G \; | \; l_i \to r_i \;is\; a\; rule\; in\;  R and \not\exists j, m_j \sim m_i \}$.

\end{definition}
\forgetlongv{
\begin{remark} For all two matches $m^{b_1}_1$ and $m^{b_2}_2$ in   $M(\grs,G)^{auto}$, $m^{b_1}_1
\not\sim^l m^{b_2}_2$. This means that the choice of matchings in
$M(\grs,G)^{auto}$ are not unique. From every equivalence class of a
match w.r.t. the equivalence relation $\sim^l$, only one
representative is considered. Therefore, one may wonder if the relation $\ar$ is
confluent. The answer is positive, that is to say, whatever the match 
representatives are chosen (up to automorphism), the relation $\ar$
rewrites a given graph to a same pregraph \emph{up to isomorphism}.
\end{remark}
}

%he following statement establishes that the relation is deterministic.
\begin{theorem}\label{theo:automorphism}
Let $\grs$ be a  \admissiblecsgrs\ whose rules satisfy the symmetry
  condition. Then $\ar$ is deterministic. That is, for all graphs $g$,  
  ($g \ar g'_1$ and $g \ar g'_2$) implies that $g'_1$ and $g'_2$
  are isomorphic.

\end{theorem}
\forgetlongv{
\begin{proof}[Sketch]
 Let $M_1$ (resp. $M_2$) be the set of matches used in the rewrite step
 $g \ar g'_1$ (resp. $g \ar g'_2$). Let us assume that $M_1 \not= M_2$.
 By definition of sets $M_1$ and $M_2$, for all matches $m_i^{b_i}: l_i \to g$
 in $M_1$, there exits a match $m_i'^{b_i'}: l'_i \to g$ in $M_2$ such that $l_i$ and
 $l_i'$ are the left-hand sides of two variants $l_i \to r_i$ and  $l'_i
 \to r'_i$ of a rule $l \to r$ in
 $\grs$ such that $m_i^{b_i}(l_i) = m_i'^{b_i'}(l'_i)$. That is to say, there exist four isomorphisms reflecting the
   variant status of these two rules, say $v_i^{c_i}: l \to l_i$,
   $v_i'^{c_i}: r \to r_i$, $w_i^{d_i}: l \to l'_i$ and $w_i'^{d_i}: r \to r'_i$ such that
   $l_i = v_i^{c_i}(l)$, $r_i = v_i'^{c_i}(r)$, $l'_i = w_i^{d_i}(l)$, $r'_i = w_i'^{d_i}(r)$, 
   $\forall x \in r_i^{env}, \,   v_i^{c_i}(x) = v_i'^{c_i}(x)$ and
   $\forall x \in r_i^{'env}, \,  w_i^{d_i}(x) = w_i'^{d_i}(x)$.

   From the hypotheses, there exist two automorphisms $h_i^{a_i}: l_i \to l_i$
   and $h_i'^{a_i}: r_i \to r_i$ and two isomorphisms ${h_v}_i^{e_i}: l_i \to l'_i$
   and ${h'_v}_i^{e_i}: r_i \to r'_i$ such that
   ${h_v}_i^{e_i} = w_i^{d_i} \circ h_i^{a_i} \circ (v_i^{c_i})^{-1}$ and
   ${h'_v}_i^{e_i} = w_i'^{e_i} \circ h_i'^{a_i} \circ (v_i^{'c_i})^{-1}$.

 By following the same reasoning as in Proposition~\ref{prop-iso}, we
 can build $h''^f: g_1 \to g_2$ defined as follows, where $\mu_i^{t_i} : r_i \to
 g_1$ and $\mu_i'^{t_i} : r'_i \to g_2$ are induced by definition of
 rewrite steps ($\mu_i^{t_i}$ and $\mu_i'^{t_i}$ play the same role, for every two
 rules, as $m_1'^{b_1}$ and $m_2'^{b_2}$ in the proof of Proposition~\ref{prop-iso}).

 for $n\in \nodes_{g_1}, h''^f(n)= \left \{ \begin{array}{l}  \mu_i'^{t_i}(
                                            {h'_v}_i^{e_i}((\mu_i^{t_i})^{-1} (n)))
                                            \; if \; n \in \mu_i^{t_i}(r_i)\\
 n \; otherwise \end{array} \right.$

% %\item
 for $p\in \ports_{g_1}, h''^f(p)= \left \{ \begin{array}{l}  \mu_i'^{t_i}(
                                            {h'_v}_i^{e_i}((\mu_i^{t_i})^{-1} (p)))
                                            \; if \; p \in \mu_i^{t_i}(r_i)\\
 p \; otherwise \end{array} \right.$

% %\item
 for $(p,n)\in \nodes_{g_1}, h''^f(p,n)=(h''^f(p),h''^f(n))$

% %\item 
 for $(p,p')\in \nodes_{g_1}, h''^f(p,p')=(h''^f(p),h''^f(p'))$
% %\end{itemize}

 Clearly $h''^f$ is an isomorphism between pregraphs $g_1$ and $g_2$. Therefore, by
 Proposition~\ref{iso-pregraph}, $g'_1$ (which equals $\overline g_1$) is isomorphic to
 $g'_2$ (which equal $\overline g_2$). 

 \end{proof}
}

%%%%%%%%%%%%%%%%%%%%%%%%%%%%%%%%%%%%%%%%%%%%%%%%%%%%%%%%%%%%%%%%%%%%%%%%%%%%%%%%%%%
\section{Examples}
\label{sect:5}

\forgetshortv{
We illustrate the proposed framework through two examples.  We provide simple confluent rewrite systems encoding  the Koch snowflake and the mesh refinement.

}

\forgetlongv{
We illustrate the proposed framework through three examples borrowed
from different fields. We particularly provide simple confluent rewrite systems encoding cellular automata, the koch snowflake and the mesh refinement.

\subsection{Cellular automata (CA)}

A cellular automaton is based on a fixed grid composed of cells. 
 Each cell computes its new state synchronously. At instant
 $t+1$, the value of a state $k$,  denoted $x_k(t+1)$ may depend
  on the valuations at instant $t$ of
  the state $k$ itself, $x_k(t)$, and the states  $x_n(t)$ such
  that $n$ is a neighbor of $k$. Such a formula is of the following
  shape, where $f$ is a given function and  $\nu(k)$ is the set of the
  neighbors of cell $k$: $x_k(t+1)=f(x_k(t),x_n(t), n \in \nu(k))$
% $\sigma$ corresponds to the set of neighbors of the node $k$. 
In the case of a graph $g$, the neighbors of a cell (node) $k$,
$\nu(k)$, is defined by :  
 $l \in \nu(k)$ iff $\exists p_1, \exists p_2, (p_1, k) \in \pn_g \wedge (p_2,l) \in {\pn_g} \wedge (p_1,p_2) \in {\pp_g}$.
Usually, the grid is oriented such that any cell of $\nu(k)$ has a unique relative position with respect to the cell $k$.
This orientation is easily modeled by distinguishing attributes on ports.
For instance, one can consider  Moore's neighborhood~\cite{Moore69} on a 2-dimensional grid. This neighborhood of radius 1 is composed of 8 neighbors. The distinguishing attributes on ports belong to the set $\Att=\{e,w,n,s,ne,se,nw,sw \}$
which defines the 8 directions where e = east, w = west, n = north, s =
 south etc. %, north-east, north-west, south-east and south-west.

The grid is defined by a graph 
 $g=( {\N_g}, {\ports_g}, {\pn_g}, \pp_g, \Att_g, \att_g)$ such that :
\begin{itemize}
\item $\N_g=\{m_{i,j}\}_{i \in I, j \in J} $, where intervals $I$ and
  $J$ are defined as $I=[- N,N]\cap \mathbb{Z} $ and $J=[- N',N']\cap
  \mathbb{Z} $ for some natural numbers $N$ and $N'$.

\item ${\ports_g}=\{ e_{i,j} , w_{i,j},s_{i,j},n_{i,j} , ne_{i,j} , nw_{i,j},$
$se_{i,j},sw_{i,j} |\  i \in I, j \in J\} $,

\item ${\pn_g}=\{
(e_{i,j}, m_{i,j}  ) , (w_{i,j}, m_{i,j}),(s_{i,j},m_{i,j}),   
(n_{i,j},m_{i,j}), (ne_{i,j},m_{i,j}),
 (nw_{i,j},m_{i,j}), (se_{i,j},
m_{i,j}),$
$(sw_{i,j},m_{i,j}) | i \in I, j \in J \}$,

\item ${\pp_g}=\{
(e_{i,j}, w_{i,j+1}), (w_{i,j}, e_{i,j-1}), (n_{i,j}, $
$s_{i-1,j}), (s_{i,j}, n_{i+1,j}),
 (ne_{i,j}, sw_{i-1,j+1}), $
$(se_{i,j},\\ nw_{i+1,j+1}), (nw_{i,j}, se_{i-1,j-1}), (sw_{i,j}, $
$ne_{i+1,j-1}) | \ i \in I, j \in J  \}$,

\item $\forall i \in I, \forall j \in J $,  ${\att_g}(m_{i,j}) \subseteq \Att_{g}$,
\item $\forall i \in I, \forall j \in J $, ${\att_g}(e_{i,j})= \{e\}$, ${\att_g}(w_{i,j})=\{w\}$,${\att_g}(s_{i,j})=\{s\}$, ${\att_g}(n_{i,j})=\{n\}$, ${\att_g}(ne_{i,j})=\{ne$, ${\att_g}(nw_{i,j})=\{nw\}$, ${\att_g}(se_{i,j})=\{se\}$, ${\att_g}(sw_{i,j})=\{sw\}$.
\end{itemize}

The attributes of the nodes correspond to states of the cells. They
belong to a set $\Att$. To implement the dynamics of the
automaton one needs only one rewrite rule $\{\rho = \lhs \to \rhs\}$ which
corresponds to the function $f$. The rule does not modify the
structure of the grid but modifies the attributes of nodes. Thus a
left-hand side has a structure of a star with one central node (see Figure~\ref{game}), for
which the rule at hand expresses its dynamics, surrounded by its
neighbors. Nodes, ports and edges of the left-hand side belong to the
environment part of the rule. Only the attribute of the central node
belongs to the cut part since this attribute is modified by the rule. In the left-hand-side, the attributes of
nodes are variables to which values are assigned during the matches.
The right-hand-side is reduced to a single node named $i$. Its
attribute corresponds to the new part of the right-hand side. 

Figure~\ref{game} illustrates such rules by implementing the well
known \emph{game of life}. It is defined using
Moore's neighborhood and the dynamics of the game is defined on a
graph $g$ such that attributes of nodes are in $\{0,1 \}$ and

$x_i(t+1)= ((\sum_{l \in \nu(i)}x_l(t))=?=3) + ((x_i(t)=?=1) $
$\times (\sum_{l \in \sigma(i)}x_l(t))=?=2)) $

where $(x=?=y) \Leftrightarrow \left \{ \begin{array}{l}  1 \; if \; x=y\\0 \; otherwise
\end{array} \right .$

The neighborhood of a node $i$ and its dynamics verify the symmetry
condition, thus there is no need to define attributes on ports. 
The rewriting relation $\ar$ is applied 
 on the rewrite system ${\cal R}=\{\rho=\lhs \to  \rhs \}$ reduced to
 one rule depicted in Figure~\ref{game}. More precisely the graphs of
 the rule as defined as follows: 

$\lhs=(\nodes_l,\ports_l, \pn_l, \pp_l, \Att_l, \att_l)$ with
\begin{itemize}
\item
$\nodes_l=\nodes_l^{env}=\{i,a,b,c,d,e,f,g,h\}$,
\item
$\ports_l = \ports_l^{env}= \{i_1,i_2,i_3,i_4,i_5,i_6, i_7,i_8,$
$a_1,b_1,c_1,d_1,e_1,f_1,g_1,h_1\}$,
\item $\pn_l=\pn_l^{env}=\{(a_1,a),(b_1,b),(c_1,c),(d_1,d),$
$(e_1,e),(f_1,f),(g_1,g),(h_1,h),(i_1,i)(i_2,i),(i_3,i),$
$(i_4,i),(i_5,i),(i_6,i),(i_7,i),(i_8,i) \}$,
\item
$\pp_l=\pp_l^{env}=\{(i_1,a_1)(i_2,b_1),(i_3,c_1),(i_4,d_1),$
$(i_5,e_1),(i_6,f_1),(i_7,g_1),(i_8,h_1) \}$.
\item $\Att_l = \{0,1, x_i\} \cup \{y_q \; | \; q \in \{a,b,c,d,e,f,g,h\}\}$ and $\att_l= \att_l^{env} \cup \att_l^{cut}$ with
$\att_l^{cut}: \{i\} \to \Att_l$ such that $\att_l^{cut}(i) =\{x_i\}$
; and
$\att_l^{env}:  \{a,b,c,d,e,f,g,h \} \to \Att_l $ such that
$\att_l^{env}(q) = \{y_q\}$
\end{itemize}

$\rhs=(\nodes_r,\ports_r, \pn_r, \pp_r, \Att_r, \att_r)$ with
\begin{itemize}
\item
$\nodes_r=\nodes_r^{env}=\{i\}$,
\item
$\ports_r=\emptyset$, $\pn_r=\emptyset$, $\pp_r=\emptyset$.
%\item $\pn_r=\emptyset$,
%\item $\pp_r=\emptyset$.
\item $\Att_r = \Att_l$
\item Moreover, on nodes,  $\att_r=  \att_r^{new}$  ($\att_r^{env}$ being
  empty) with
$\att_r^{new}: \{i\} \to Att_r $ and $ \att_r^{new}(i)= \{ ((y_a+y_b+y_c+y_d+y_e+y_f+y_g+y_h)=?=3) +((x_i=?=1) \times  (( y_a+y_b+y_c+ y_d+y_e\\+y_f+y_g+y_h)=?=2))  \}$.

\end{itemize}

\begin{figure}[t] \centering
\includegraphics[scale=0.3]{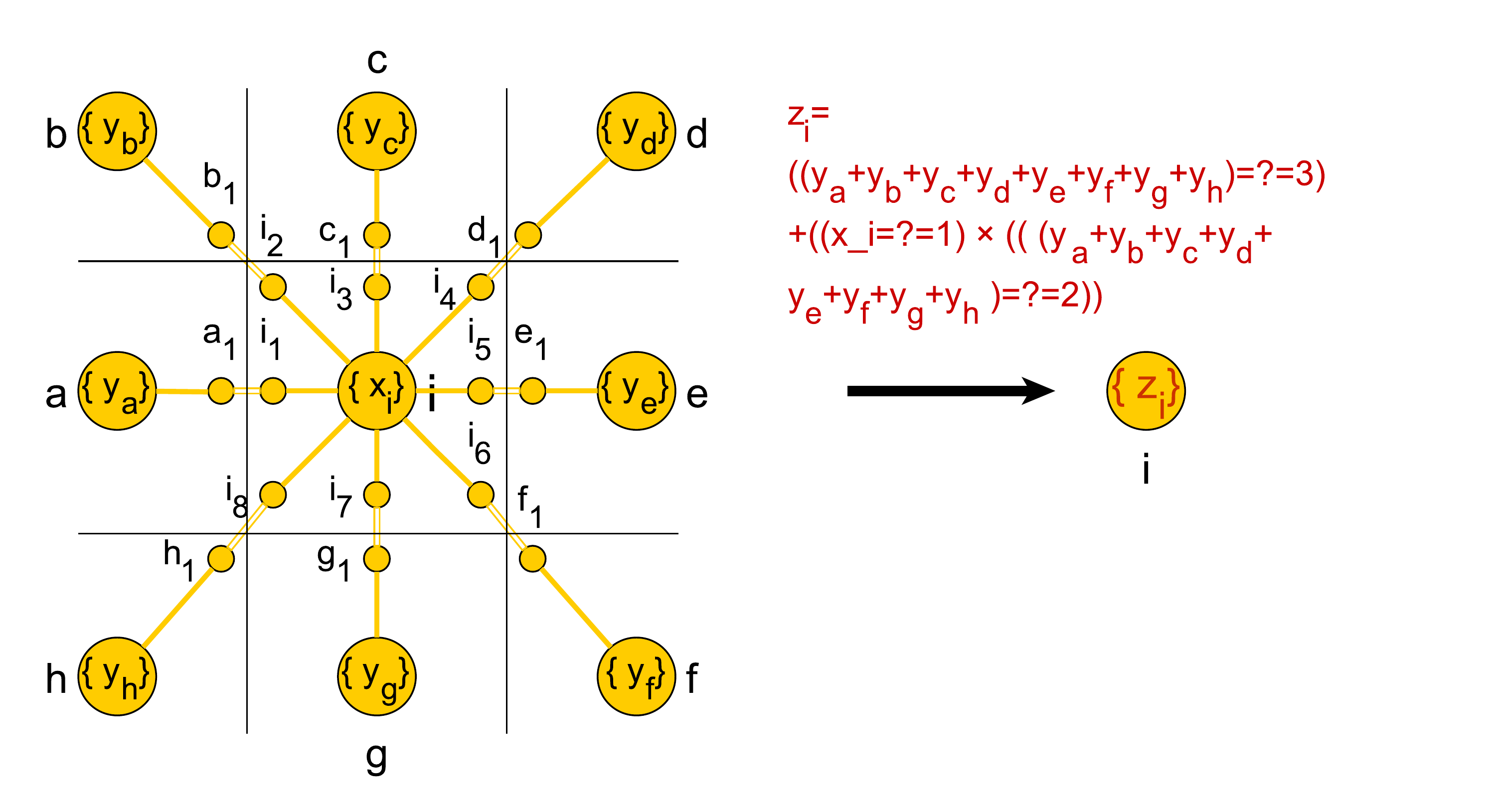}

\caption{game of life  rule }
\label{game}
\end{figure}
In the classical formulation of cellular automata, a cell contains one and only one value.
The model we propose can deal with cells  with one or several values.
For instance, the initial state of the game of life can be a grid containing $\{0\}$'s  except for $4$ cells describing a square (see Figure~\ref{grid1}(a)).
\begin{figure}[t] \centering
\tiny
$\begin{array}{c|c|c|c|c|c}
&&&&&\\
\hline
\, &\{0\}&\{0\}&\{0\}&\{0\}& \,\\
\hline
&\{0\}&\{1\}&\{1\}&\{0\}&\\
\hline
&\{0\}&\{0,1\}&\{1\}&\{0\}&\\
\hline
\, &\{0\}&\{0\}&\{0\}&\{0\}& \,\\
\hline
&&&&&\\
\end{array}
\quad
\begin{array}{c|c|c|c|c|c}
&&&&&\\
\hline
\, &\{0\}&\{0\}&\{0\}&\{0\}& \,\\
\hline
&\{0\}&\{1\}&\{1\}&\{0\}&\\
\hline
&\{0\}&\{1\}&\{1\}&\{0\}&\\
\hline
\, &\{0\}&\{0\}&\{0\}&\{0\}& \,\\
\hline
&&&&&\\
\end{array}
$

(a) \hspace{4cm} (b)
\caption{(a) initial grid; (b) fixed point }
\label{grid1}
\end{figure}

In this configuration one cell have 2 values which means, on the example, that the cell is dead or alive or we don't have any information on the state of the cell.
The behavior of all possible trajectories is computed in parallel and 
the fixed point is reached.
The initial state Figure~\ref{grid2}(a) yields Figure~\ref{grid2}(b) as a fixed point. Here we observe that the indeterminacy concerns at most 4 cells
over time.

\begin{figure}[t]
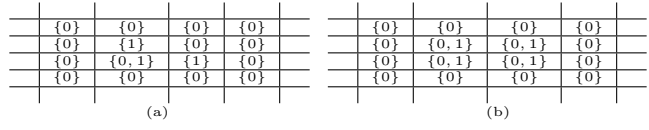
 \centering
\tiny

$\begin{array}{c|c|c|c|c|c}
&&&&&\\
\hline
\, &\{0\}&\{0\}&\{0\}&\{0\}& \,\\
\hline
&\{0\}&\{1\}&\{0\}&\{0\}&\\
\hline
&\{0\}&\{0,1\}&\{1\}&\{0\}&\\
\hline
\, &\{0\}&\{0\}&\{0\}&\{0\}& \,\\
\hline
&&&&&\\
\end{array}
\quad \begin{array}{c|c|c|c|c|c}
&&&&&\\
\hline
\, &\{0\}&\{0\}&\{0\}&\{0\}& \,\\
\hline
&\{0\}&\{0,1\}&\{0,1\}&\{0\}&\\
\hline
&\{0\}&\{0,1\}&\{0,1\}&\{0\}&\\
\hline
\, &\{0\}&\{0\}&\{0\}&\{0\}& \, \\
\hline
&&&&&\\
\end{array}
$

(a) \hspace{4cm} (b)
\caption{(a) initial grid; (b) fixed point }
\label{grid2}
\end{figure}

 }

%%%%%%%%%%%%%%%%%%%%%%%%%%%%%%%%%%%%%%%%%%%%%%%%%%%%%%%%%%%%%%%%%%%%%%%%%%%%%%

%%%%%%%%%%%%%%%%%%%%%%%%%%%%%%%%%%%%%%%%%%%%%%%%%%%%%%%%%%%%%%%%%%%%%%%%%%%%%%

\subsection{The Koch snowflake}
\label{Koch1D}
The well-known Koch snowflake is based on segment divisions (variants exist on surfaces,  both can be modeled by our formalism).
Each segment is recursively  divided into three segments of equal length as described in the following picture :

\includegraphics[scale=0.28]{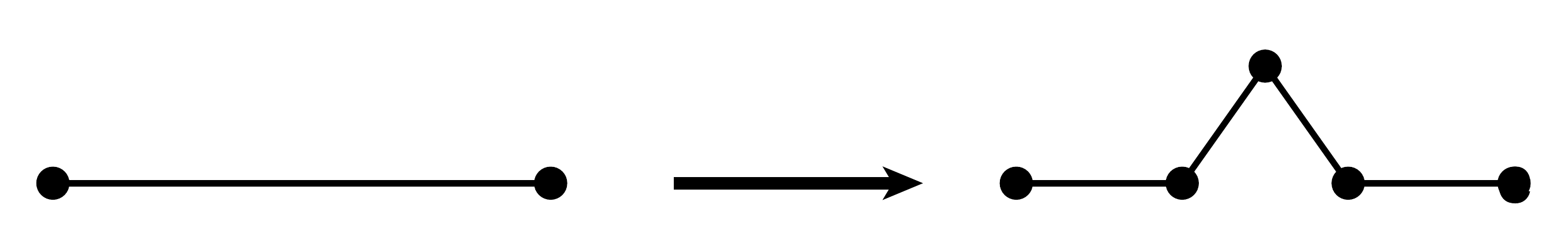}

Let us consider the following  triangle $g$ as an initial state.

 $g=( {\N}_g, \ports_g, \pn_g, \pp_g, \Att_g, \att_g)$ with
$\nodes_g=\{1,2,3 \}$ , $\ports_g=\{p_1,q_1,p_2,q_2,p_3,q_3 \}$, $\pn_g=\{ (p_1,1), (q_1,1),(p_2,2),(q_2,2),(p_3,3),(q_3,3) \}$ , $\pp_g=\{(p_1,q_2),(p_2, q_3),(p_3,q_1) \}$.

 $  \att_g(1)=\left ( \begin{array}{c} -1\\0 \end{array} \right )$, $  \att_g(2)=\left ( \begin{array}{c} 0\\ \sqrt 2\end{array} \right )$, $  \att_g(3)=\left ( \begin{array}{c} 1\\ 0\end{array} \right )$, $ \att_g(p_1) =\att_g(p_2) =\att_g(p_3) =\{-\}$, $\att_g(q_1) =\att_g(q_2) =\att_g(q_3) =\{+\}$.
 
 The attributes of ports are distinguishing attributes. The attributes
 of nodes are the $\mathbb{R}^2$ positions of the
 nodes. Every node got one attribute in $\mathbb{R}^2$, thus by abuse of
 notation, we get rid of the set notation of attributes and use a
 functional one. The implementation of both relations  $\fpr$ and $\ar$ using the rule depicted in Figure~\ref{flocon} provide the expected pictures of flakes as in Figures~\ref{flocon2}.
\begin{figure}[t] \centering
\includegraphics[scale=0.25]{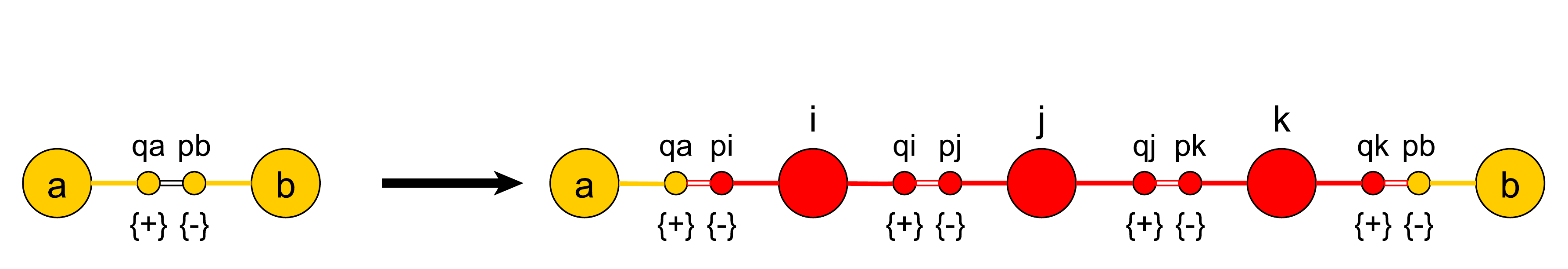}

\caption{Koch Snowflake rule $l  \to r$ with the node attribute computation  
$ \att_r(i)=\frac{2}{3}  \att_l(a)+\frac{1}{3}  \att_l(b)$,
 $ \att_r(j)=\frac{1}{2} ( \att_l(a)+  \att_l(b))+\frac{\sqrt 3}{6}(-  \att_l(a)^{T}+ \att_l(b)^T)$ ,
 $ \att_r(k)=\frac{1}{3}  \att_l(a)+\frac{2}{3}  \att_l(b)$}
\label{flocon}
\end{figure}

Let us denote
$  \att_l(a)=\left ( \begin{array}{c} x_a\\ y_a\end{array} \right )$
and
$  \att_l(b)=\left ( \begin{array}{c} x_b\\ y_b \end{array} \right
)$.
In this example, the attributes of nodes $i, j$ and
$k$ are defined as follows:
$  \att_r(i)=\frac{2}{3}  \att_l(a)+\frac{1}{3}$ \;\;\;\;\; $\att_l(b)= \left
  ( \begin{array}{c} \frac{2}{3} x_a+\frac{1}{3} x_b \\ \frac{2}{3}
    y_a+\frac{1}{3} y_b\end{array} \right ) $,

$ \att_r(j)= \frac{1}{2} ( \att_l(a)+ \att_l(b))+\frac{\sqrt
  3}{6}( \att_l(a)^{T}+ \att_l(b)^T)=\left ( \begin{array}{c}
    \frac{1}{2} (x_a+x_b)+\frac{\sqrt
                                                 3}{6} (y_a-y_b ) \\
                                                 \frac{1}{2}
                                                 (y_a+y_b)+\frac{\sqrt
                                                 3}{6} (-x_a+x_b
                                                 )\end{array} \right )
 $, and $  \att_r(k)=\frac{1}{3}  \att_l(a)+\frac{2}{3}  \att_l(b)= \left ( \begin{array}{c} \frac{1}{3} x_a+\frac{2}{3} x_b \\ \frac{1}{3} y_a+\frac{2}{3} y_b\end{array} \right ) $

\begin{figure}[t] \centering
\hspace{-0.5cm}
\includegraphics[scale=0.3]{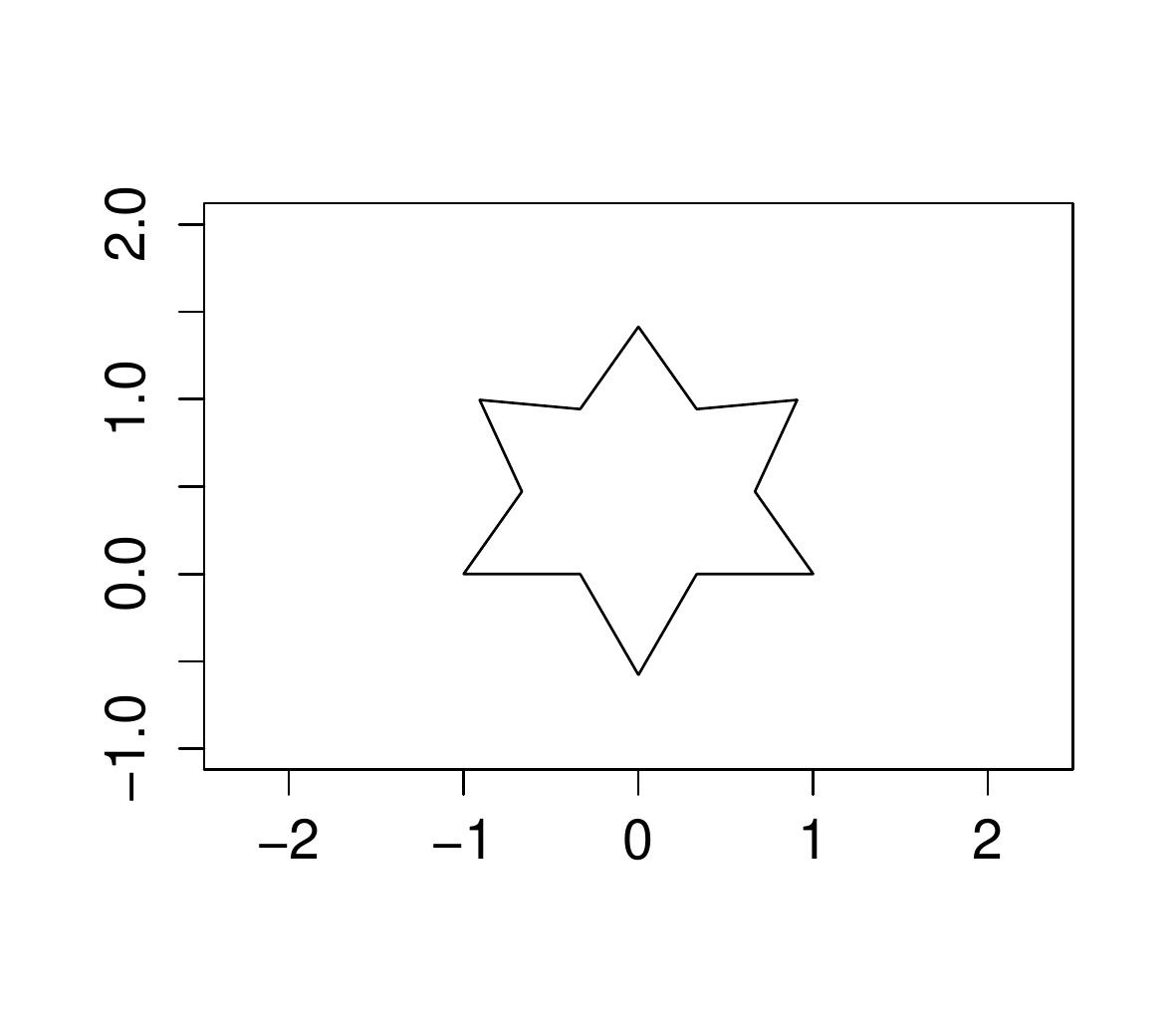} \hspace{-0.3cm}
\includegraphics[scale=0.3]{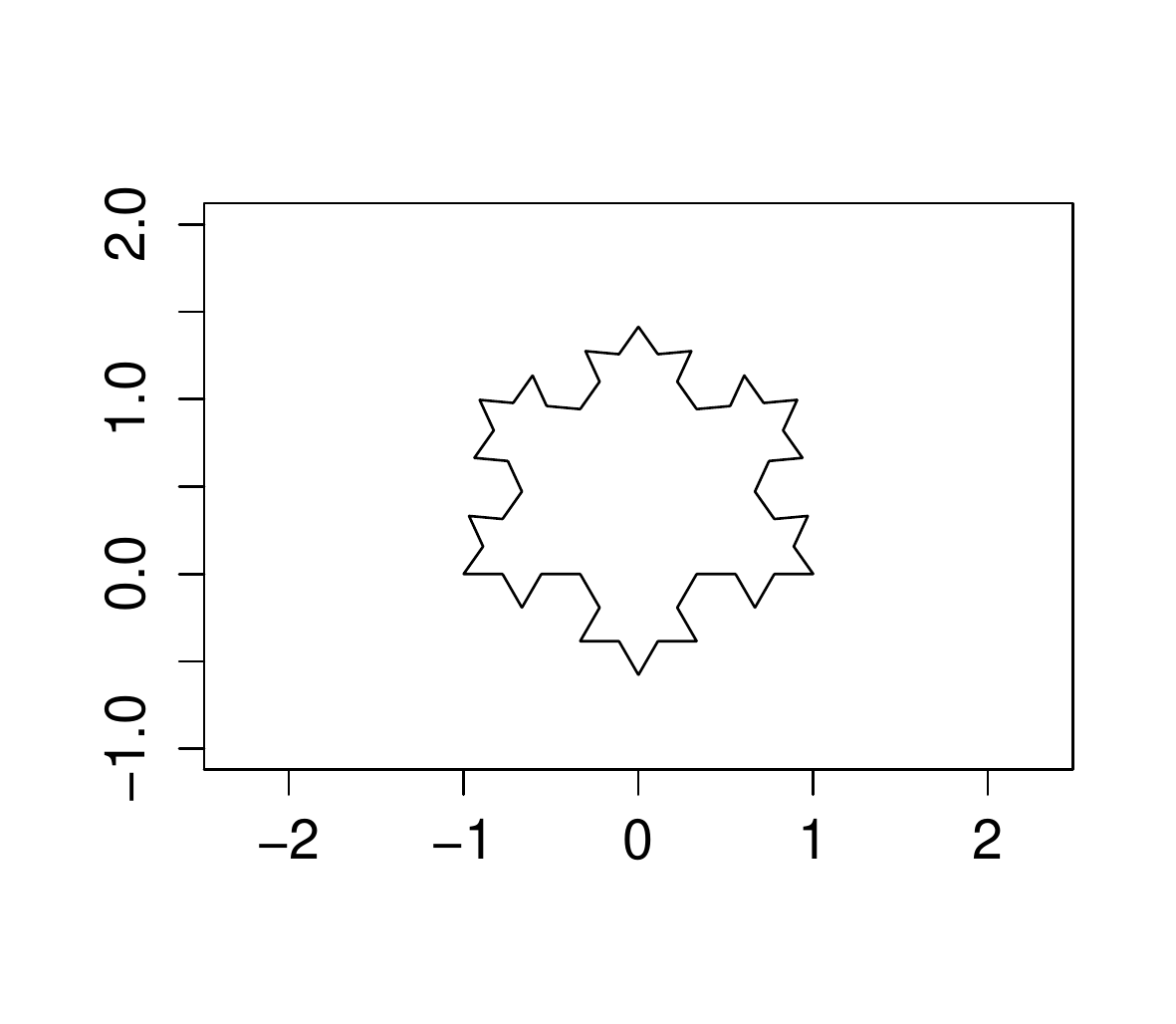} \hspace{-0.3cm}
\includegraphics[scale=0.3]{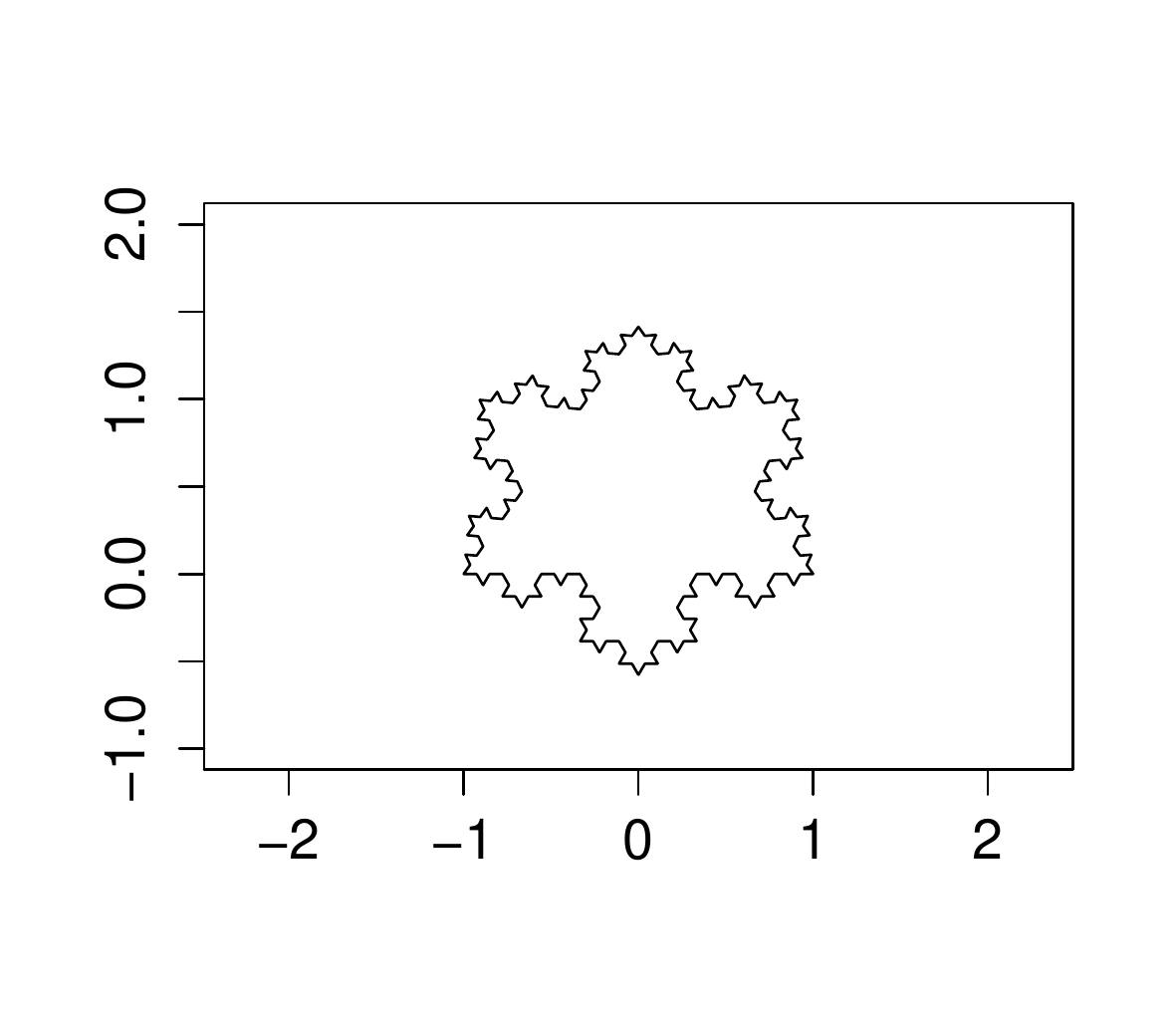} \hspace{-0.3cm}
\includegraphics[scale=0.3]{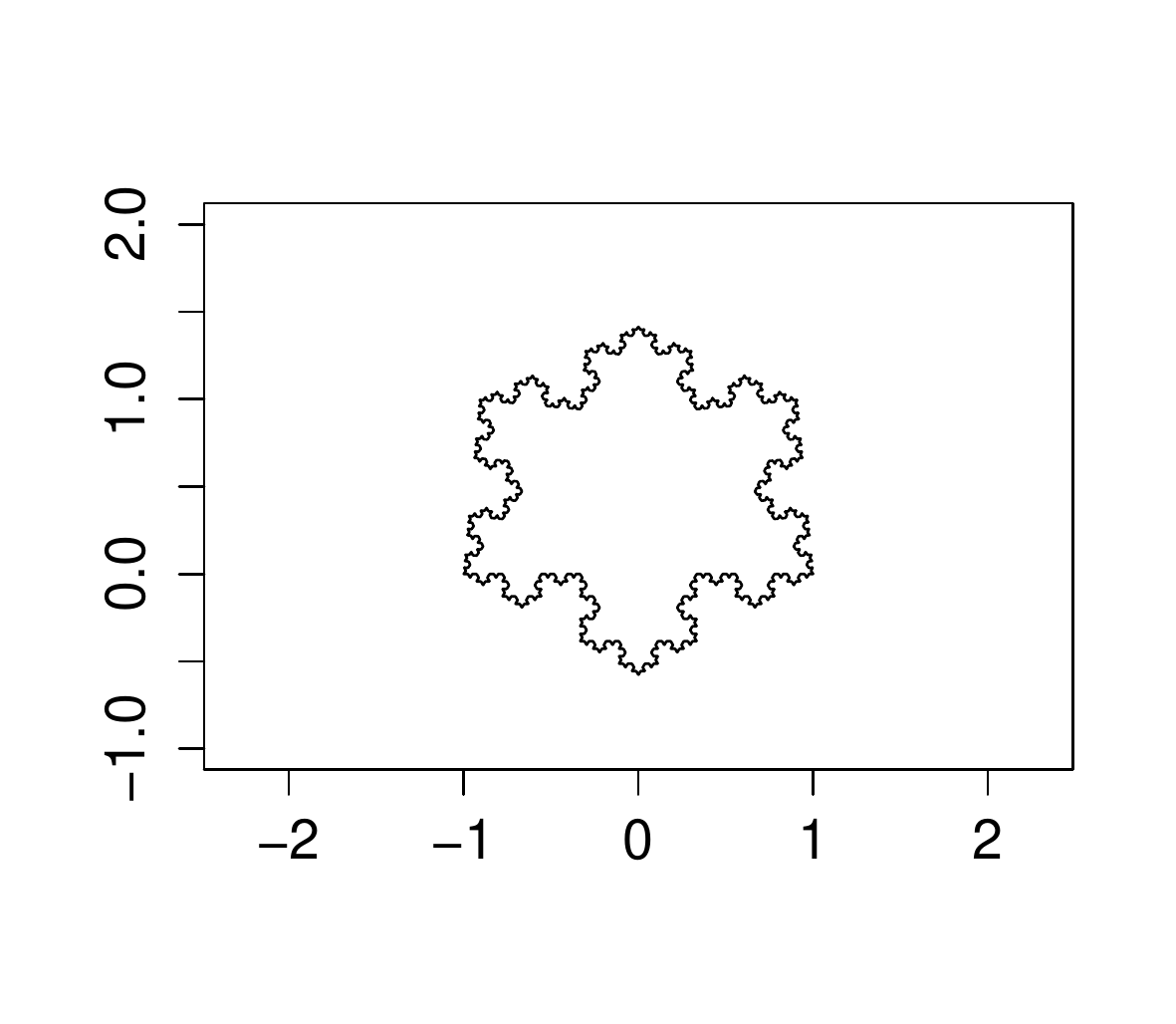}

\caption{Flake results : flake at the different time steps 1,2,3 and 4}
\label{flocon2}
\end{figure}

\subsection{Mesh refinement}

Mesh refinement  consists in creating iteratively new partitions of the considered space.
The initial mesh $g$ we consider is depicted
Figure~\ref{result}. Distinguishing attributes are given on
ports. Attributes on nodes are omitted but we can easily consider
coordinates.
% as in the Koch snowflake example.
Triangle refinements are given in Figure~\ref{refinement}. The three
rules verify the \syc\ and we apply the $\ar$ relation on  $g$ to
obtain the graph $g'$ described in Figure~\ref{result}.
Iteratively, the rewrite system can be applied again on $g'$ and so forth.

\begin{figure}[t] \centering
\hspace*{-0.6cm} \includegraphics[scale=0.2]{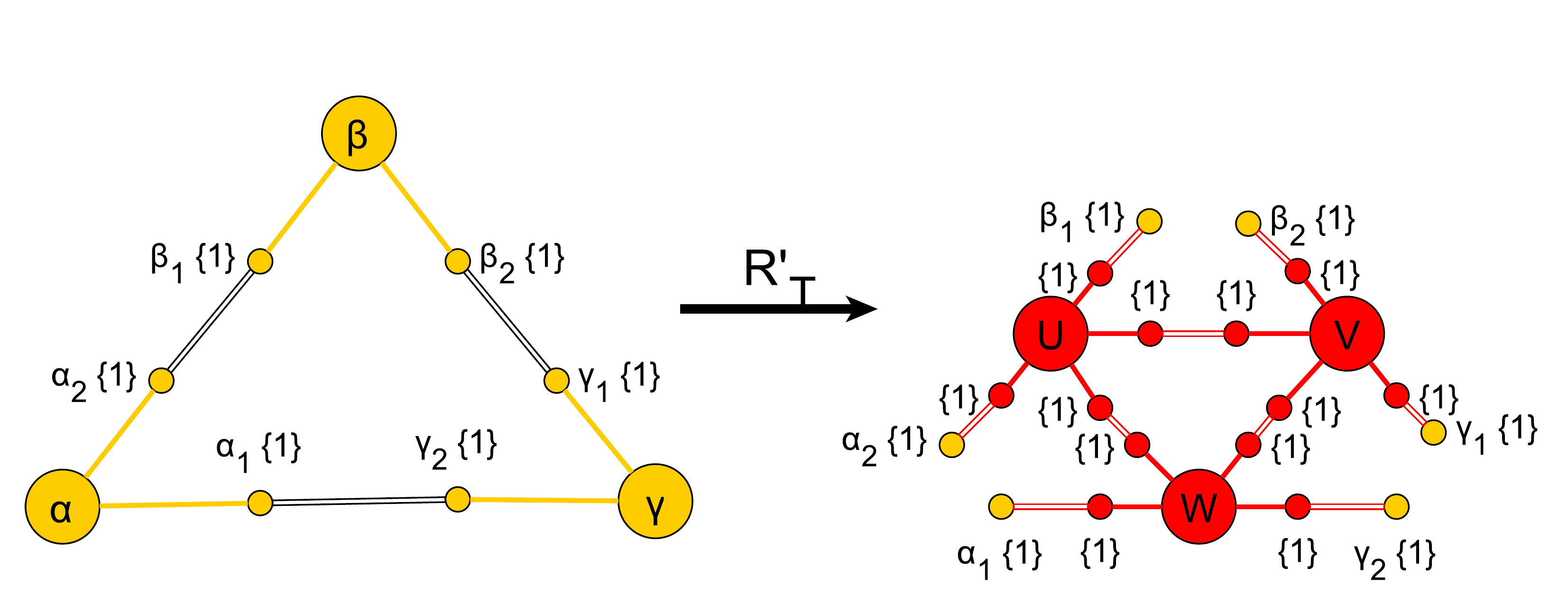}
\includegraphics[scale=0.2]{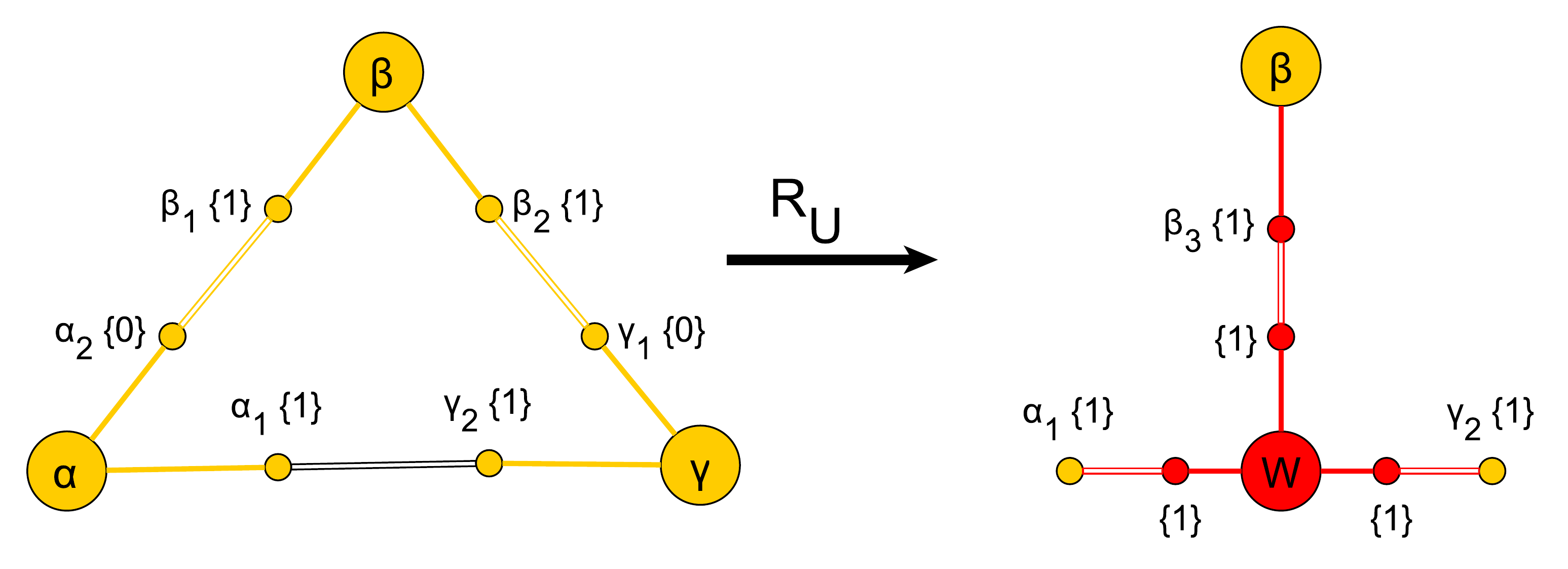}
\includegraphics[scale=0.2]{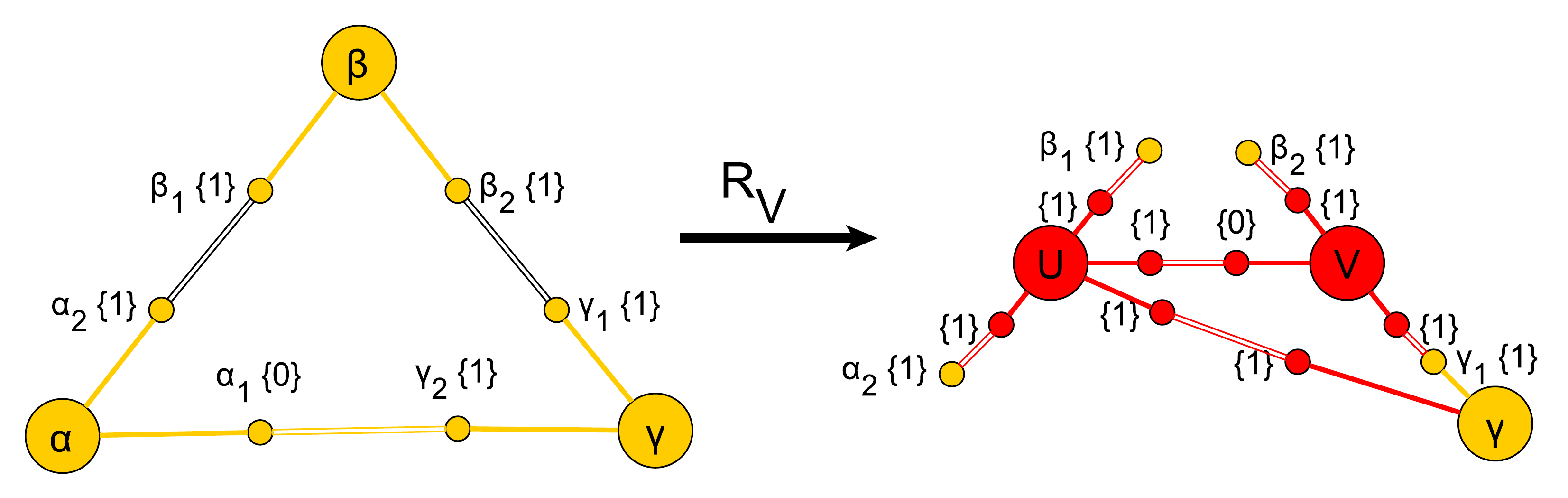}
\caption{The rules $R'_T$, $R_U$, $R_V$ are refinement rules defined
  e.g. in \cite{Bank83}}
\label{refinement}
\end{figure}

\begin{figure}[t] \centering
\includegraphics[width=90mm,height=58mm]{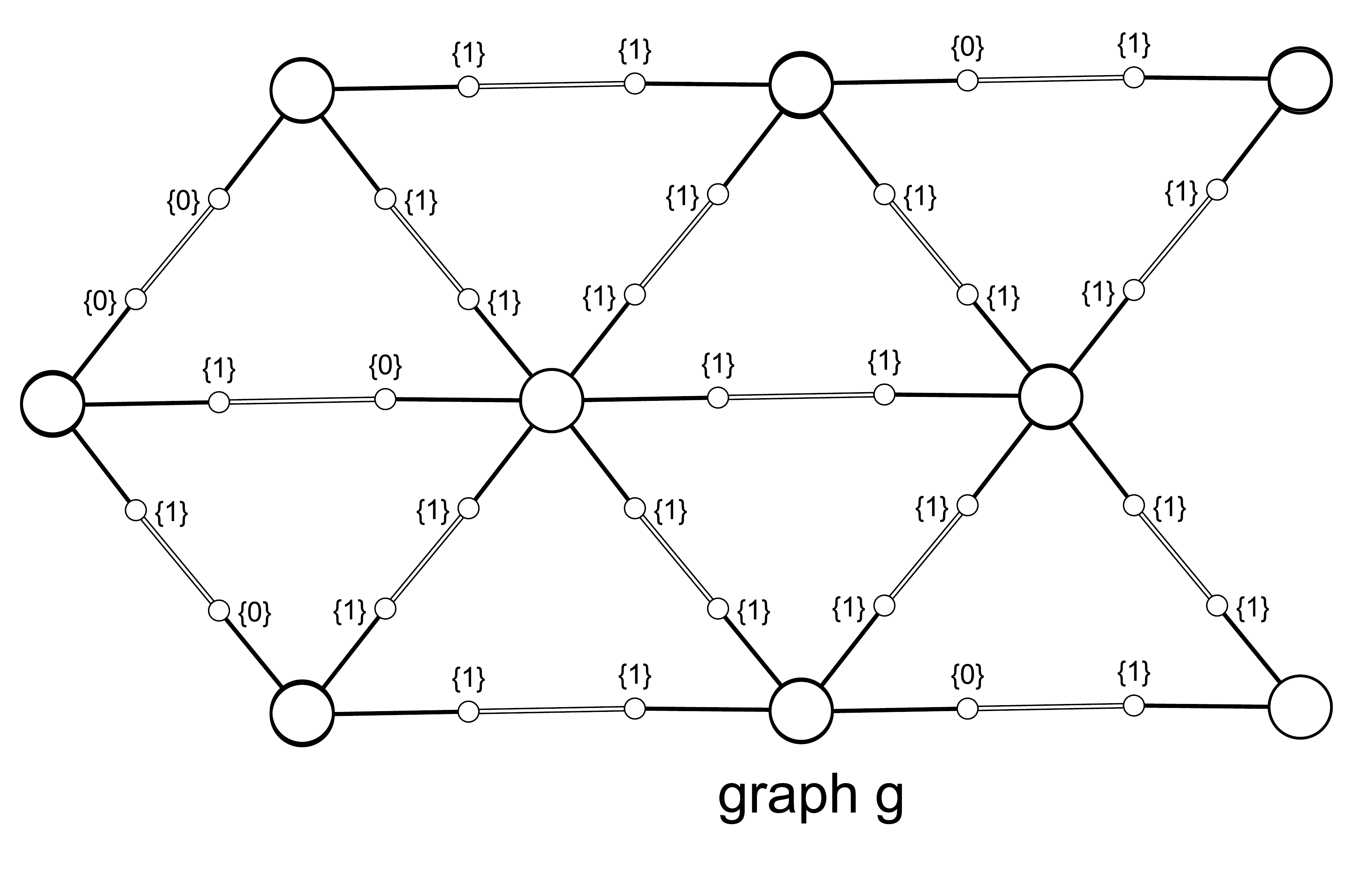}
\includegraphics[width=90mm,height=58mm]{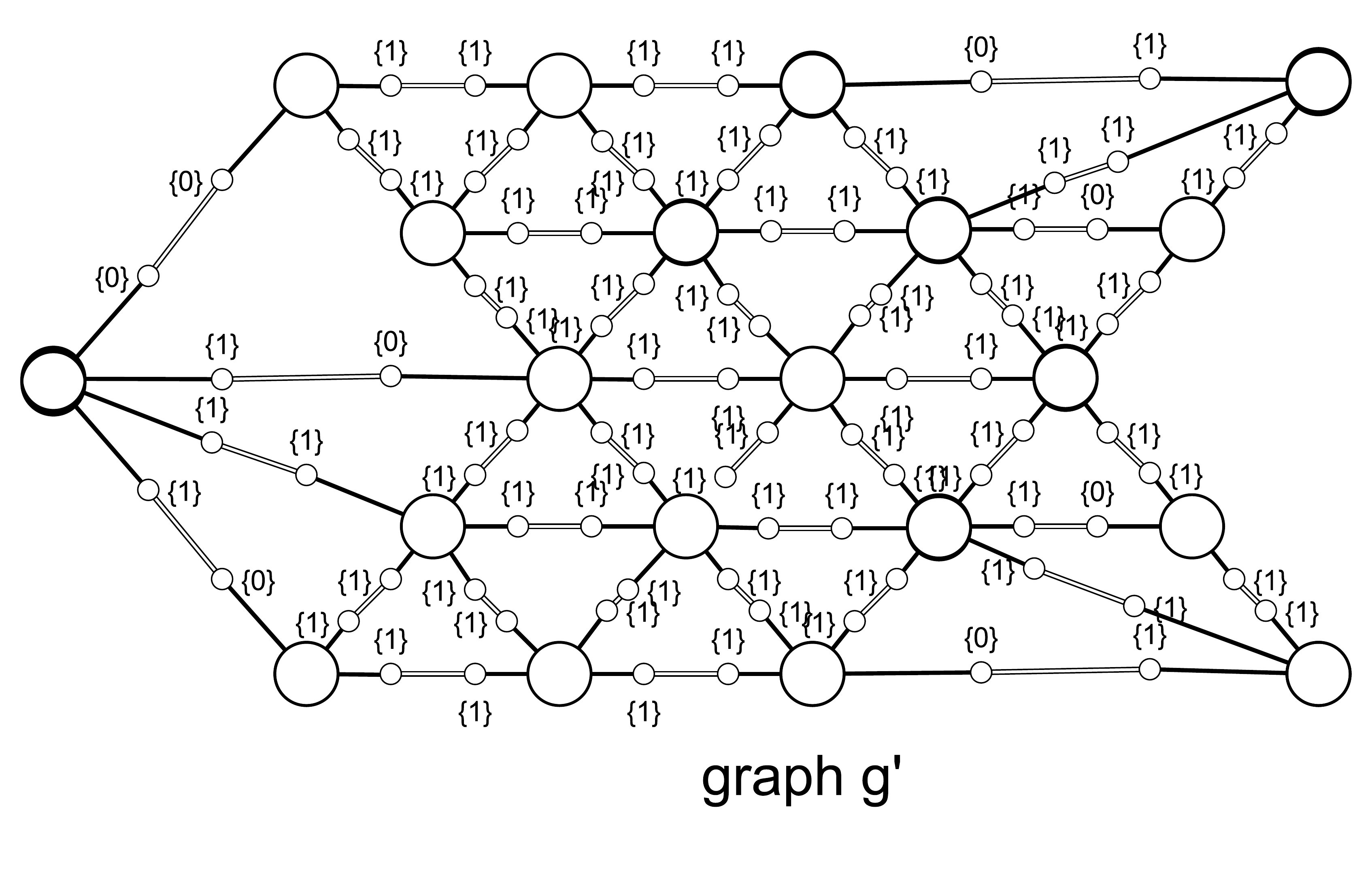}
\caption{$g \ar g'$}
\label{result}
\end{figure}

%%%%%%%%%%%%%%%%%%%%%%%%%%%%%%%%%%%%%%%%%%%%%%%%%%%%%%%%%%%%%%%%%%%%%%%%%%%%%%%%%%%%%
\section{Conclusion and Related Work}
\label{sect:6}

Parallel rewriting technique is a tough issue when it has to deal with
overlapping reducible expressions. In this paper, we have proposed a
framework, based on the notion of rewriting modulo, to deal with graph transformation where parallel reductions may
overlap some parts of the transformed graph. In general, these
transformations do no lead to graphs but to a structure we call
pregraphs. We proposed sufficient conditions which ensure that graphs are
closed under parallel transformations. We also defined two parallel
transformations: (i) one that fires all possible rules in parallel
(full parallel) and (ii) a second rewrite relation which takes
advantage of the possible symmetries that may occur in the rules by
reducing the possible number of matches that one has to consider.  The
two proposed parallel rewrite relations are confluent (up to
isomorphisms).

Our proposal subsumes some existing formalisms where simultaneous
transformations are required such as cellular automata
\cite{Wolfram02} or (extensions of) L-systems
\cite{lindenmayer96}. Indeed, one can easily write graph rewriting
systems which define classical cellular automata, with possibly
evolving structures (grids) and where the content of a cell, say $C$,
may depend on cells not necessary adjacent to $C$. As for L-systems,
they could be seen as formal (context sensitive) grammars which fire
their productions in parallel over a string. Our approach here
generalizes L-systems at least in two directions: first by considering
graphs instead of strings and second by considering overlapping graph
rewrite rules instead of context sensitive (or often context free)
rewrite rules. Some graph transformation approaches could also be
considered as extension of L-systems such as star-grammars
\cite{Nagl79} or hyperedge replacement \cite{Habel92}. These
approaches do not consider overlapping matches but act as context free
grammars. However, in \cite{EK76} parallel graph grammars with
overlapping matches have been considered. In that work, overlapping
subgraphs remain unchanged after reductions, contrary to our framework
which does not require such restrictions. The idea behind parallel
graph grammars has been lifted to general replacement systems in
\cite{Taen96}.  Amalgamation, see e.g.\cite{Lowe93}, aims at
investigating how the parallel application of two rules can be
decomposed into a common part followed by the remainder of the two
considered parallel rules. Amalgamation does not consider full
parallel rewriting as investigated in this paper.
Another approach based on complex transformation has been introduced
in \cite{MS15}. This approach can handle overlapping matches but
requires from the user to specify the transformation of these common
parts. This requires to provide detailed rules. For instance, the two
first cases of the triangle mesh refinement example requires about
sixteen rules including local transformations and inclusions, instead
of two rules in our framework.

The strength of our approach lies in using an equivalence relation on
the resulting pregraph. This equivalence plays an important role in
making graphs closed under rewriting. Other relations may also be
candidate to equate pregraphs into graphs. we plan to investigate such
kind of relations in order to widen the class of rewrite systems that
may be applied in parallel on graph structures in presence of
overlaps. We also plan to investigate other issues such as stochastic
rewriting and conditional rewriting which would be a plus in modeling some
natural phenomena. Analysis of the proposed systems remains to be
investigated further.

% \bibliographystyle{plain}
% \bibliography{biblio}

\begin{thebibliography}{10}

\bibitem{baader98}
Franz Baader and Tobias Nipkow.
\newblock {\em Term rewriting and all that}.
\newblock Cambridge University Press, 1998.

\bibitem{Bank83}
R.~E. Bank, A.~H. Sherman, and A.~Weiser.
\newblock Refinement algorithms and data structures for regular local mesh
  refinement.
\newblock In R.~Stepleman et~al., editor, {\em Scientific Computing}, pages
  3--17. IMACS/North-Holland, 1983.

\bibitem{Ber}
M.~J. {Berger} and P.~{Colella}.
\newblock {Local adaptive mesh refinement for shock hydrodynamics}.
\newblock {\em Journal of Computational Physics}, 82:64--84, May 1989.

\bibitem{BO93}
Ronald~V. Book and Friedrich Otto.
\newblock {\em String-Rewriting Systems}.
\newblock Texts and Monographs in Computer Science. Springer, 1993.

\bibitem{Cou}
H.~L. De~Cougny and M.~S. Shephard.
\newblock Parallel refinement and coarsening of tetrahedral meshes.
\newblock {\em International Journal for Numerical Methods in Engineering},
  46(7):1101--1125, 1999.

\bibitem{Mai}
St{\'e}phane Despr{\'e}aux, Roland Hildebrand, and Aude Maignan.
\newblock {Graph algorithm for the simulation of the interaction between
  particles}.
\newblock In {\em {ICNAAM 2012 - International Conference of Numerical Analysis
  and Applied Mathematics}}, volume 1479 of {\em AIP Conference Proceedings},
  pages 678--681, Kos, Greece, September 2012. {AIP}.

\bibitem{DuvalEPR14}
Dominique Duval, Rachid Echahed, Fr{\'{e}}d{\'{e}}ric Prost, and Leila Ribeiro.
\newblock Transformation of attributed structures with cloning.
\newblock In Stefania Gnesi and Arend Rensink, editors, {\em Fundamental
  Approaches to Software Engineering, {FASE} 2014}, volume 8411 of {\em LNCS},
  pages 310--324. Springer, 2014.

\bibitem{EK76}
Hartmut Ehrig and Hans{-}J{\"{o}}rg Kreowski.
\newblock Parallel graph grammars.
\newblock In A.~Lindenmayer and G.~Rozenberg, editors, {\em Automata,
  Languages, Development}, pages 425--447. Amsterdam: North Holland, 1976.

\bibitem{Moore69}
Moore G.A.
\newblock Automatic scanning and computer processes for the quantitative
  analysis of micrographs and equivalent subjects.
\newblock {\em Pictorial Pattern Recognition}, 1969.

\bibitem{Habel92}
Annegret Habel.
\newblock {\em Hyperedge Replacement: Grammars and Languages}, volume 643 of
  {\em Lecture Notes in Computer Science}.
\newblock Springer, 1992.

\bibitem{KleijnR80}
H.~C.~M. Kleijn and Grzegorz Rozenberg.
\newblock A study in parallel rewriting systems.
\newblock {\em Information and Control}, 44(2):134--163, 1980.

\bibitem{Kle}
Richard~I. Klein.
\newblock Star formation with 3-d adaptive mesh refinement: the collapse and
  fragmentation of molecular clouds.
\newblock {\em Journal of Computational and Applied Mathematics},
  109(1–2):123 -- 152, 1999.

\bibitem{Lowe93}
Michael L{\"{o}}we.
\newblock Algebraic approach to single-pushout graph transformation.
\newblock {\em Theor. Comput. Sci.}, 109(1{\&}2):181--224, 1993.

\bibitem{MS15}
Luidnel Maignan and Antoine Spicher.
\newblock Global graph transformations.
\newblock In Detlef Plump, editor, {\em Proceedings of the 6th International
  Workshop on Graph Computation Models}, volume 1403, pages 34--49.
  CEUR-WS.org, 2015.

\bibitem{Nagl79}
Manfred Nagl.
\newblock {\em Graph-Grammatiken: Theorie, Anwendungen, Implementierung}.
\newblock Vieweg, 1979.

\bibitem{PetersonS81}
Gerald~E. Peterson and Mark~E. Stickel.
\newblock Complete sets of reductions for some equational theories.
\newblock {\em J. {ACM}}, 28(2):233--264, 1981.

\bibitem{Ple}
Tomasz Plewa, Timur Linde, and V.~Gregory Weir, editors.
\newblock {\em Adaptive Mesh Refinement - Theory and Applications. Proceedings
  of the Chicago Workshop on Adaptive Mesh Refinement Methods, Sept. 3–5},
  volume~41.
\newblock Lecture Notes in Computational Science and Engineering, Springer,
  2003.

\bibitem{Plu}
Pat Plunkett, Brian~A. Camley, Kimberly~L. Weirich, Jacob Israelachvili, and
  Paul~J. Atzberger.
\newblock Simulation of edge facilitated adsorption and critical concentration
  induced rupture of vesicles at a surface.
\newblock {\em Soft Matter}, 9:8420--8427, 2013.

\bibitem{lindenmayer96}
Przemyslaw Prusinkiewicz and Aristid Lindenmayer.
\newblock {\em The algorithmic beauty of plants}.
\newblock Springer, 1996.

\bibitem{handbook1}
Grzegorz Rozenberg, editor.
\newblock {\em Handbook of Graph Grammars and Computing by Graph
  Transformations, Volume 1: Foundations}. World Scientific, 1997.

\bibitem{ST2012}
Donald Sannella and Andrzej Tarlecki.
\newblock {\em Foundations of Algebraic Specification and Formal Software
  Development.}
\newblock EATCS Monographs on theoretical computer science. Springer, 2012.

\bibitem{Taen96}
Gabriele Taentzer.
\newblock {\em Parallel and distributed graph transformation - formal
  description and application to communication-based systems}.
\newblock Berichte aus der Informatik. Shaker, 1996.

\bibitem{Wolfram02}
Stephen Wolfram.
\newblock {\em A new kind of science}.
\newblock Wolfram-Media, 2002.

\end{thebibliography}

\end{document}